\documentclass[a4paper,UKenglish,cleveref, autoref, thm-restate]{lipics-v2021}

\usepackage[T1]{fontenc}
\usepackage[utf8]{inputenc}
\usepackage{thmtools}
\usepackage{thm-restate}
\usepackage{soul}
\usepackage{color}
\usepackage{xcolor}
\usepackage{ifthen}
\usepackage{xspace}	
\usepackage{thm-restate}
\usepackage{amssymb,amsmath}

\usepackage[makeroom]{cancel}
\usepackage{tikz}	
\usetikzlibrary{arrows, decorations.markings}
\usetikzlibrary{decorations.pathmorphing, patterns,shapes}
\usetikzlibrary{arrows}
\usepackage{multirow}
\makeatletter
\newcommand*{\rom}[1]{\expandafter\@slowromancap\romannumeral #1@}
\makeatother
\usepackage[xcolor,makeidx=compatibility,quotation,notion]{knowledge}

\usepackage{subcaption}
\usepackage{relsize}

\usepackage{MnSymbol}

\newrobustcmd \textiff{\text{if and only if}}
\newrobustcmd \qtextiff{\quad\textiff\quad}
\newrobustcmd \qqtextiff{\qquad\textiff\qquad}

\newrobustcmd \suchthat{\text{such that }}
\newrobustcmd \qsuchthat{\quad\suchthat}
\newrobustcmd \qqsuchthat{\qquad\suchthat}

\NewDocumentCommand\TODO{ o }{%
			\marginpar{\textcolor{red}{{\LARGE !}}}%
			\IfNoValueTF{#1}{}{\textcolor{red}{#1}}{}}

\definecolor{srcolor}{RGB}{255,20,147}

\newrobustcmd{\defs}{::=}
\newrobustcmd{\ie}{i.e. }

\knowledgenewrobustcmd\mset[2]
	{\ensuremath{\cmdkl{\{}#1\cmdkl\,\mid\,#2\cmdkl{\}}}}

\knowledgenewrobustcmd\openset[2]{\ensuremath{\cmdkl{(}#1\cmdkl,#2\cmdkl{)}}}

\knowledgenewrobustcmd{\leftopenset}[2]{\ensuremath{\cmdkl(#1,#2\cmdkl]}}
\knowledgenewrobustcmd{\rightopenset}[2]{\ensuremath{\cmdkl[#1,#2\cmdkl)}}
\knowledgenewrobustcmd{\closedset}[2]{\ensuremath{\cmdkl[#1,#2\cmdkl]}}
\knowledgenewrobustcmd{\openprefix}[1]{\ensuremath{\cmdkl({-}\infty,#1\cmdkl)}}
\knowledgenewrobustcmd{\closedprefix}[1]{\ensuremath{\cmdkl({-}\infty,#1\cmdkl]}}
\knowledgenewrobustcmd{\opensuffix}[1]{\ensuremath{\cmdkl(#1,\infty\cmdkl)}}
\knowledgenewrobustcmd{\closedsuffix}[1]{\ensuremath{\cmdkl[#1,\infty\cmdkl)}}
\knowledgenewrobustcmd{\leftclosedset}[2]{\ensuremath{\cmdkl[#1,#2 \rangle}}

\knowledgenewrobustcmd\relto[2]{#2\cmdkl{~\mathrm{in}~}#1}
\newrobustcmd{\rightrelativize}[2]{#1[\leq #2]}
\newrobustcmd{\relativize}[3]{#1[\geq #2,\leq #3]}
\newrobustcmd{\substitute}[2]{[#1/ #2]}

\newrobustcmd{\tuple}{\vec}
\newrobustcmd{\ifthen}{\Rightarrow}
\newrobustcmd{\lang}[1]{\mathcal{L}(#1)}

\knowledgenewrobustcmd{\dom}[1]{\cmdkl{\mathit{dom}(#1)}}
\knowledgenewrobustcmd{\length}[1]{\cmdkl|#1\cmdkl|}

\newrobustcmd\rationals{\mathbb{Q}}
\newrobustcmd\reals{\mathbb{R}}
\newrobustcmd\nats{\mathbb{N}}

\newrobustcmd{\Bin}{\kl[\Bin]{{\in}}}
\knowledge\Bin{notion}
\newrobustcmd{\Bnin}{\kl[\Bnin]{{\not\in}}}
\knowledge\Bnin{notion}
\newrobustcmd{\Win}{\mathbin{\kl[\Win]{\in}}}
\knowledge\Win{notion}
\newrobustcmd\product{\kl[\product]{\pi}}
\knowledge\product{notion}

\knowledgenewrobustcmd\Words[1]{\cmdkl{\fformula{Words}[}#1\cmdkl{]}}

\newrobustcmd{\sublem}[1]{.#1}
\newrobustcmd{\card}[1]{|#1|}

\newrobustcmd{\languageOf}[1]{\mathrm{L}_{#1}}

\newrobustcmd{\LSing}{\kl[\LSing]{\languageOf{\mathrm{Sing}}}}
\knowledge {\LSing}{notion}
\newrobustcmd{\LFinite}{\kl[\LFinite]{\languageOf{\mathrm{Fin}}}}
\knowledge {\LFinite}{notion}
\newrobustcmd{\LComplete}{\kl[\LComplete]{\languageOf{\mathrm{Complete}}}}
\knowledge {\LComplete}{notion}
\newrobustcmd{\LEvenGaps}{\kl[\LEvenGaps]{\languageOf{\mathrm{EvenGaps}}}}
\knowledge {\LEvenGaps}{notion}
\newrobustcmd{\LOrd}{\kl[\LOrd]{\languageOf{\mathrm{Ord}}}}
\knowledge {\LOrd}{notion}
\newrobustcmd{\LScat}{\kl[\LScat]{\languageOf{\mathrm{Scat}}}}
\knowledge {\LScat}{notion}
\newrobustcmd{\LGap}{\kl[\LGap]{\languageOf{\mathrm{Gap}}}}
\knowledge {\LGap}{notion}
\newrobustcmd{\LEven}{\kl[\LEven]{\languageOf{\mathrm{Even}}}}
\knowledge {\LEven}{notion}
\newrobustcmd{\LPerfectlyDense}{\kl[\LPerfectlyDense]{\languageOf{\mathrm{PerfDense}}}}
\knowledge {\LPerfectlyDense}{notion}
\newrobustcmd{\LFiniteGaps}{\kl[\LFiniteGaps]{\languageOf{\mathrm{FiniteGaps}}}}
\knowledge {\LFiniteGaps}{notion}

\newrobustcmd{\monoid}{\ensuremath{\mathbf{M}}}
\newrobustcmd{\monoidset}{\ensuremath{M}}
\newrobustcmd{\monoidN}{\mathbf{N}}
\newrobustcmd{\monoidNset}{N}
\newrobustcmd{\monoidOf}[2][]{\ensuremath{(#2#1,{\product})}\xspace}
\newrobustcmd{\algebraOf}[1]{\ensuremath{(#1,\unit,\productoper,\algomegasymb, \algomegastarsymb,\algshufflesymb)}\xspace}
\newrobustcmd{\semigroup}{\mathbf{S}}
\newrobustcmd{\semigroupset}{S}
\newrobustcmd{\semigroupOf}[1]{\ensuremath{(#1,\product)}\xspace}

\newrobustcmd{\monoidSing}{\kl[\monoidSing]{\ensuremath{\mathbf{Sing}}}}
\knowledge {\monoidSing}{notion}
\newrobustcmd{\monoidOrd}{\kl[\monoidOrd]{\ensuremath{\mathbf{Ord}}}}
\knowledge {\monoidOrd}{notion}
\newrobustcmd{\monoidFinite}{\kl[\monoidFinite]{\ensuremath{\mathbf{Fin}}}}
\knowledge {\monoidFinite}{notion}
\newrobustcmd{\monoidCut}{\kl[\monoidCut]{\ensuremath{\mathbf{Cut}}}}
\knowledge {\monoidCut}{notion}
\newrobustcmd{\monoidScat}{\kl[\monoidScat]{\ensuremath{\mathbf{Scat}}}}
\knowledge {\monoidScat}{notion}
\newrobustcmd{\monoidEven}{\kl[\monoidEven]{\ensuremath{\mathbf{Even}}}}
\knowledge {\monoidEven}{notion}
\newrobustcmd{\monoidGap}{\kl[\monoidGap]{\ensuremath{\mathbf{Gap}}}}
\knowledge {\monoidGap}{notion}
\newrobustcmd{\monoidMin}{\kl[\monoidMin]{\ensuremath{\mathbf{Min}}}}
\knowledge {\monoidMin}{notion}
\newrobustcmd{\monoidEvenGaps}{\kl[\monoidGap]{\ensuremath{\mathbf{EvenGaps}}}}
\knowledge {\monoidEvenGaps}{notion}
\newrobustcmd{\monoidFiniteGaps}{\kl[\monoidGap]{\ensuremath{\mathbf{FinGaps}}}}
\knowledge {\monoidFiniteGaps}{notion}
\knowledgenewrobustcmd{\monoidPerfectlyDense}{\cmdkl{\ensuremath{\mathbf{PerfDense}}}}

\newrobustcmd{\dupmonoidGap}{\ensuremath{\mathbf{Gap}}}

\newrobustcmd{\morphism}{\pi}
\newrobustcmd{\invmorphism}{\pi^{-1}}

\newrobustcmd{\fformula}[1]{\ensuremath{\mathtt{#1}}}
\newrobustcmd{\formula}[1]{\kl[\formula{#1}]{\fformula{#1}}}
\newrobustcmd{\setfamily}[1]{\ensuremath{\mathtt{#1}}}

\knowledgenewrobustcmd\fMin{\cmdkl{\fformula{Min}}}
\knowledgenewrobustcmd\fMax{\cmdkl{\fformula{Max}}}
\knowledgenewrobustcmd\fCut{\cmdkl{\fformula{Cut}}}
\knowledgenewrobustcmd\fDense{\cmdkl{\fformula{Dense}}}
\knowledgenewrobustcmd\fGap{\cmdkl{\fformula{Gap}}}
\knowledgenewrobustcmd\fOrdinal{\cmdkl{\fformula{Ordinal}}}
\knowledgenewrobustcmd\fReverseOrdinal{\cmdkl{\fformula{Ordinal^*}}}
\knowledgenewrobustcmd\fFinite{\cmdkl{\fformula{Finite}}}
\knowledgenewrobustcmd\fEven{\cmdkl{\fformula{Even}}}
\knowledgenewrobustcmd\fOdd{\cmdkl{\fformula{Odd}}}
\knowledgenewrobustcmd\fScattered{\cmdkl{\fformula{Scattered}}}
\knowledgenewrobustcmd\fOneGap{\cmdkl{\fformula{oneGap}}}
\knowledgenewrobustcmd\fNoGap{\cmdkl{\fformula{noGaps}}}
\knowledgenewrobustcmd\fFiniteGaps{\cmdkl{\fformula{FiniteGaps}}}
\knowledgenewrobustcmd\fEvenGaps{\cmdkl{\fformula{EvenGaps}}}
\knowledgenewrobustcmd\fSucc{\cmdkl{\fformula{Succ}}}

\knowledgenewrobustcmd\fDensea{\cmdkl{\fformula{Dense(\mathnormal{a})}}}
\knowledgenewrobustcmd\fOddbDensea{\cmdkl{\fformula{Odd(\mathnormal{b})\text{-}Dense(\mathnormal{a})}}}
\knowledgenewrobustcmd\fWellFounded{\cmdkl{\fformula{WellFounded}}}
\knowledgenewrobustcmd\fNob{\cmdkl{\fformula{No(\mathnormal{b})}}}

\knowledgenewrobustcmd\fAtleastTwo{\cmdkl{\fformula{AtleastTwo}}}
\knowledgenewrobustcmd\fexistab{\cmdkl{\fformula{existFactor}[ab]}}
\knowledgenewrobustcmd\fnoab{\cmdkl{\fformula{noFactor}[ab]}}
\knowledgenewrobustcmd\fomegaab{\cmdkl{\fformula{Omega}\text{-}{ab}}}

\knowledgenewrobustcmd\fonegap{\cmdkl{\fformula{oneGap}}}
\knowledgenewrobustcmd\ffingap{\cmdkl{\fformula{finGap}}}
\knowledgenewrobustcmd\fnogap{\cmdkl{\fformula{noGap}}}
\knowledgenewrobustcmd\fminonegap{\cmdkl{\alphabet \fonegap}}

\knowledgenewrobustcmd\impliciti{\cmdkl{\mathrm{i}}}
\knowledgenewrobustcmd\implicitoi{\cmdkl{\mathrm{oi}}}
\knowledgenewrobustcmd\implicitosi{\cmdkl{\mathrm{o^*i}}}
\knowledgenewrobustcmd\implicitsc{\cmdkl{\mathrm{sc}}}
\knowledgenewrobustcmd\implicitsh{\cmdkl{\mathrm{sh}}}

\knowledgenewrobustcmd\Juright[1]{\cmdkl{\overrightarrow{J}_{#1}}}
\knowledgenewrobustcmd\Juleft[1]{\cmdkl{\overleftarrow{J}_{#1}}}

\knowledgenewrobustcmd\Jright[1]{\cmdkl{\overset{\mapsto}{J}_{#1}}}
\knowledgenewrobustcmd\eright[1]{\cmdkl{\overset{\mapsto}{#1}}}

\knowledgenewrobustcmd\Jleft[1]{\cmdkl{\overset{\leftmapsto}{J}_{#1}}}
\knowledgenewrobustcmd\eleft[1]{\cmdkl{\overset{\leftmapsto}{#1}}}

\knowledgenewrobustcmd\euright[1]{\cmdkl{\overrightarrow{#1}}}
\knowledgenewrobustcmd\euleft[1]{\cmdkl{\overleftarrow{#1}}}

\knowledgenewrobustcmd\fNotScattered{\cmdkl{\fformula{NotScattered}}}
\knowledgenewrobustcmd\fInfinity{\cmdkl{\fformula{Infinitely}}}

\renewrobustcmd{\ostar}{\circledast} 

\newrobustcmd{\circplussymbol}{\ensuremath{\small \oplus}}
\newrobustcmd{\circsemigroup}{\circplussymbol-semigroup\xspace}	

\newrobustcmd\nepowerset[1]{\mathcal{P}(#1)\setminus\{\emptyset\}}

\newrobustcmd{\algomegasymb}{\ensuremath{{\boldsymbol {\kl[\omegaoper]{\omega}}}}}
\newrobustcmd{\algomegastarsymb}{\ensuremath{\boldsymbol{\kl[\omegastaroper]{\omega^*}}}}
\newrobustcmd{\algshufflesymb}{\ensuremath{\boldsymbol{\kl[\shuffleoper]\shuffle}}}

\knowledgenewrobustcmd{\productoper}{\ensuremath{\mathbin{\cmdkl\cdot}}}
\knowledgenewrobustcmd{\omegaoper}[1]{\ensuremath{#1^{\cmdkl\algomegasymb}}}
\knowledgenewrobustcmd{\omegastaroper}[1]{\ensuremath{#1^{\cmdkl\algomegastarsymb}}}
\knowledgenewrobustcmd{\shuffleoper}[1]{\ensuremath{#1^{\cmdkl\algshufflesymb}}}

\newrobustcmd{\shuffle}{\eta}
\knowledgenewrobustcmd{\perfectShuffle}[1]{\cmdkl{\ensuremath{\mathit{perfectshuffle}}}(#1)}

\newrobustcmd{\Nat}{\nats}
\newrobustcmd{\intersect}{\cap}
\newrobustcmd{\leftlimit}{\leftarrowtail}
\newrobustcmd{\rightlimit}{\rightarrowtail}
\newrobustcmd{\recurrentpoint}{|}

\newrobustcmd\gJ{\kl[\gJ]{\mathcal{J}}}
\newrobustcmd\gL{\kl[\gL]{\mathcal{L}}}
\newrobustcmd\gR{\kl[\gR]{\mathcal{R}}}
\newrobustcmd\gH{\kl[\gH]{\mathcal{H}}}
\newrobustcmd\gD{\kl[\gD]{\mathcal{D}}}

\knowledgenewrobustcmd{\Jclass}{\cmdkl{\mathcal{J}}}
\knowledgenewrobustcmd{\Lclass}{\cmdkl{\mathcal{L}}}
\knowledgenewrobustcmd{\Rclass}{\cmdkl{\mathcal{R}}}
\knowledgenewrobustcmd{\Hclass}{\cmdkl{\mathcal{H}}}
\knowledgenewrobustcmd{\Dclass}{\cmdkl{\mathcal{D}}}

\newrobustcmd{\Jeq}{\mathrel{\kl[\Jeq]{\mathcal J}}}
\newrobustcmd{\nJeq}{\mathrel{\kl[\Jeq]{\cancel{\mathcal{J}}}}}
\knowledge{\Jeq}{notion}

\newrobustcmd{\Req}{\mathrel{\kl[\Req]{\mathcal{R}}}}
\newrobustcmd{\nReq}{\mathrel{\kl[\Req]{\cancel{\mathcal{R}}}}}
\knowledge{\Req}{notion}

\newrobustcmd{\Leq}{\mathrel{\kl[\Leq]{\mathcal{L}}}}
\newrobustcmd{\nLeq}{\mathrel{\kl[\Leq]{\cancel{\mathcal{L}}}}}
\knowledge{\Leq}{notion}

\newrobustcmd{\Heq}{\mathrel{\kl[\Heq]{\mathcal{H}}}}
\newrobustcmd{\nHeq}{\mathrel{\kl[\Heq]{\cancel{\mathcal{H}}}}}
\knowledge{\Heq}{notion}

\newrobustcmd{\Jleq}{\mathrel{\kl[\Jleq]{\leqslant_{\mathcal{J}}}}}
\newrobustcmd{\Jgeq}{\mathrel{\kl[\Jleq]{\geqslant_{\mathcal{J}}}}}
\newrobustcmd{\nJleq}{\mathrel{\kl[\Jleq]{\nleqslant_{\mathcal{J}}}}}
\newrobustcmd{\nJgeq}{\mathrel{\kl[\Jleq]{\ngeqslant_{\mathcal{J}}}}}
\knowledge{\Jleq}{notion}

\newrobustcmd{\Rleq}{\mathrel{\kl[\Rleq]{\leqslant_{\mathcal{R}}}}}
\newrobustcmd{\Rgeq}{\mathrel{\kl[\Rleq]{\geqslant_{\mathcal{R}}}}}
\knowledge{\Rleq}{notion}

\newrobustcmd{\Lleq}{\mathrel{\kl[\Lleq]{\leqslant_{\mathcal{L}}}}}
\newrobustcmd{\Lgeq}{\mathrel{\kl[\Lleq]{\geqslant_{\mathcal{L}}}}}
\knowledge{\Lleq}{notion}

\newrobustcmd{\Jl}{\mathrel{\kl[\Jl]{<_{\mathcal{J}}}}}
\newrobustcmd{\nJl}{\mathrel{\kl[\Jl]{\nless_{\mathcal{J}}}}}
\newrobustcmd{\Jg}{\mathrel{\kl[\Jl]{>_{\mathcal{J}}}}}
\newrobustcmd{\nJg}{\mathrel{\kl[\Jl]{\ngtr_{\mathcal{J}}}}}
\knowledge{\Jl}{notion}
\newrobustcmd{\Rl}{\mathrel{\kl[\Rl]{<_{\mathcal{R}}}}}
\newrobustcmd{\Rg}{\mathrel{\kl[\Rl]{>_{\mathcal{R}}}}}
\knowledge{\Rl}{notion}
\newrobustcmd{\Ll}{\mathrel{\kl[\Ll]{<_{\mathcal{L}}}}}
\newrobustcmd{\Lg}{\mathrel{\kl[\Ll]{>_{\mathcal{L}}}}}
\knowledge{\Ll}{notion}

\newrobustcmd{\idemp}{\kl[\idemp]{\mathtt{E}}}
\knowledge\idemp{notion,ensuremath}
\newrobustcmd{\giidemp}{\kl[\giidemp]{\mathtt{GiE}}}
\knowledge\giidemp{notion,ensuremath}
\newrobustcmd{\oidemp}{\kl[\oidemp]{\mathtt{OE}}}
\knowledge\oidemp{notion,ensuremath}
\newrobustcmd{\ostidemp}{\kl[\ostidemp]{\mathtt{O^*E}}}
\knowledge\ostidemp{notion,ensuremath}
\newrobustcmd{\scatidemp}{\kl[\scatidemp]{\mathtt{ScE}}}
\knowledge\scatidemp{notion,ensuremath}
\newrobustcmd{\shuffleidemp}{\kl[\shuffleidemp]{\mathtt{ShE}}}
\knowledge\shuffleidemp{notion,ensuremath}
\newrobustcmd{\shsimp}{\kl[\shsimp]{\mathtt{SSE}}}
\knowledge\shsimp{notion,ensuremath}

\newrobustcmd{\aperiodic}{\kl[\aperiodic]{\mathtt{Ap}}}
\knowledge\aperiodic{notion,ensuremath}

\newrobustcmd{\s}{yes}
\newrobustcmd{\n}{no}
\newrobustcmd{\ds}{(yes)}
\newrobustcmd{\Eeqs}{e = e^2}
\newrobustcmd{\GIeqs}{\omegaoper e \omegastaroper e = e}
\newrobustcmd{\OIeqs}{e = \omegaoper e}
\newrobustcmd{\OIstareqs}{e = \omegastaroper e}
\newrobustcmd{\SIeqs}{e = \omegastaroper e \omegaoper e}
\newrobustcmd{\ShIeqs}{e = \shuffleoper e}
\newrobustcmd{\SSeqs}{\shsimp}

\newrobustcmd\smallsubseteq{\text{\raisebox{.12em}{\scalebox{.6}{${\subseteq}$}}}}
\knowledgenewrobustcmd{\eigi}{\ensuremath{\cmdkl{\mathtt{E{\smallsubseteq}GiE}}}}
\knowledgenewrobustcmd{\scish}{\ensuremath{\cmdkl{\mathtt{ScE {\smallsubseteq} ShE}}}}
\knowledgenewrobustcmd{\shiss}{\ensuremath{\cmdkl{\mathtt{ShE {\smallsubseteq} SSE}}}}
\knowledgenewrobustcmd{\oigi}{\ensuremath{\cmdkl{\mathtt{OE {\smallsubseteq} GiE}}}}
\knowledgenewrobustcmd{\ostigi}{\ensuremath{\cmdkl{\mathtt{O^*E {\smallsubseteq} GiE}}}}

\newrobustcmd{\foeqs}{\ensuremath{\aperiodic,\eigi,\scish, \shiss}}
\newrobustcmd{\fofeqs}{\ensuremath{\aperiodic,\oigi,\ostigi,\scish,\shiss}}
\newrobustcmd{\focuteqs}{\ensuremath{\aperiodic,\scish,\shiss}}
\newrobustcmd{\fofiniteeqs}{\ensuremath{\oigi,\ostigi,\scish,\shiss}}
\newrobustcmd{\fofinitecuteqs}{\ensuremath{\scish,\shiss}}
\newrobustcmd{\foscatteredeqs}{\ensuremath{\shiss}}

\newrobustcmd{\aperiodiceqs}{\exists n ~e^n=e^{n+1}}
\newrobustcmd{\eigieqs}{e=e^2 \rightarrow e^{\omega}e^{\omega*}=e}
\newrobustcmd{\scisheqs}{e=e^{\delta} \rightarrow e=e^{\shuffle}}
\newrobustcmd{\shisseqs}{e=e^{\shuffle} \rightarrow e \mbox{ is } \shsimp}
\newrobustcmd{\oigieqs}{e=e^{\omega} \rightarrow e^{\omega}e^{\omega*}=e}
\newrobustcmd{\ostigieqs}{e=e^{\omega*} \rightarrow e^{\omega}e^{\omega*}=e}

\newrobustcmd \emptyword {\varepsilon}

\newrobustcmd{\natorder}{\omega}
\newrobustcmd{\negorder}{\omega^*}
\newrobustcmd{\intorder}{\delta}
\newrobustcmd{\rationalorder}{\shuffle}
\newrobustcmd{\integers}{\mathbb{Z}}

\newrobustcmd\alphabet{\kl[\alphabet]{\ensuremath{\Sigma}}}
\knowledge\alphabet{notion}
\newrobustcmd{\subword}[2]{#1\kl[\subword]{|}_{#2}}
\knowledge\subword{notion}
\newrobustcmd\words[1]{#1^{\kl[\words]{\circsymbol}}}
\knowledge\words{notion}
\newrobustcmd{\nonemptywords}[1] {#1^{\kl[\nonemptywords]{\circplussymbol}}}
\knowledge\nonemptywords{notion}
\newrobustcmd{\reg}{\mathrm{r}}
\newrobustcmd{\factor}[2]{#1#2}

\knowledgenewrobustcmd{\omegaword}[1] {\ensuremath{#1^{\cmdkl\omega}}}
\knowledgenewrobustcmd{\omegastarword}[1] {\ensuremath{#1^{\cmdkl{\omega^*}}}}
\knowledgenewrobustcmd{\shuffleword}[1] {\ensuremath{#1^{\cmdkl{\shuffle}}}}
\knowledgenewrobustcmd{\negation}[1]{\ensuremath{\overline{#1}}}
\knowledgenewrobustcmd{\finitewords}[1]{\ensuremath{#1^{\cmdkl{*}}}}

\knowledgenewrobustcmd{\fo}{\ensuremath{\cmdkl{\mathsf{fo}}}}
\knowledgenewrobustcmd{\wmso}{\ensuremath{\cmdkl{\mathsf{wmso}}}}
\knowledgenewrobustcmd{\mso}{\ensuremath{\cmdkl{\mathsf{mso}}}}

\knowledgenewrobustcmd{\foq}[1][]{\ensuremath{\cmdkl{\mathsf{fo}[}#1\cmdkl{]}}}

\newrobustcmd{\focut}{\ensuremath{\foq[\existsCut]}}
\knowledge\focut{notion}
\newrobustcmd{\fofinite}{\ensuremath{\foq[\existsFinite]}}
\knowledge\fofinite{notion}
\newrobustcmd{\fofinitecut}{\ensuremath{\foq[\existsFinite,\existsCut]}}
\knowledge\fofinitecut{notion}
\newrobustcmd{\foordinal}{\ensuremath{\foq[\existsOrdinal]}}
\knowledge\foordinal{notion}
\newrobustcmd{\foscattered}{\ensuremath{\foq[\existsScattered]}}
\knowledge\foscattered{notion}

\newrobustcmd{\rexists}[1]{\kl[\rexists]{\exists^{#1}}}
\newrobustcmd{\rforall}[1]{\kl[\rforall]{\forall^{#1}}}

\newrobustcmd\existsIn[1]{\kl[\existsIn]{\exists}^{#1}}
\knowledge\existsIn{notion}
\newrobustcmd\forallIn[1]{\kl[\forallIn]{\forall}^{#1}}
\knowledge\forallIn{notion}

\newrobustcmd{\existsFinite}[1][]{\kl[\existsFinite]{\exists_{#1}^{\mathsf{finite}}}}
\knowledge\existsFinite{notion}
\newrobustcmd{\existsCut}{\kl[\existsCut]{\exists^{\mathsf{cut}}}}
\knowledge\existsCut{notion}
\newrobustcmd{\existsGap}{\kl[\existsGap]{\exists^{\mathsf{gap}}}}
\knowledge\existsGap{notion}
\newrobustcmd{\existsOrdinal}{\kl[\existsOrdinal]{\exists^{\mathsf{ordinal}}}}
\knowledge\existsOrdinal{notion}
\newrobustcmd{\existsScattered}{\kl[\existsScattered]{\exists^{\mathsf{scattered}}}}
\knowledge\existsScattered{notion}
\newrobustcmd{\forallFinite}{\kl[\forallFinite]{\forall^{\mathsf{finite}}}}
\knowledge\forallFinite{notion}
\newrobustcmd{\forallCut}{\kl[\forallCut]{\forall^{\mathsf{cut}}}}
\knowledge\forallCut{notion}
\newrobustcmd{\forallGap}{\kl[\forallGap]{\forall^{\mathsf{gap}}}}
\knowledge\forallGap{notion}
\newrobustcmd{\forallOrdinal}{\kl[\forallOrdinal]{\forall^{\mathsf{ordinal}}}}
\knowledge\forallOrdinal{notion}
\newrobustcmd{\forallScattered}{\kl[\forallScattered]{\forall^{\mathsf{scattered}}}}
\knowledge\forallScattered{notion}

\newrobustcmd{\cci}{{\mathsmaller{[\,]}}}
\newrobustcmd{\oci}{{\mathsmaller{(\,]}}}
\newrobustcmd{\coi}{{\mathsmaller{[\,)}}}
\newrobustcmd{\ooi}{{\mathsmaller{(\,)}}}
\newrobustcmd{\cui}{{\mathsmaller{[\,\rangle}}}
\newrobustcmd{\uci}{{\mathsmaller{\langle\,]}}}
\newrobustcmd{\uui}{{\mathsmaller{\langle\,\rangle}}}

\newrobustcmd{\unit}{\ensuremath{1}\xspace}
\newrobustcmd{\zero}{\ensuremath{0}\xspace}

\newrobustcmd{\ci}{\mathsmaller{[}}
\newrobustcmd{\oi}{\mathsmaller{(}}

\knowledgenewrobustcmd{\ccim}[1]{\cmdkl{\mathsmaller{[ #1 ]}}}
\knowledgenewrobustcmd{\ocim}[1]{\cmdkl{\mathsmaller{( #1 ]}}}
\knowledgenewrobustcmd{\coim}[1]{\cmdkl{\mathsmaller{[ #1 )}}}
\knowledgenewrobustcmd{\ooim}[1]{\cmdkl{\mathsmaller{( #1 )}}}
\knowledgenewrobustcmd{\cuim}[1]{\cmdkl{\mathsmaller{[ #1 \rangle}}}
\knowledgenewrobustcmd{\ucim}[1]{\cmdkl{\mathsmaller{\langle #1]}}}
\knowledgenewrobustcmd{\uuim}[1]{\cmdkl{\mathsmaller{\langle #1 \rangle}}}
\knowledgenewrobustcmd{\ouim}[1]{\cmdkl{\mathsmaller{(#1 \rangle}}}
\knowledgenewrobustcmd{\uoim}[1]{\cmdkl{\mathsmaller{\langle #1)}}}

\newrobustcmd{\morphsetmodify}[1] {#1}
\newrobustcmd{\invmorph}[1]{\langle {\morphsetmodify #1} \rangle}
\newrobustcmd{\minmorph}[1]{[ {\morphsetmodify #1} \rangle}
\newrobustcmd{\maxmorph}[1]{\langle {\morphsetmodify #1} ]}
\newrobustcmd{\minmaxmorph}[1]{[ {\morphsetmodify #1} ]}
\newrobustcmd{\nomaxmorph}[1]{\langle {\morphsetmodify #1} )}
\newrobustcmd{\nominmorph}[1]{( {\morphsetmodify #1} \rangle}
\newrobustcmd{\nominmaxmorph}[1]{( {\morphsetmodify #1} )}
\newrobustcmd{\minnomaxmorph}[1]{[ {\morphsetmodify #1} )}
\newrobustcmd{\maxnominmorph}[1]{( {\morphsetmodify #1} ]}

\newrobustcmd{\idap}{x^{\impliciti}=x^{\impliciti}\cdot x}
\newrobustcmd{\ideigi}{\omegaoper{(x^{\impliciti})} \cdot \omegastaroper{(x^{\impliciti})} = x^{\impliciti}}
\newrobustcmd{\idoigi}{x^{\implicitoi}\cdot \omegastaroper x =x^{\implicitoi}}
\newrobustcmd{\idostigi}{\omegaoper x \cdot x^{\implicitosi}=x^{\implicitosi}}
\newrobustcmd{\idscish}{\shuffleoper{\{x^{\implicitsc}\}} =x^{\implicitsc}}
\newrobustcmd{\idshiss}{\implicitsh(x,y_1,\dots,y_k)= \shuffleoper{(\{\implicitsh(x,y_1,\dots,y_k)\}\cup\{y_1,\dots,y_k\})}}

\knowledgenewrobustcmd{\expr}{\cmdkl{\ensuremath{\mathcal{X}}}}
\knowledgenewrobustcmd{\equations}{\cmdkl{\ensuremath{\mathcal{E}}}}

\knowledgenewrobustcmd{\logic}{\cmdkl{\ensuremath{\mathrm{Log}}}}
\knowledgenewrobustcmd{\extralogic}[1][]{\cmdkl{\ensuremath{\logic[\exists #1]}}}

\knowledgenewrobustcmd{\extclosure}{\cmdkl{\fformula{extensionClosure}}}
\knowledgenewrobustcmd{\starclosure}{\cmdkl{\ensuremath{^*}\fformula{Closure}}}

\knowledgenewrobustcmd{\cofinal}{\cmdkl{\fformula{coFinal}}}
\knowledgenewrobustcmd{\covering}{\cmdkl{\fformula{Covering}}}
\knowledgenewrobustcmd{\initial} {\cmdkl{\fformula{existsPrefix}}}
\knowledgenewrobustcmd{\coinitial}{\cmdkl{\fformula{coInitial}}}
\knowledgenewrobustcmd{\final} {\cmdkl{\fformula{existsSuffix}}}

\knowledgenewrobustcmd\fprefix{\cmdkl{\fformula{forallPrefix}}}
\knowledgenewrobustcmd\fomega{\cmdkl{\fformula{Omega}}}
\knowledgenewrobustcmd\fsuffix{\cmdkl{ \fformula{forallSuffix}}}

\knowledgenewrobustcmd{\definterval}{\cmdkl{\fformula{DefInterval}}}

\knowledgenewrobustcmd{\ramseyomega}[1]{\cmdkl{\omegaword{\cuim{#1}}}}
\knowledgenewrobustcmd{\ramseyomegastar}[1]{\cmdkl{\omegastarword{\ucim{#1}}}}

\newrobustcmd{\fletter}[1]{\formula{Letter}_{#1}}
\newrobustcmd{\fleft}{\formula{LeftInterval}}

\newrobustcmd{\circsymbol}{\ensuremath{\ostar}}

\knowledgenewrobustcmd\mainJ{\cmdkl{J}}
\knowledgenewrobustcmd\mainZ{\cmdkl{A}}
\newrobustcmd\nmainZ{\kl[\mainZ]{\bar A}}

\knowledgenewrobustcmd\mainZp{\cmdkl{B}}
\newrobustcmd\nmainZp{\kl[\mainZp]{\bar{B}}}

\newrobustcmd\omegasimilar{\mathtt{\textcolor{red}{omegasimilar}}}
\newrobustcmd\omegastarsimilar{\mathtt{\textcolor{red}{omegastarsimilar}}}
\newrobustcmd\similar{\mathtt{\textcolor{red}{similar}}}
\newrobustcmd\densesimilar{\mathtt{\textcolor{red}{densesimilar}}}
\newrobustcmd\definablecuts{\mathtt{\textcolor{red}{definablecuts}}}
\newrobustcmd\startWord{\textcolor{red}{startWord}}
\newrobustcmd\deprec[1]{\expandafter\newrobustcmd\csname#1\endcsname{\ensuremath{\mathtt{\textcolor{red}{#1}}}}}

\newrobustcmd{\sop}[1]{{#1}^{\kl[\sop]{sc}}}
\knowledge\sop{notion}

\newrobustcmd{\fop}[1] {{#1}^{\kl[\fop]{*}}}
\knowledge\fop{notion}

\newcommand{\mreg} {H}

\knowledgeconfigure{diagnose line=true} 

\knowledge{latex}{text=\LaTeX}

\knowledge{notion}
 | main theorem

\knowledge {e.g.}{italic}
\knowledge {i.e.}{italic}
\knowledge {resp.}{italic}
\knowledge {iff}{italic}
\knowledge {a priori}{italic}
\knowledge {dom}{link=domain,text=\ensuremath{\mathit{dom}}}

\knowledge {Green's relations}[relations of Green]{notion}

\knowledge{notion}
 | regular $\gJ$-class
 | regular@J-class
 | regular $\gJ$-classes
 | Regular $\gJ$-classes
 | regular
\knowledge {regular J-class}{link=regular $\gJ$-class, text=regular $\mathcal J$-class}
\knowledge {regular J-classes}{link=regular $\gJ$-class, text=regular $\mathcal J$-classes}

\knowledge {$\gJ$-class}[$\gJ$-classes]{notion}
\knowledge {J-class}{link=$\gJ$-class,text=$\mathcal J$-class}
\knowledge {J-classes}{link=$\gJ$-class,text=$\mathcal J$-classes}

\knowledge {$\gL$-class}[$\gL$-classes]{notion}
\knowledge {L-class}{link=$\gL$-class,text=$\mathcal L$-class}
\knowledge {L-classes}{link=$\gL$-class,text=$\mathcal L$-classes}

\knowledge {$\gR$-class}[$\gR$-classes]{notion}
\knowledge {R-class}{link=$\gR$-class,text=$\mathcal R$-class}
\knowledge {R-classes}{link=$\gR$-class,text=$\mathcal R$-classes}

\knowledge {$\gH$-class}[$\gH$-classes]{notion}
\knowledge {H-class}{link=$\gH$-class,text=$\mathcal H$-class}
\knowledge {H-classes}{link=$\gH$-class,text=$\mathcal H$-classes}

\knowledge {definable cuts}[definable cut|Definable cuts]{notion}

\knowledge{i->gi}{notion,text=$\mathtt{i{\rightarrow}gi}$}
\knowledge{oi->gi}{notion,text=$\mathtt{oi{\rightarrow}gi}$}
\knowledge{o*i->gi}{notion,text=$\mathtt{o^{*}i{\rightarrow}gi}$}
\knowledge{sc->sh}{notion,text=$\mathtt{sc{\rightarrow}sh}$}
\knowledge{sh->ss}{notion,text=$\mathtt{sh{\rightarrow}ss}$}

\knowledge {linear ordering}[linear orderings|Linear orderings]{notion}
\knowledge {successor}[successors|Successors]{notion}
\knowledge {predecessor}[predecessors|Predecessors]{notion}
\knowledge {condensed linear ordering}{notion}
\knowledge {countable linear ordering}[Countable linear orderings|countable linear orderings]{link=linear ordering}
\knowledge {scattered linear ordering}[scattered|scattered linear orderings|Scattered linear orderings]{notion}

\knowledge {generalized sum}{notion}
\knowledge {generalized concatenation}{notion}
\knowledge {generalized associativity}{notion}
\knowledge {dense}{notion}
\knowledge {convex}{notion}

\knowledge{notion}
 | convex
 | convex subset
 | convex subsets

\knowledge{notion}
  | condensation
  | condensed
  | condensations
  | Condensations
  | condensation of a linear ordering

\knowledge {consecutive}{notion}

\knowledge{notion}
 | Dedekind cut
 | Dedekind cuts
 | cut
 | cuts
 | Cuts

\knowledge{notion}
 | unnatural cut
 | unnatural cuts
 | Unnatural cuts

\knowledge{notion}
 | natural cut
 | natural cuts
 | Natural cuts

\knowledge{notion}
  | countable word
  | countable words
  | Countable words
  | word
  | words
  | Words
  | countable@word

\knowledge{notion}
 | scattered countable word
 | scattered countable words
 | scattered word
 | scattered words
 | Scattered words
 | scattered@word

\knowledge {notion}
  | alphabet
  | alphabets
  | Alphabet
  
\knowledge {notion}
  | letter
  | letters
  | Letters

\knowledge {notion}
  | factor
  | factors
  | Factors

\knowledge {notion}
  | subword
  | subwords
  | Subwords
  
\knowledge {notion}
 | empty word
\knowledge {notion}
 | perfect shuffle
 | perfect shuffles
\knowledge {notion}
 | perfectly dense

\knowledge{notion}
  | language
  | languages
  | Languages
  | language of countable words
  | languages of countable words

\knowledge {notion}
  | $\circsymbol$-monoid
  | $\circsymbol$-monoids
\knowledge {o-monoid}{text=$\circsymbol$-monoid,link=$\circsymbol$-monoid}
\knowledge {o-monoids}{text=$\circsymbol$-monoids,link=$\circsymbol$-monoid}

\knowledge {notion}
  | $\circsymbol$-semigroup
  | $\circsymbol$-semigroups
\knowledge {o-semigroup}{text=$\circsymbol$-semigroup,link=$\circsymbol$-semigroup}
\knowledge {o-semigroups}{text=$\circsymbol$-semigroups,link=$\circsymbol$-semigroup}

\knowledge {notion}
  | syntactic $\circsymbol$-monoid
  | syntactic $\circsymbol$-monoids
  | Syntactic $\circsymbol$-monoids
  | syntactic@o
\knowledge {syntactic o-monoid}{text=syntactic $\circsymbol$-monoid,link=syntactic $\circsymbol$-monoid}
\knowledge {syntactic o-monoids}{text=syntactic $\circsymbol$-monoids,link=syntactic $\circsymbol$-monoid}
\knowledge {Syntactic o-monoids}{text=Syntactic $\circsymbol$-monoids,link=syntactic $\circsymbol$-monoid}

\knowledge {notion}
  | free $\circsymbol$-monoid
  | free $\circsymbol$-monoids
  | Free $\circsymbol$-monoids
  | free@o
\knowledge {free o-monoid}{text=free $\circsymbol$-monoid,link=free $\circsymbol$-monoid}
\knowledge {free o-monoids}{text=free $\circsymbol$-monoids,link=free $\circsymbol$-monoid}
\knowledge {Free o-monoids}{text=Free $\circsymbol$-monoids,link=free $\circsymbol$-monoid}

\knowledge {notion}
  | $\circsymbol$-algebra
  | $\circsymbol$-algebras
\knowledge {o-algebra}{text=$\circsymbol$-algebra,link=$\circsymbol$-algebra}
\knowledge {o-algebras}{text=$\circsymbol$-algebras,link=$\circsymbol$-algebra}

\knowledge {semigroup}[Semigroups|semigroups]{notion}
\knowledge {monoid}[monoids|Monoids]{ignore}
\knowledge{product}[products|Products]{notion}
\knowledge{unit}{notion}
\knowledge{zero}{notion}
\knowledge{notion}
  | neutral element
  | neutral elements
  | Neutral elements
  | neutral
\knowledge {idempotent}[idempotents|Idempotents]{notion}

\knowledge {recognition}[recognized|recognizes|recognize|recognizability|recognizable]{notion}
\knowledge {notion}
  | derived operation
  | derived operations
  | Derived operations
  | operations derived
  | derived
  | derived operators
  
\knowledge {notion}
  | $\circsymbol$-monoid morphism
  | $\circsymbol$-monoid morphisms
  | morphism@o
  | morphisms@o
  | Morphims@o
  | morphism

\knowledge{identity}[identities|Identities]{notion}
\knowledge{identity element}{notion}
\knowledge{divides}[divide]{notion}

\knowledge {$\gJ $-class}[\gJ-class|\gJ-classes|$\gJ$-classes]{notion}
\knowledge {J-class}{link = $\gJ $-class,text=$\gJ $-class}
\knowledge {J-classes}{link = $\gJ $-class,text=$\gJ $-classes}

\knowledge {\legJ}[\leq _\gJ]{notion}
\knowledge {$\gJ $-equivalent}[\gJ]{notion}

\knowledge{notion}
 | regularJclass
 | regular@J-class
\knowledge {regular J-class}{link = regularJclass ,text=regular $\gJ $-class}
\knowledge {regular J-classes}{link = regularJclass,text=regular $\gJ $-classes}
\knowledge {non-regular J-class}{link = regularJclass ,text=non-regular $\gJ $-class}
\knowledge {non-regular J-classes}{link = regularJclass,text=non-regular $\gJ $-classes}

\knowledge{notion}
 | ordinal regular
 | ordinal@regular 
\knowledge{notion}
 | ordinal* regular
 | ordinal*@regular
\knowledge{notion}
 | shuffle regular
 | shuffle@regular
\knowledge{notion}
 | scattered regular
 | Scattered regular
 | scattered@regular
\knowledge{notion}
 | gap insensitive regular
 | gap insensitive@regular
\knowledge{notion}
 | shuffle simple regular
 | shuffle simple@regular

\knowledge{\Bin}{notion}
\knowledge{\Bnin}{notion}
\knowledge {\projection}{notion}
\knowledge {restricted projection}[\restrictedProjection]{notion}
\knowledge {\restrictedMonoidPowerset}{notion}

\knowledge {\Win}{notion}
\knowledge{\proj}{notion}

\knowledge{notion}
 | implicit operation
 | implicit operations

\knowledgedefault{}

\knowledge{notion}
  | weak set quantifier
  | weak set quantifiers
  | Weak set quantifiers

\knowledge{notion}
  | ordinal quantifier
  | ordinal quantifiers
  | Ordinal quantifiers

\knowledge{notion}
  | scattered set quantifier
  | scattered set quantifiers
  | Scattered set quantifiers

\knowledge{notion}
  | $\omega $-words
  | words of length~$\omega $

\knowledge{notion}
  | gap
  | gaps
  | Gaps

\knowledge{notion}
  | IHJ

\knowledge{notion}
  | concatenation witness
  | Concatenation witness
  | concatenation witnesses
  | Concatenation witnesses

\knowledge{notion}
  | omega witness
  | Omega witness
  | omega witnesses
  | Omega witnesses

\knowledge{notion}
  | omega$^*$ witness
  | Omega$^*$ witness
  | omega$^*$ witnesses
  | Omega$^*$ witnesses
  
\knowledge{notion}
  | shuffle witness
  | Shuffle witness
  | shuffle witnesses
  | Shuffle witnesses

\knowledge{notion}
  | monadic second-order logic
  | Monadic second-order logic
  | \mso -formulae
  | \mso -formula
  | \mso -sentence
  | \mso -sentences
  | Monadic second-order formulae
  | monadic second-order formulae
  | monadic second-order formula
  | Monadic second-order sentences
  | monadic second-order sentences
  | monadic second-order sentence

\knowledge{notion}
  | set variables
  | monadic variables
  | monadic variable

\knowledge{notion}
  | first-order logic
  | First-order logic
  | \fo -formulae
  | \fo -formula
  | First-order formulae
  | first-order formulae
  | first-order formula
  | \fo -sentence
  | \fo -sentences
  | First-order sentences
  | first-order sentences
  | first-order sentence

\knowledge{notion}
  | first-order variable
  | first-order variables
  | First-order variables

\knowledge{notion}
  | weak monadic second-order logic
  | weak second-order monadic logic
  | \wmso
  | WMSO

\knowledge{notion}
  | $L$-formulae

\knowledge{notion}
  | free variable
  | free variables

\knowledge{notion}
  | sentence
  | sentences
  | $\logic $-sentence
  | $\logic $-sentences

\knowledge{notion}
  | dense in itself

\knowledge{notion}
  | well-founded
  | well ordering
  | well ordered

\knowledge{notion}
  | first-order quantifier
  | first-order quantifiers
  | First-order quantifiers

\knowledge{notion}
  | infinite tree
  | infinite trees
  | Infinite trees

\knowledge{notion}
  | scattered set
  | scattered sets
  | Scattered sets

\knowledge{notion}
  | gap insensitive idempotent
  | gap insensitive idempotents
  | gap insensitive
  | Gap insensitive@idempotent
  | gap insensitive@idempotent
  | Gap insensitive idempotents

\knowledge{notion}
  | Ordinal idempotents
  | ordinal@idempotent
  | ordinal idempotent
  | ordinal idempotents

\knowledge{notion}
  | ordinal* idempotent
  | ordinal*@idempotent
  | ordinal* idempotents
  | Ordinal* idempotents

\knowledge{notion}
  | scattered idempotent
  | Scattered idempotents
  | scattered idempotents
  | scattered@idempotent

\knowledge{notion}
  | shuffle idempotent
  | Shuffle idempotents
  | shuffle idempotents
  | shuffle@idempotent

\knowledge{notion}
  | shuffle simple idempotent
  | shuffle simple@idempotent
  | Shuffle simple idempotents

\knowledge{notion}
  | relativized
  | relativization
  | Relativization
  | relativized formula
  | relativized formulae

\knowledge{notion}
  | definable
  | defines
  | define

\knowledge{notion}
  | domain
  | domains
  | domain@word

\knowledge{notion}
  | order-type

\knowledge{notion}
  | interval notations

\knowledge{notion}
  | set quantifier
  | set quantifiers
  | Set quantifiers

\knowledge{notion}
  | $C$-restricted set quantifier
  | restricted set quantifiers
  | restricted set quantifier
  | Restricted quantifiers
  | restricted quantifiers
  | restricted quantifier

\knowledge{notion}
  | extension-closed

\knowledge{notion}
  | letter witness
  | Letter witness
  | letter witnesses
  | Letter witnesses

\knowledge{notion}
  | \formula {Letter}

\knowledge{notion}
  | \formula {Language}

\knowledge{notion}
  | regular operations
  | Regular operations

\knowledge{notion}
  | Regular expressions
  | Regular expression
  | regular expressions
  | regular expression
  | expressions
  | expression

\knowledge{notion}
  | $\expr $-expression
  | $\expr $-expressions
  | $\expr $-expressible
  | $\expr$-expressibility

\knowledge{notion}
  | gap fall
  | Gap fall
  | Gap Fall

\knowledge{notion}
  | Scatter expressions
  | Scatter expression
  | scatter expression
  | scatter expressions

\knowledge{notion}
  | interval
  | intervals

\knowledge{notion}
  | closed
  | right closed
  | left closed

\knowledge{notion}
  | open
  | left open
  | right open

\knowledge{notion}
  | ordinal

\knowledge{notion}
  | subordering

\knowledge{notion}
  | marked Kleene star

\knowledge{notion}
  | marked concatenation
  
\knowledge{notion}
  | unrestricted concatenation

\knowledge{notion}
  | Kleene star

\knowledge{notion}
  | scatter operator
  | scatter operation
  | scatter operations
  | scatter
  | Scatter operator

\knowledge{notion}
  | omega word

\knowledge{notion}
  | omega* word

\knowledge{notion}
  | shuffle

\knowledge{notion}
  | omega power
  | omega
  | omega@power

\knowledge{notion}
  | omega* power
  | omega*
  | omega*@power

\knowledge{notion}
  | shuffle power

\knowledge{notion}
  | power-free expression
  | power-free expressions
  | power-free
  | Power-free expressions
  | Power-free expression

\knowledge{notion}
  | scatter-free expression
  | scatter-free expressions
  | scatter-free
  | Scatter-free expression
  | Scatter-free expressions

\knowledge{notion}
  | marked expression
  | marked expressions
  | marked
  | Marked expression
  | Marked expressions
  | marked operations

\knowledge{notion}
  | marked star-free expression
  | marked star-free expressions
  | marked star-free
  | Marked star-free expression
  | Marked star-free expressions

\knowledge{notion}
  | left marked expression
  | left marked expressions
  | left marked
  | Left marked expression
  | Left marked expressions

\knowledge{notion}
  | \formula {Left}

\knowledge{notion}
  | \formula {Right}

\knowledge{notion}
  | \formula {uWords}

\knowledge{notion}
  | \formula {LeftInterval}

\knowledge{notion}
 | variety
 | varieties

\knowledge{notion}
  | profinite identity
  | Profinite identities

\knowledge{notion}
  | consecutive intervals
\knowledge{notion}
  | unrestricted Kleene star

\knowledge{notion}
  | syntactic

\knowledge{notion}
  | extension closed
\knowledge{notion}
  | $\mainJ $-witness
  | witness

\knowledge{notion}
  | aperiodicity
  | aperiodic
\usepackage{mathtools}

\nolinenumbers

\title{Regular expressions over countable words}
\author{Thomas Colcombet}{IRIF Paris}{thomas.colcombet@irif.fr }{}{}
\author{A V Sreejith}{IIT Goa}{sreejithav@iitgoa.ac.in}{}{}
\authorrunning{Colcombet and Sreejith}

\begin{document}

\setcounter{tocdepth}{4}\tableofcontents\clearpage

\maketitle
\begin{abstract}

We investigate the expressive power of regular expressions for languages of countable words and establish their expressive equivalence with logical and algebraic characterizations.  
Our goal is to extend the classical theory of regular languages - defined over finite words and characterized by automata, monadic second-order logic, and regular expressions - to the setting of countable words. 
In this paper, we introduce and study five classes of expressions: \emph{marked star-free expressions}, \emph{marked expressions}, \emph{power-free expressions}, \emph{scatter-free expressions}, and \emph{scatter expressions}. 
We show that these expression classes characterize natural fragments of logic over countable words and possess decidable algebraic characterizations. 
As part of our algebraic analysis, we provide a precise description of the relevant classes in terms of their \(\mathcal{J}\)-class structure. 
These results complete a triad of equivalences - between logic, algebra, and expressions - in this richer setting, thereby generalizing foundational results from classical formal language theory.
        
\end{abstract}


\newpage
\section{Introduction and contributions}

This paper contributes to the theory of regular languages over "countable words"—that is, infinite words whose domains are countable linear orders, as opposed to $\omega$-words whose domains are isomorphic to the natural numbers. The foundational result in this direction is the decidability of satisfiability for monadic second-order logic (MSO) over countable words or chains~\cite{Rabin69,Shelah75}. Since then, the study of countable words has evolved into a rich theory that retain many foundational aspects of classical regular language theory over finite words. In particular, regular languages of finite words admit equivalent characterizations through several formal systems: automata (deterministic or nondeterministic), definability in MSO, recognizability by finite monoids, and regular expressions. Analogous equivalences have been established for languages of countable words; notably, MSO-definable languages over countable linear orders are known to coincide with those recognizable by "o-monoids"~\cite{CartonColcombetPuppis18}.

A striking branch of research concerning regular languages of finite words is the algebraic characterization of sub-logics of MSO. The seminal result in this direction is the celebrated Schützenberger–McNaughton–Papert theorem, that states that a language is recognized by an aperiodic monoid if and only if it is definable by a star-free expression, which is itself equivalent to being definable in first-order logic (FO). This theorem initiated a deep line of research in formal language theory.

In the setting of countable words, \cite{colcombeticalp15} establishes five correspondences between varieties of "o-monoids" and natural sub-logics of MSO. These sub-logics arise by restricting set quantification in various ways: to singletons (yielding FO), to Dedekind cuts, to finite sets, and so on.

This raises a natural question: can we identify natural classes of expressions that characterize these logics and algebraic varieties?

To better understand the expressive power of the different classes of expressions we study, we now present a few illustrative examples. 

\begin{example}
    Let us consider the alphabet $\alphabet = \{a,b\}$ and denote by $\emptyset$ the empty language. The set of all countable words $\words \alphabet$ is then the complement of the empty language, written as $\negation \emptyset$. We present a few illustrative expressions:
    
    \begin{itemize}
        \item The set of all words without any occurrence of the letter $b$ is given by the "marked star-free expression":
        \[
        \fNob \;\defs\; \negation{\words \alphabet b \words \alphabet},
        \]
        where the concatenation is marked by the letter $b$.
    
        \item The set of all \emph{dense} words over $a$ is described by the following "marked star-free expression":
        \[
        \fDensea \;\defs\; \negation{\words \alphabet (aa) \words \alphabet}\ \bigcap\ \fNob\ \bigcap\ \words \alphabet a \words \alphabet,
        \]
        which expresses that there are no consecutive $a$’s, no occurrence of $b$, and at least one occurrence of $a$.
    
        \item The set of all words containing an \emph{odd number of $b$’s} separated by dense occurrences of $a$ is given by the "marked expression":
        \[
        \fOddbDensea \;\defs\; \fop{\big(b \ \fDensea\ b \ \fDensea \big)}{b},
        \]
        where the "Kleene star" is marked by the leftmost letter $b$.
    
        \item The property that the \emph{domain is well-founded} (i.e., there is no infinite descending sequence) can be expressed by a "power-free expression":
        \[
        \fWellFounded \;\defs\; \negation{\words \alphabet\ \big( \negation{\alphabet \words \alphabet\ +\ \words \alphabet (a + b) \words \alphabet } \big)}.
        \]
        Similarly, a corresponding expression characterizes \emph{reverse well-foundedness}. The domain is finite when it is both well-founded and reverse well-founded.
    
        \item The property that the \emph{domain is scattered} is captured by the "scatter expression":
        \[
        \sop \alphabet,
        \]
        and the property that the occurrences of the letter $b$ are scattered is given by:
        \[
        \sop {\big(\fNob \ b \ \fNob \big)}.
        \]
    \end{itemize}
\end{example}
    
In this work, we complete the results of \cite{colcombeticalp15} by introducing five classes of regular expressions that match the five logical and algebraic characterizations. This leads to our main theorem, which establishes an expressive equivalence between logic (e.g., $\fo$, \dots), algebraic characterizations (e.g., varieties of o-monoids such as $\aperiodic$, \dots), and classes of regular expressions (e.g., "marked expressions", \dots). 

A key technical contribution of our work is a detailed structural analysis of o-monoids. While "J-classes" play a central role in classical algebraic characterizations of regular languages, their behavior in the setting of "o-monoids" is more intricate. In this paper, we provide a framework for analyzing the J-classes and establish several new properties that are crucial for connecting the algebraic, logical and expression perspectives. In particular, our structure theorem for "o-monoids" (\Cref{thm:fundamental}) reveals how idempotents interact in the countable setting. These properties provide for classifying "o-monoids" in terms of their "J-classes" and also the type of idempotents they contain. We then show that these natural classes turn out to be equivalent to interesting classes of expressions. 

Although we have not yet introduced all the necessary technical definitions,\footnote{Note that in the electronic version of this document, most technical terms are hyperlinked to their definitions.} we state the main result here:

\begin{theorem}[main theorem]\AP\label{theorem:main}\phantomintro{main theorem}%
    Let $\monoid$ be the "syntactic" "o-monoid" of a \kl{language} $L\subseteq\words\alphabet$, then:
    \begin{enumerate}
        \item $L$ is {definable} in \fo\ if and only if $\monoid$ satisfies the equations $\foeqs$, if and only if all its "regular J-classes" are "shuffle simple regular", if and only if $L$ is definable by "marked star-free expressions".
        \item $L$ is {definable} in \fofinite\ if and only if $\monoid$ satisfies the equations $\fofiniteeqs$, if and only if all its "ordinal regular" "J-classes" and "ordinal* regular" "J-classes" are "shuffle simple regular", if and only if $L$ is {definable} by "marked expressions".
        \item $L$ is {definable} in \focut\ if and only if $\monoid$ satisfies the equations $\focuteqs$, if and only if it is "aperiodic" and all its "scattered regular" "J-classes" are "shuffle simple regular", if and only if $L$ is {definable} by "power-free expressions".
        \item $L$ is {definable} in \fofinitecut\ if and only if $\monoid$ satisfies the equations  $\fofinitecuteqs$, if and only if all its "scattered regular" "J-classes" are "shuffle simple regular", if and only if $L$ is {definable} by "scatter-free expressions".
        \item $L$ is {definable} in \foscattered\ if and only if $\monoid$ satisfies the equations $\foscatteredeqs$, if and only if all its "shuffle regular" "J-classes" are "shuffle simple regular", if and only if $L$ is {definable} in "scatter expression".
    \end{enumerate}
    Furthermore, the problem of deciding whether a language is definable by any of the above expression is decidable.
\end{theorem}

We summarize the main contributions of this paper as follows: (1) the definition of regular expressions over "countable words"; (2) a precise structural analysis of the "J-classes" of "o-monoids"; and (3) the establishment of expressiveness equivalences between sub-classes of "o-monoids" and classes of regular expressions. The equivalence between logic and algebraic identities was previously established in \cite{colcombeticalp15}, and the translation from expressions to logic is relatively straightforward.

Concretely, \Cref{theorem:main} is derived by combining the logic-to-algebraic translation from \cite{colcombeticalp15}, the structural insights from \Cref{thm:jclass}, and the expressiveness analysis from \Cref{lemma:main}. The main technical challenge lies in the need to treat all five cases simultaneously.

\subsection*{Related work}
Some algebraic studies of regular languages over infinite words have already been undertaken (e.g., \cite{Etessami99}). In particular, Bes and Carton studied the special case of "scattered countable words@scattered", and established algebraic characterizations for fragments such as first-order logic (\fo) \cite{BesCarton11} and weak monadic second-order logic (\wmso) \cite{BesCartonPersonal}.

Over general countable words, the seminal work of Carton, Colcombet, and Puppis \cite{CartonColcombetPuppis18} introduced the algebraic framework of "o-monoids", which laid the foundation for subsequent logical characterizations. Building on this, we \cite{colcombeticalp15} established characterizations of various logical fragments—such as first-order logic (with and without cuts), and weak monadic second-order logic (with and without cuts)—in terms of varieties of o-monoids. Further developments include the study of the two-variable fragment of first-order logic \cite{amalmfcs16}, temporal logic \cite{bharatlics19} and logics with infinitary first-order quantifiers \cite{bharatlics19}. An alternative approach to understanding o-monoids through decomposition theorems has been explored in \cite{bharatlics19, bharatfct21, bharatjcss23}. Finally, o-monoids have also been applied to the study of first-order separation problems \cite{colcombet22}.

\subsection*{Structure of the document}
In \Cref{section:preliminaries}, we introduce the basic notions related to "countable words".
In \Cref{section:regex}, we introduce various classes of "regular expressions" along with motivating examples.
\Cref{section:algebra} presents the algebraic formalism of "o-monoids", and develops the necessary the structural analysis of "o-monoids". This culminates with the characterization of the varieties by means "J-classes", stated in \Cref{thm:jclass}.
\Cref{section:monoid to expression} is devoted to the technically involved direction of our main result: the construction of expressions from algebraic identities, leading to the proof of \Cref{lemma:main}.


\section{Preliminaries: Languages and Logic}
\label{section:preliminaries}

In this section, we introduce the basic notions concerning "linear orderings" in \Cref{subsection:linear orderings}, then concerning words in \Cref{subsection:words}. In \Cref{subsection:mso}, we introduce "monadic@monadic second-order logic" and "first-order logic" and various other logics.

\subsection{Linear orderings}
\label{subsection:linear orderings}

\AP
A countable ""linear ordering"" $\alpha = (X,<)$ is a countable set~$X$ equipped with a total order $<$. In this case, we say that $X$ is the \intro{domain} of $\alpha$. 
We say that two linear orderings have the same \intro{order-type} if there is an order preserving bijection between their domains. We denote by $\natorder, \negorder, \rationalorder$ the order types of $(\nats, <), (-\nats, <), (\rationals,<)$ respectively.
Given \kl{linear orderings} $(\beta_i)_{i\in\alpha}$ (assumed disjoint up to isomorphism) indexed with a \kl{linear ordering}~$\alpha$,
their  \intro{generalized sum} $\sum_{i\in\alpha}\beta_i$ is the \kl{linear ordering} over the (disjoint) union of the sets of the $\beta_i$'s,
with the order defined by $x<y$ if either $x\in\beta_i$ and $y\in\beta_j$ with $i<j$, or $x,y\in\beta_i$ for some~$i$, and $x<y$ in $\beta_i$.

\AP
Let $\alpha = (X,<)$ be a linear ordering. A \intro{Dedekind cut} (or simply a \reintro{cut}) of $\alpha$ is a left-closed subset $Y \subseteq X$, \kl{i.e.}, for all $x<y$ if $y\in Y$ then $x\in Y$. A \intro{gap} is a \kl{cut} which neither has a maximum element nor its complement has a minimal element.
An ordering $\beta = (Y,<)$ is a ""subordering"" of $\alpha$ if $Y\subseteq X$ and the order on $Y$ is induced by the order on $X$.
We denote $\intro*\subword{\alpha}{Y}$ the induced subordering of $\alpha$. A subset $I\subseteq \alpha$ is an \intro{interval} if whenever $x,y\in I$
and $x<z<y$, $z\in I$. Two \kl{intervals} $I,J$ of a \kl{linear ordering} are \intro{consecutive} if $I$ and $J$
are disjoint and their union is an interval.
\AP Given elements $x,y$ in $\alpha$ we denote by \intro[\openprefix]{\openprefix{y}} the "interval" $\{z \in \alpha \mid z < y \}$ and by \intro[\closedprefix]{\closedprefix{y}} the interval \intro[\openprefix]{\openprefix{y}} union $\{y\}$. We similarly define intervals \intro[\opensuffix]{\opensuffix x} and \intro[\closedsuffix]{\closedsuffix x}. The interval \intro[\openset]{\openset{x}{y}} is defined as the intersection of sets \intro[\openprefix]{\openprefix y} and \intro[\opensuffix]{\opensuffix x}. Similarly we define intervals \intro[\rightopenset]{\rightopenset{x}{y}}, \intro[\leftopenset]{\leftopenset{x}{y}} and \intro[\closedset]{\closedset{x}{y}}. An interval is ""left open"" if the left endpoint is not included. It is ""left closed"" otherwise. Similarly, we have ""right open"" and ""right closed"" intervals. An ""open"" (resp. ""closed"") interval is both left and right open (resp. closed).

\AP
We say linear ordering $\alpha$ is finite (resp. infinite) if its domain is finite (resp. infinite); is ""dense"" if it contains at least two elements and between any two points in the set there is another point; is \intro{scattered} if none of its sub-orderings is \intro{dense}; is a \intro{well ordering} if every subset has a minimal element; an \intro{ordinal} if it does not have an $\omega^*$ sub-ordering. Examples: $(\rationals,<)$ is dense and $(\nats,<)$ and $(\integers, <)$ are scattered.



For more information on linear orderings see \cite{Rosenstein}.

\subsection{Countable words}
\label{subsection:words}

\AP
Fix an \intro{alphabet} $\intro*\alphabet$ which is a finite set of symbols called ""letters"". 
Given a \kl{linear ordering} $\alpha$, a \intro{word} of \intro{domain} $\alpha$ is a mapping from $w: \alpha \rightarrow \alphabet$.
The \kl{domain} of a \kl{word} is~$\alpha$, and is denoted $\dom w$.
\AP ""Countable words"" are words of countable "domain@@word".
\AP The set of all "countable words" over an alphabet~$\Sigma$ is
denoted $\intro*\words \Sigma$ and the set of all words over non-empty countable domain is
denoted $\intro*\nonemptywords \Sigma$.
A \intro{language} is a subset of $\words\alphabet$.
\AP
For a word $w$ and a subset $I \subseteq \dom w$, $\subword{w}{I}$ denotes the \intro{subword} got by restricting $w$ to $I$. 
If $I$ is a convex subset then $\subword{w}{I}$ is called a \intro{factor} of $w$. 
\AP
The \intro{generalized concatenation} of the \kl{words} $(w_i)_{i\in\alpha}$ (supported by disjoint \kl{domains}) indexed by a \kl{linear ordering}~$\alpha$ is
$ \prod_{i \in \alpha} w_i$
and denotes the \kl{word} of \kl{domain} $\sum_i \dom {w_i}$ which coincides with each $w_i$ over $\dom{w_i}$ for all $i\in\alpha$.

\AP
Let us now look at some special words. The \intro{empty word} $\emptyword$, is the only \kl{word} of empty \kl{domain}. The ""omega word"" $\intro*\omegaword{a}$ is the \kl{word} over a single \kl{letter} $a$ and of \kl{domain} $(\nats,<)$.
Similarly the ""omega* word"" $\intro*\omegastarword{a}$ is the word over a single letter $a$ and of domain negative integers. For a language $L$, the language $\omegaword L$ consists of all words $\prod_{i \in \nats} u_i$ where $u_i \in L$ for all $i \in \nats$. Similarly we define $\omegastarword L$.
\AP
Finally, the ""shuffle"" of a non-empty set $S \subseteq \alphabet$ (denoted by $\intro*\shuffleword S$) is a word of \kl{domain}~$(\rationals,<)$ in which all non-empty intervals $(x,y)$ contain at least one occurrence of each letter in $S$. The \kl{word} is unique up to isomorphism (see Shellah \cite{Shelah75}). We can extend the notion of shuffle to a finite set of languages $\mathcal{L} = \{L_1, \dots, L_k\}$. We define $\shuffleword {\mathcal{L}}$ to be the set of all words of the form $\prod_{i \in \rationals} w_{f(i)}$ where $w_{f(i)} \in L_{f(i)}$, and $f$ is the unique shuffle over the set of letters $\{1,2,\dots,k\}$.
\begin{example}
Let the alphabet be $\alphabet = \{a,b\}$. Consider the word $w: \mathbb{Q} \rightarrow \alphabet$ where the letter at the position $\frac{p}{q}$ in reduced fractional form is given by 
\begin{align*}
  w\Big(\frac{p}{q}\Big)=\begin{cases}
               a, \text{ if } q \text{ is odd} \\
               b, \text{ otherwise.}\\
            \end{cases}
\end{align*}
Clearly it is countable, has at least two elements, no minimal nor maximal elemements, and all intervals having at least two elements contain both the letters $a$ and $b$. As a consequence, it is isomorphic to the "shuffle" of~$\{a,b\}$.
\end{example}

\subsection{"First-order@first-order logic", "monadic second-order@monadic second-order logic" and logics in between}
\label{subsection:mso}

We study several logics that are restrictions of "monadic second-order logic" interpreted over "countable words".

\AP
The ""first-order variables"" range over positions of a word and the ""set variables"" (also called ""monadic variables"") range over sets of positions of a word. ""Monadic second-order formulae"" (\reintro{\mso-formulae} for short\phantomintro\mso) are defined according to the following recursive syntax:
\begin{align*}
	x<y \mid x \in X \mid a(x) \mid \phi_1 \wedge \phi_2 \mid \phi_1 \vee \phi_2 \mid  \neg \phi \mid \exists x ~\phi \mid \exists X ~\phi\ ,
\end{align*}
in which~$a$ is a letter in a finite alphabet~$\alphabet$, $x$ is a "first-order variable", $X$ a "monadic variable", and $\phi$, $\phi_1$, $\phi_2$ are "\mso-formulae".\phantomintro{set quantifier}
""First-order formulae"" (or \reintro{\fo-formulae}\phantomintro{\fo}) are "\mso-formulae" which do not use "set quantifiers".


We also consider logics sitting somewhere between \fo\ and \mso. These are defined using restricted forms of set quantifiers.
\begin{definition}\label{definition:logics}
	We consider the following quantifiers:
	\begin{itemize}
	\itemAP $\intro*\existsFinite X\ \phi$ holds if there exists a finite set~$X$ such that~$\phi$ holds.
	\itemAP $\intro*\existsCut X\ \phi$ holds if there exists a "cut" $X$ such that $\phi$ holds.
	\itemAP $\intro*\existsScattered X\ \phi$ holds if there exists a "scattered" set $X$ such that $\phi$ holds.
	\end{itemize}
	\AP The logic~$\intro*\foq[Q]$ is the extension of \fo\ with quantifiers in~$Q$, for~$Q$ a set of
	quantifiers among~$\existsCut$, $\existsFinite$, and $\existsScattered$.
\end{definition}
\AP In particular, we recognize in this definition $\foq[\emptyset]$ is nothing but \fo, and that \fofinite\ is what is usually called ""weak second-order monadic logic"" (\intro*\wmso\ for short). We refer to \cite{colcombeticalp15} for more details about these logics.


\section{Regular Expressions}
\label{section:regex}

\AP
In this section, we define the ""regular operations"" used in this paper. Our expressions build upon the usual operations over finite words -- negation, union (denoted by $+$), intersection, concatenation (denoted by $.$), and ""Kleene star"" (denoted by $\intro[\fop]{}\fop{}$) -- and are extended by the ""scatter operator"" (denoted by $\intro[\sop]{}\sop{}$). 

\AP We define ""scatter expressions"" (or simply ""expressions"") over an alphabet $\alphabet$ by the following grammar:
\[
 S \defs \emptyset \mid \sigma \in \alphabet \mid S + S \mid S \cap S \mid \negation S \mid S S \mid \fop S \mid \sop S
\]
\AP We are also interested in a few sub classes of expressions. A ""power-free expression"" is an expression that does not use the operators $\fop S$ or $\sop S$. A ""scatter-free expression"" is an expression that does not use the operator $\sop S$. 

The semantics of expression $r$ is given by the language $\lang r \subseteq \words \Sigma$ that is defined inductively.
The base cases are $\lang {\emptyset} = \emptyset$ and $\lang {\sigma} = \{\sigma\}$. 
The semantics of the boolean operators are natural.
The semantics of other operators are: $\lang{\reg_1 \reg_2} = \lang {\reg_1}\lang {\reg_2}$, 
\begin{align*}
\lang {\fop \reg} & = \{\prod_{i \in \alpha} u_i \mid \alpha \text{ is a finite ordering and } u_i \in \lang{\reg} \}, \text{ and }\\
\lang{\sop \reg} & = \{\prod_{i \in \alpha} u_i \mid \alpha \text{ is a "scattered ordering@scattered linear ordering" and } u_i \in \lang{\reg}\}.
\end{align*}

\AP
A ""left marked expression"" is an expression of the form $\sigma \reg$ where $\sigma$ is a letter and $r$ is an expression.
We define ""marked expressions"" by the grammar:
\[
  M \defs \emptyset \mid  M + M \mid M \cap M \mid \negation M \mid M \sigma M \mid M \fop {(\sigma M)}
\]
where $\sigma$ is a letter. A ""marked star-free expression"" is a marked expression that does not contain the "Kleene star". We say that "marked expressions" are closed under ""marked concatenation"" and ""marked Kleene star"" in contrast to ""unrestricted concatenation"" and ""unrestricted Kleene star"" for expressions in general.

\begin{remark}
	Let $\mreg_1$ and $\mreg_2$ be "marked expressions". Since the language $\{\epsilon\}$ is definable by a marked expression (see \Cref{sec:examples}.\Cref{itm:epsilon}), we can define expressions like $\fop{(\sigma \mreg_1)}$ within the marked expression syntax.
	
	Moreover, the expression $\fop{(\sigma_1 \mreg_1 + \sigma_2 \mreg_2)}$ is equivalent to the marked expression:
	\[
	\fop{(\sigma_1 \mreg_1 \fop{(\sigma_2 \mreg_2)})} + \fop{(\sigma_2 \mreg_2 \fop{(\sigma_1 \mreg_1)})}.
	\]
	Hence, marked expressions are closed under the "Kleene star" applied to unions of "left marked expressions".
\end{remark}

	\subsection{Examples}\label{sec:examples}
	We now describe several examples that illustrate the expressive power of various classes of expressions. Many of these will also be used in our later arguments. We begin with languages definable using "marked star-free expressions".
	\begin{enumerate}
	\itemAP \label{itm:epsilon}
	  The set $\words \alphabet$ is defined by $\negation \emptyset$. The set of all words \emph{not containing} the letter $a$ is $\words \alphabet \backslash \words \alphabet a \words \alphabet$. The set $\nonemptywords \alphabet \defs \words \alphabet \backslash \epsilon$, where $\epsilon = \negation{\words \alphabet \alphabet \words \alphabet}$ and $\words \alphabet \alphabet \words \alphabet \defs \bigcup_{a \in \alphabet} \words \alphabet a \words \alphabet$.
	
	\itemAP ``There is at least one letter'' is defined by the expression $\words \alphabet \alphabet \words \alphabet$.
	
	\itemAP \label{item:min}
	  ``The domain has a minimum (respectively, maximum)'' is defined by the expressions $\intro*\fMin \defs \alphabet \words \alphabet$ (resp. $\intro*\fMax \defs \words \alphabet \alphabet$).
	
	\itemAP \label{item:ex-dense}
	  ``The domain is dense'' is expressed as $\negation {\words \alphabet \alphabet \alphabet \words \alphabet \cup \epsilon}$.
	
	\itemAP \label{item:ex-ab}
	  The expression $\intro*\fexistab \defs \words \alphabet (ab) \words \alphabet$ accepts words that contain the factor $ab$, and $\intro*\fnoab \defs \negation{\fexistab}$ accepts words that do not.
	
	\itemAP \label{item:ex-omega-ab}
	  ``There are occurrences of the factor $ab$ arbitrarily close to the end of the word, but no maximal one'' is defined as $ \negation{\words \alphabet (ab) \fnoab} \cap \fexistab$.
	
	\itemAP \label{item:ex-finite}
	  ``The domain is finite'' and ``the domain has even size'' are both definable using marked expressions: $\fop \alphabet$ and $\fop{(\alphabet \alphabet)}$, respectively.
	\end{enumerate}
	
	We now turn to examples definable using "power-free expressions".
	
	\begin{enumerate}\setcounter{enumi}{7}
	\itemAP \label{item:ex-ordinal}
	  ``The domain is well-founded'' is defined as the complement of ``there is a cut after which there is no first letter'':\\
	  $\negation{\words \alphabet \negation{\alphabet \words \alphabet + \epsilon}}$.\\
	  Similarly, one can express that the ``domain is reverse well-founded''.
	
	\itemAP
	  ``The domain is finite'' is defined as ``the domain is well-founded and reverse well-founded''.
	
	\itemAP \label{item:ex-gap}
	  ``There is a gap'' is defined as $\intro*\fGap \defs (\negation{\words \alphabet \alphabet + \epsilon})(\negation{\alphabet \words \alphabet + \epsilon})$,\\
	  ``there are no gaps'' by $\intro*\fnogap \defs \negation{\fGap}$,\\
	  ``exactly one gap'' by $\intro*\fonegap \defs (\fnogap \cap \negation{\words \alphabet \alphabet + \epsilon})(\fnogap \cap \negation{\alphabet \words \alphabet + \epsilon})$, and\\
	  ``exactly one gap with a minimum'' by $\intro*\fminonegap$.
	
	\itemAP
	  ``There is an $\omega$-sequence of gaps'' is expressed as\\
	  $\big(\negation{\words \alphabet \fonegap}\ \wedge\ \fGap \big) \words \alphabet$.\\
	  Similarly, one can express the existence of an $\omega^*$-sequence of gaps.
	
	\itemAP \label{item:ex-finite-gap}
	  ``There are at most finitely many gaps'' is accepted by $\intro*\ffingap$, which asserts that there is no $\omega$ or $\omega^*$ sequence of gaps.
	
	\itemAP \label{item:ex-even-gap}
	  ``There is an even number (non-zero) of gaps'' is definable using a scatter-free expression:\\
	  $\fonegap\ \fop {(\fminonegap\ \fminonegap)} \fminonegap$.
	
	\itemAP \label{item:ex-scattered}
	  ``The domain is scattered'' is defined using the scatter expression $\sop \alphabet$;\\
	  ``the letter $a$ is scattered'' is expressed as\\
	  $\sop {\big(\words{(\Sigma \backslash \{a\})} a \words{(\Sigma \backslash \{a\})}\big)}$.
	\end{enumerate}

\subsection{Operations on Languages}
\label{sec:logic-closure}

\AP
In this section, we describe constructions on languages that can be expressed using the operations defined in the previous section. We say that a language $F$ is ""extension closed"" if $\words \alphabet F \words \alphabet \subseteq F$ (equivalently, $F = \words \alphabet K \words \alphabet$ for some language~$K$).

\begin{lemma} \label{lem:operations}
Let $F$ be a language.
\begin{enumerate}
\itemAP\label{item:definitial}%
	The expressions $F \alphabet \words \alphabet$ and $\words \alphabet \alphabet F$ respectively define the languages
	\[
		\intro* \initial(F) \defs \big \{w \mid \exists x\ \subword w {(- \infty, x)} \in F \big\}
			\quad \text{and} \quad
		\intro* \final(F) \defs \big \{w \mid \exists y\ \subword w {(y,\infty)} \in F \big\}.
\]

\itemAP\label{item:defprefix}%
	The expressions $\negation{\negation {F} \alphabet \words \alphabet}$ and $\negation{\words \alphabet \alphabet \negation {F}}$ respectively define the languages
	\[
		\intro* \fprefix (F) \defs\ \big\{w \mid \forall x\ \subword w {(- \infty, x)} \in F \big \}
			\quad \text{and} \quad
		\intro* \fsuffix(F) \defs\ \big\{w \mid \forall y\ \subword w {(y,\infty)} \in F \big \}.
	\]

\itemAP\label{item:defcofinal}%
	The expressions \( \negation{(\words \alphabet \alphabet) \negation{(\words\alphabet \alphabet F \alphabet \words \alphabet)}} \) and \( \negation{(\negation{\words \alphabet \alphabet F\alphabet\words \alphabet})(\alphabet \words \alphabet)} \) respectively define the languages
	\[
		\intro*\cofinal(F) \defs \big\{ w \mid \forall x\exists y\exists z\ x< y < z \wedge \subword w {(y,z)} \in F \big \}, \text{ and }
	\]
	\[
		\intro*\coinitial(F) \defs \big\{ w \mid \forall z \exists x \exists y\  x<y<z \wedge  \subword w {(x,y)} \in F \big \}.
	\]

\itemAP\label{item:definfinity}%
	If $F$ is "extension closed", then the set of words containing infinitely many non-overlapping factors from $F$ can be expressed as:
\begin{enumerate}
\item $\negation{\negation{F \alphabet F}\ \fop{(\alphabet\ \negation{F \alphabet F})}}$ using only "marked" operations, or
\item $\words \alphabet \big(\cofinal(F) + \coinitial(F) \big) \words \alphabet$ using only "unrestricted concatenation".
\end{enumerate}

\itemAP\label{item:defdense}%
	Let $F$ be "extension closed". The expression $\negation{\sop{(\negation{FF})}}$ defines the set of all words containing densely many non-overlapping factors from $F$.
\end{enumerate}
\end{lemma}

\begin{proof}
We prove each item separately:
\begin{enumerate}
\item The expression $F \alphabet \words \alphabet$ defines the set of words that have some prefix in $F$, followed by a letter. This is precisely the definition of $\initial(F)$. The dual case for $\final(F)$ is symmetric.

\item Let $L \defs \negation{F} \alphabet \words \alphabet$, which includes all words that have a prefix in $\negation F$ followed by a letter. Hence, its complement consists of words where every such prefix lies in $F$, which matches $\fprefix(F)$. The dual case for $\fsuffix(F)$ is similar.

\item Suppose $w \notin \cofinal(F)$. Then there exists $x \in \dom w$ such that for all $y,z \in \dom w$ with $x < y < z$, we have $\subword w {(y,z)} \notin F$. Equivalently, $\subword w {(x,z)} \notin \words \alphabet \alphabet F$, and hence all prefix of $\subword w {(x, \infty)}$ is not in $\words \alphabet \alphabet F$. In other words, $w$ satisfies the expression $(\words \alphabet \alphabet) (\negation{\words \alphabet \alphabet F \alphabet \words \alphabet})$. 

For the other direction, consider a $w \in (\words \alphabet \alphabet) (\negation{\words \alphabet \alphabet F \alphabet \words \alphabet})$. Then, there is a position $x \in \dom w$ such that for all $z > x$, $\subword w {(x,z)} \notin \words \alphabet \alphabet F$ or equivalently for all $y$ where $z > y > x$, $\subword w {(y,z)} \notin F$. Therefore $w \notin \cofinal(F)$.

\item 
	Let $L$ be the set of all words that contain infinitely many non-overlapping factors from $F$, where $F$ is an extension-closed language. Let
	\[
	H = \negation{F \alphabet F} \fop{(\alphabet\ \negation{F \alphabet F})}.
	\]
	
	We prove that $w \in L$ if and only if $w \notin H$.
	
	\medskip
	\noindent
	\emph{($\Rightarrow$) Suppose $w \in L$.} Then every finite factorization of $w = u_1 u_2 \cdots u_k$ must include at least one factor $u_i$ that can be further factorized as $u_i = x\sigma y$ for some $\sigma \in \alphabet$ such that both $x, y$ contain a factor from $F$. Since $F$ is "extension-closed", this is equivalent to saying $x, y \in F$. Therefore, $w$ cannot be factorized in such a way that all $u_i \in \negation{(F \alphabet F)}$. Hence, $w \notin H$.

	\medskip
	\noindent
	\emph{($\Leftarrow$) Suppose $w \notin L$.} Then there exists a finite factorization $w = u_1 u_2 \cdots u_k$ such that each $u_i$ cannot be factorized into $x \sigma y$ with both $x, y \in F$. We now show that $w \in H$.

	If $u_i$ has a last letter we can refine the factorization $u_i = x_i' \sigma_i z_i$ where $z_i = \varepsilon$.

	If $u_i$ does not have a last letter, we can refine the factorization $u_i = x_i' \sigma_i z_i$ where for all factorization of $z_i = p \sigma q$, we have $p \not \in F$.

	Similarly, if $x_i'$ has a first letter we can refine the factorization $x_i' = x_i \gamma_i y_i$ where $x_i = \varepsilon$.
	If $x_i'$ does not have a first letter, we can refine the factorization $x_i' = x_i \gamma_i y_i$ where for all factorization of $x_i = q \sigma p$, we have $p \not \in F$.

	For each $i$, consider the factorization of \( u_i = x_i\ \sigma_i \ y_i\ \gamma_i\ z_i \). We have that exactly one of the following conditions holds:
	\begin{itemize}
		\item $y_i \in F$,
		\item $x_i \in F$ and for all factorizations of $x_i$ as $q \sigma p$ with $p \not \in F$,
		\item $z_i \in F$ and for all factorizations of $z_i$ as $p \sigma q$ with $p \not \in F$.
	\end{itemize}
	
	Now consider the following refinement of the original factorization:
	\[
	w = x_0\ \sigma_0 y_0\ \gamma_0 z_0x_1\ \sigma_1 y_1\ \gamma_1 z_1x_2\ \sigma_2y_2\  \cdots\ \gamma_kz_k.
	\]
	We claim that each factor (between marked letters) in this expression belongs to $\negation{(F \alphabet F)}$. Precisely, for each $i$, the words $x_0$, $z_k$, $y_i$ and $z_{i-1}x_{i}$ are in $\negation{(F \alphabet F)}$. It folows from our construction that $x_0$, $z_k$ and $y_i$ satisfies this condition. We now verify this for $z_{i-1}x_{i}$.
	
	Consider an arbitrary factorization $z_{i-1}x_{i} = p\sigma qx_i$ or $z_{i-1}x_{i} = z_{i-1}q\sigma p$. From our construction, we have $p \notin F$ in both the cases. 
	
	We conclude that $w \in H$.
	
\item The proof for dense occurrence is similar. A word contains densely many $F$-factors if every non-empty open interval includes an $F$-factor. This is captured precisely by $\negation{\sop{(\negation{FF})}}$.
\end{enumerate}
\end{proof}






\section{Algebras and "o-monoids"}
\label{section:algebra}
In this section, we first recall the algebraic framework that was introduced in \cite{CartonColcombetPuppis18} and further developed in \cite{colcombeticalp15}. We then give several interesting properties of this algebra, including different subclasses and decidable characterizations of these subclasses.
\subsection{"o-monoids", "o-algebras", "recognizability"}
\label{subsection:o-monoids}
\AP
A ""o-monoid"" $\monoid = \monoidOf \monoidset$ is a set $\monoidset$ equipped with an operation $\product:\words \monoidset \rightarrow \monoidset$, called the \intro{product}, that satisfies $\product(a)=a$ for all $a\in \monoidset$, and the \intro{generalized associativity} property: for every \kl{words} $u_i$ over $\monoidset$ with $i$ ranging over a \kl{countable linear ordering} $\alpha$,
$\product\left(\prod_{i\in\alpha}u_i\right)=\product\left(\prod_{i\in\alpha}\product(u_i)\right)$.
We reserve the notation $\unit=\product(\emptyword)$ for the \intro{identity element}.

An example of a
"o-monoid" is the free "o-monoid" $(\words \alphabet, \emptyword, \prod)$ over the
alphabet $\alphabet$ with the product being the "generalized concatenation".
Now we discuss some natural algebraic notions.
A ""morphism"" from a "o-monoid" $\monoidOf \monoidset$ to a "o-monoid" $\monoidOf['] \monoidset$ is a map $h: \monoidset \rightarrow \monoidset'$ such that, for
every $w \in \words{\monoidset}$, $h(\pi(w))=\pi'(\bar{h}(w))$ where $\bar{h}$ is the pointwise extension of $h$ to words.
We skip the notions sub-"o-monoid" and direct products since they are as expected.
We say $\monoid$ \intro{divides} $\monoid'$ if there exists a sub-"o-monoid"
$\monoid''$ of $\monoid'$ and a surjective morphism from $\monoid''$ to $\monoid$.

\AP
A "o-monoid" $\monoidOf \monoidset$
is said to be finite if \monoidset\ is so. But, even for a finite "o-monoid",
the product operation $\product$ has an infinitary description. However, $\product$ can be captured using finitely presentable ""derived operations"".
Corresponding to a "o-monoid" $\monoidOf{\monoidset}$ there is an induced ""o-algebra"" $\monoid = \algebraOf \monoidset$ where the operations are defined as following: for all $a, b \in \monoidset$,
\AP $a \intro*\productoper b = \product(ab)$, the ""omega power"" of $a$, $\intro*\omegaoper a = \product(\omegaword a)$, the ""omega* power"" of $a$, $\intro*\omegastaroper a = \product(\omegastarword a)$ and for all $\emptyset \neq E \subseteq \monoidset$, the ""shuffle power"" of $E$, $\intro*\shuffleoper E = \product(\shuffleword E)$.
For a singleton set $\{m\}$, we write $\shuffleoper m =\shuffleoper{\{m\}}$. These "derived operators" satisfy certain natural axioms (see \cite{CartonColcombetPuppis18} for details):
\begin{enumerate}
\item $(\monoidset,\unit,\productoper)$ forms a monoid with $\unit$ as identity element.
\item for all $a,b \in \monoidset$, $\omegaoper a = a \productoper \omegaoper a$, $\omegaoper{(a \productoper b)} = a\productoper \omegaoper{(b \productoper a)}$ and $\omegaoper{(a^n)} = \omegaoper{a}$ for all $n > 0$.
\item for all $a,b \in \monoidset$, $\omegastaroper a = \omegastaroper a \productoper a$, $\omegastaroper{(a \productoper b)} = \omegastaroper{(b \productoper a)} \productoper b$ and $\omegastaroper{(a^n)} = \omegastaroper{a}$ for all $n > 0$.
\item For all non-empty subset $P \subseteq \monoidset$, every element $c \in P$ and every subset $P' \subseteq P$ and every non-empty subset $P'' \subseteq \{\shuffleoper P, a \productoper \shuffleoper P, \shuffleoper P \productoper b, a \productoper \shuffleoper P \productoper b \mid a, b \in P\}$, we have $\shuffleoper P = \shuffleoper P \productoper \shuffleoper P = \shuffleoper P \productoper c \productoper \shuffleoper P = \omegaoper{(\shuffleoper P)} = \omegaoper{(\shuffleoper P \productoper c)} = \omegastaroper{(\shuffleoper P)} = \omegastaroper{(c \productoper \shuffleoper P)} = \shuffleoper{(P' \cup P'')}$.
\item $\unit = \omegaoper \unit = \omegastaroper \unit = \shuffleoper \unit$ and for all non-empty subset $P \subseteq \monoidset$, $\shuffleoper P = \shuffleoper{(P \cup \{1\})}$.
\end{enumerate}
\AP
It has been established in \cite{CartonColcombetPuppis18} that
an {\em arbitrary} finite "o-algebra" $\algebraOf \monoidset$ satisfying these
natural axioms is induced by a unique "o-monoid" $\monoidOf \monoidset$.
It is rather straightforward to define the notions of morphisms, subalgebras,
direct-products as well as division for "o-algebra".

\AP
The "o-algebras" can be used to \kl{recognize} \kl{languages}. Consider a "o-algebra"
$\monoid=\algebraOf{\monoidset}$, a map $h$ from an \kl{alphabet} $\alphabet$ to $M$ and a set $F\subseteq M$,
then $\monoid,h,F$ \intro{recognizes} the \kl{language} $L=\{u\in\words\alphabet\mid\product(h(u))\in F\}$, where $h$
has been extended implicitly into a map from $\words\alphabet$ to $\words M$ and $\product$ is the unique generalized associative product defined by the "o-algebra".
\AP Given a language $L$ recognized by a "o-algebra",
there exists a ""syntactic"" "o-algebra" $\monoid$ which is minimal in the sense that $\monoid$ divides all "o-algebras" that recognize $L$ \cite{CartonColcombetPuppis18}. It was also shown in \cite{CartonColcombetPuppis18} that a language is \mso-definable if and only if it is recognized by a finite "o-algebra".

Let us consider a few examples of "o-algebras". The examples may contain a special element called \intro[zero]{}$\zero$. It is an absorbing element and $\morphism(w) = 0$ for any $w$ that has $\zero$ somewhere. Since the operations on $\unit$ and $\zero$ are fixed we skip those in the examples.
\begin{example}
\AP
The "o-monoid" $\intro*\monoidMin$ = $\monoidOf{\{\unit,\ci,\oi\}}$ where the operations are: $\ci \productoper x = \ci$ and $\oi \productoper x = \oi$ for all $x$. The omega / omega* and shuffle power operations are: $\omegaoper \ci = \ci$ and $\omegastaroper \ci = \omegaoper \oi = \omegastaroper \oi = \oi$ and $\shuffleoper S = \oi$ for all $S \neq \{\unit\}$. For the alphabet $\alphabet = \{a,b\}$ and "morphism" $h:\words \alphabet \rightarrow \monoidMin$ that extends $h(a)=h(b) = \ci$, the tuple $\monoidMin,h,\{\ci\}$ \kl{recognizes} the \kl{language}  of words that have a first letter.
\end{example}
\begin{example}\AP
	The "o-monoid" $\intro*\monoidGap$= $\monoidOf{\{\unit,\cci, \oci, \coi, \ooi, \zero\}}$. Operations: $\shuffleoper S = \zero$ for all $S \neq \{\unit\}$ and those given below.
\begin{align*}
\begin{array}{c|cccccc|c|c}
\productoper&\,\,\cci\,&\,\coi\,&\,\oci\,&\,\ooi\,&\zero &\,\quad&\, \algomegasymb\, &\,\algomegastarsymb\\
\hline
\cci	\,&\cci	&\coi&\cci&\coi\,	&\zero &\,	&\,\coi\,&\,\oci	\\
\coi	\,&\cci&\coi	& \zero & \zero\,			&\zero &\,	&\,\coi\,&\,\ooi		\\
\oci	\,&\oci &\ooi &\oci &\ooi \,	&\zero &\,	&\,\ooi \,&\,\oci	\\
\ooi	\,& \oci &\ooi &\zero&\zero\,					&\zero &\,	&\,\zero\,&\,\zero \\
\zero & \zero & \zero & \zero & \zero & \zero &  & \zero & \zero
\end{array}
\end{align*}
Let the alphabet be $\alphabet = \{a,b\}$ and \kl{morphism} $h: \words \alphabet \mapsto \monoidGap$ extends the map $h(a)=h(b)=\ci$. Then $\monoidGap,h,\{\zero\}$ \kl{recognizes} the \kl{language} of all words that contain at least one gap. 
\end{example}
\begin{example}\AP%
	The "o-algebra" $\intro*\monoidPerfectlyDense = \monoidOf {\{\unit, s, g, \zero\}}$ where the "derived operations" are defined as follows:
	\begin{align*}
	\begin{array}{c|cc|c|c}
	\productoper&\,\,s\,&\,g\,&\,\algomegasymb\, &\,\algomegastarsymb\\
	\hline
	s	\,&g	&g&\,g\,&\,g	\\
	g	\,&g&g	&\,g\,&\,g		\\
	\end{array}
	&&
	\shuffleoper S&=\begin{cases}
				\unit&\text{if }S = \{\unit\} \\
				\zero & \text{if } s \in S \\
				g&\text{otherwise}
				\end{cases}
	\end{align*}
	For  $\alphabet = \{a,b\}$ and "morphism" $h\colon\words \alphabet \to \monoidPerfectlyDense$ that extends 	$h(a)=s$ and $h(b)=g$, the tuple $\monoidPerfectlyDense,h,\{\zero\}$ \kl{recognizes} the \kl{language}  of words that contain an $a$-labelled set $X$ such that all points $x < y < z$ with $y \in X$ are such that $(x,y)$ and $(y,x)$ both intersect $X$.
\end{example}
\begin{example}
\AP
The "o-algebra" $\intro*\monoidEven = \monoidOf {\{\unit, s, s^2, \zero\}}$ where $s^2 = s \productoper s = s^2 \productoper s^2$ and $s = s \productoper s^2 = s^2 \productoper s$. The "omega@@power", "omega*@@power" and "shuffle power" of both $s$ and $s^2$ is the $\zero$ element.  For an alphabet $\alphabet = \{a,b\}$ and a map $h(a) = s$ and $h(b) = \unit$ we have that $\monoidEven, h, \{s^2\}$ \kl{recognizes} the \kl{language} of words that contain even number of $a$'s.
\end{example}

%

\subsection{Structure theorem for "o-monoids"}
\label{subsection:green's relations}
Let us fix a finite "o-algebra" $\monoid = \algebraOf{\monoidset}$ and associated "o-monoid" $\monoidOf{\monoidset}$. For the rest of the paper, we can either view that $\monoid$ is a "o-algebra" or $\monoid$ is a "o-monoid" and the operators are derived operators. Our discussion in previous section shows that both views are consistent.

\AP
Green's relations are related to the theory of ideals of monoids.
 We define the following relations (\intro{Green's relations}) between two elements $a,b$ in $\monoid$:\AP
\phantomintro\Rleq\phantomintro\Req\phantomintro\Lleq\phantomintro\Leq\phantomintro\Jleq\phantomintro\Jeq\phantomintro\Heq\phantomintro{$\gJ$-equivalent}
\begin{align*}
	& a\reintro*\Rleq b\ \text{if }a=b\productoper x~\text{for $x\in\monoidset$},\
	a\reintro*\Req b\ \text{if }a\Rleq b\ \text{and}\  b\Rleq a,\\
	& a\reintro*\Lleq b\ \text{if }a=x\productoper b~\text{for $x\in\monoidset$},\
	a\reintro*\Leq b\ \text{if }a\Lleq b\ \text{and}\  b\Lleq a,\\
	& a\reintro*\Jleq b\  \text{if }a=x\productoper b\productoper y~\text{for $x,y\in\monoidset$},\
	a\reintro*\Jeq b\  \text{if }a\Jleq b\ \text{and}\ b\Jleq a,\\
    & \text{and}\ a\reintro*\Heq b\ \text{if }a\Leq b\ \text{and}\ a\Req b.
\end{align*}\AP
We write $a \reintro*\nReq b$ (resp. $a \reintro*\nLeq b, a \reintro*\nJeq b, a \reintro*\nHeq b$) if $a$ is not $\Req$ (resp. $\Leq, \Jeq, \Heq$) equivalent to $b$. 
\AP
We also denote by $a \intro*\Rl b$ to mean $a \Rleq b$ and $a \nReq b$ (and similarly for $a \intro*\Ll b$ and $a \intro*\Jl b$).
\AP The ""J-class"" of $b \in \monoidset$ is the set of all elements $\intro*\Jclass(b) = \{a\mid a \Jeq b\}$. Similarly, we denote by $\intro*\Rclass(b), \intro*\Lclass(b), \intro*\Hclass (b)$ the ""R-classes"", ""L-classes"" and ""H-classes"" of $b$. The best way to view the "J-classes" in a finite monoid is the ``egg-box'' view. The "R-classes" and "L-classes" form rows and columns in the "J-class".

\AP
An \intro{idempotent} is an element $e \in \monoidset$ such that $e \productoper e=e$.
Since $(\monoidset,\productoper)$ is a "monoid", the properties of monoids is carried over to "o-monoids" (see \cite{Pin86} for details).
\begin{proposition}\label{lem:monoidprop}
	The following hold for elements $a, b \in \monoidset$ and idempotents $e,f \in \monoidset$.
	\begin{enumerate}
	\item (Green's Lemma) \label{itm:green} Let $a \Req b$ and $b = a \productoper c$. Then $x \rightarrow x \productoper c$ is a bijection from $\Lclass(a)$ to $\Lclass(b)$.
	\item \label{itm:JlRisR} If $a \Jeq b$ and $a\Rleq b$, then $a \Req b$.
	\item \label{itm:location} (Location lemma) Let $s \Jeq t$. Then $s \Jeq (s \productoper t)$
			if and only if $\Rclass(t) \cap \Lclass(s)$ contains an idempotent.
	\item \label{itm:Hclass} $G$ is a group with identity $e$ if and only if $\Hclass(e) = G$. For any element $a$, $\Hclass(a) = \{a\}$ if and only if $\Jclass(a)$ do not contain a non-trivial group.
	\item \label{itm:Jfall} Let $b,t,c,a \in \monoid$,
			such that $b \productoper t \Jgeq a$, $t \productoper c \Jgeq a$
			and $b \productoper t \productoper c \nJgeq a$. Then $t \Jg a$.
	\end{enumerate}
\end{proposition}

We now state the \emph{Structure theorem} for "o-monoids".
\begin{theorem}[Structure theorem for "o-monoids"]
\label{thm:fundamental}
The following properties hold true for any "o-monoid".
\begin{enumerate}
\item 
\label{itm:omegaL} Let $e$ and $f$ be idempotents. If $e \Jeq f$, then $\omegaoper e \Leq \omegaoper f$ and $\omegastaroper e \Req \omegastaroper f$. In particular, if $e \Jeq f \Jeq \omegaoper e$ then $f \Req \omegaoper f$, and if $e \Jeq f \Jeq \omegastaroper e$ then $f \Leq \omegastaroper f$.


\item \label{itm:nogroup} Let $e$ be an idempotent. If $e \Jeq \omegaoper e$ or $e \Jeq \omegastaroper e$, then $\Jeq(e)$ do not contain any non-trivial group. In other words $\card{\Hclass(e)} = 1$.
%

\item \label{itm:shuffle} If $\shuffleoper S \Jeq \shuffleoper R$ then $\shuffleoper S = \shuffleoper R$.
\end{enumerate}
\end{theorem}
\begin{proof}
We prove the Items one by one.
\begin{enumerate}
\item Let $e$ and $f$ be idempotents such that $e \Jeq f$. From Location lemma (\Cref{lem:monoidprop}.{\Cref{itm:location}}) it follows that there exists elements $x,y$ such that $e = y \productoper x$ and $f=x \productoper y$. From the equations of "o-algebra" $\omegaoper f = \omegaoper{(x \productoper y)} = x  \productoper \omegaoper{(y \productoper x)} = x  \productoper \omegaoper{e}$. Therefore $\omegaoper f \Lleq \omegaoper e$. 
A similar derivation gives $\omegaoper e = y \productoper \omegaoper f$ and therefore $\omegaoper e \Lleq \omegaoper f$. Hence $\omegaoper e \Leq \omegaoper f$. By a symmetric derivation, we have that $\omegastaroper e \Req \omegastaroper f$. This concludes the proof of the first implication. To prove the special case assume $e \Jeq f  \Jeq \omegaoper e$. Since $\omegaoper e \Jeq \omegaoper f$, we have that $f \Jeq \omegaoper f$. Furthermore, from $\omegaoper f \Rleq f$ and \Cref{lem:monoidprop}.{\Cref{itm:JlRisR}} it follows that $f \Req \omegaoper f$.

\item Consider an idempotent $e$ such that $\omegaoper e \in \Jclass(e)$. We show by contradiction that $\Jclass(e)$ do not contain a non-trivial group. This will imply $\card{\Hclass(e)} = 1$ from \Cref{lem:monoidprop}.\Cref{itm:Hclass}. Assume $\Jclass(e)$ contains a non-trivial group. \Cref{lem:monoidprop}.{\Cref{itm:Hclass}} shows there is an idempotent $f \in \Jclass(e)$ such that $\Hclass(f)$ is a non-trivial group. Consider an element $a \neq f$ and $a \Heq f$. Since $\Hclass(f)$ is a non-trivial finite group, there is a natural number $n > 1$ such that $a^n = f$. Since $e \Jeq f \Jeq \omegaoper e$, from the previous Item  $f \Req \omegaoper f$. Therefore, there exists an element $x$ such that $f = a^n = \omegaoper f  \productoper x = \omegaoper{(a^n)}  \productoper x = a  \productoper \omegaoper{(a^n)}  \productoper x = a^{n+1} = a$. This is a contradiction. By a symmetrical argument $\omegastaroper e \in \Jclass(e)$ implies $\Hclass(e) = \{e\}$.

\item Let $S^{\shuffle} = e$ and $R^{\shuffle}=f$. From the axioms of "o-algebra", $e = \omegaoper e$ and $f = \omegaoper f$. Since $e \Jeq f$, it follows from \Cref{itm:omegaL} that $ \omegaoper e \Leq \omegaoper f$ and therefore $e \Leq f$. By a symmetric argument, $e \Req f$ and therefore $e \Heq f$. From the previous Item we have that $e = f$.
\end{enumerate}
\end{proof}

\Cref{itm:omegaL} of the Structure Theorem states that if two idempotents lie in the same "J-class", then their $\omega$-powers belong to the same "L-class", and their $\omega^*$-powers belong to the same "R-class". The most notable and structurally informative case is when an idempotent, its $\omega$-power, and its $\omega^*$-power all lie within the same "J-class". An example of such a "J-class" is illustrated in \Cref{fig:idealJclass}.

\Cref{itm:nogroup} further asserts that if a "J-class" contains an idempotent together with its $\omega$- or $\omega^*$-power, then the "J-class" cannot contain a non-trivial group (i.e., a group of size at least two). This means that all "H-classes" within such a "J-class" must be singletons.

Additionally, \Cref{itm:shuffle} establishes that there can be at most one element in a "J-class" that is a shuffle power. As depicted in \Cref{fig:idealJclass}, this shuffle power lies in a unique "H-class".

We now establish this claim and explore further structural properties of "o-monoids". These properties follow directly from the Structure Theorem and classical results on finite monoids.

 \begin{figure}[h!]
\centering
 \begin{tabular}{|p{0.5cm}|p{0.5cm}|p{0.5cm}|p{0.5cm}|}
 \hline
 $e$ & ~ & ~ & $\omegaoper e$ \\
 \hline
 ~ & ~ & ~ & ~ \\
 \hline
  ~ & $f$ & ~ & $\omegaoper f$ \\
 \hline
 ~ & & & \\
 \hline
$\omegastaroper e$ & $\omegastaroper f$& & $\shuffleoper S$ \\
 \hline
 \end{tabular}
\caption{The "omega@omega power" (resp. "omega*@omega* power") "powers@omega power" lie in the same "L-class" (resp. "R-class"). The "shuffle power" lies in a unique "H-class".}
\label{fig:idealJclass}
\end{figure}


 \begin{lemma}
 \label{lem:omegaprop}
 Let $a \Jeq \omegaoper a$. Then the following hold.
 \begin{enumerate}
 \item \label{itm:omega}  $a$ is an idempotent and $a \Req \omegaoper a$.
 \item \label{itm:omegaLclass}  "R-classes" of $\Jclass(a)$ contain an "idempotent" $f$ and its omega power $\omegaoper f$. 
 \item \label{itm:omegaidemp} There exists an idempotent $f \in \Jclass(a)$ such that $\omegaoper f$ is also an idempotent. Hence, $\omegaoper f = \omegaoper {(\omegaoper f)}$.
 \end{enumerate}
 \end{lemma}
 \begin{proof}
 The proofs are given below.
 \begin{enumerate}
 \item Since $a \Jeq \omegaoper a$ and $a \Rgeq aa \Rgeq aa \omegaoper a = \omegaoper a$ we have that $a \Req aa$. Since $a \Lgeq aa$, from \Cref{lem:monoidprop}.\Cref{itm:JlRisR}, it follows that $a \Heq aa$. Thanks to the Structure theorem (\Cref{thm:fundamental}.\Cref{itm:nogroup}), $a = aa$ and hence $a$ is an idempotent.
Since $a  \Jeq \omegaoper a$, \Cref{thm:fundamental}.\Cref{itm:omegaL} gives $a \Req \omegaoper a$.
%
%

 \item It is a property of finite monoids that, if a "J-class" contains an idempotent, all "R-classes" contain idempotents. Consider an idempotent $f$ in an "R-class". From the Structure theorem (\Cref{thm:fundamental}.\Cref{itm:omegaL}) and the fact that $\omegaoper a \in \Jclass(a)$, $f \Req \omegaoper f$ and $\omegaoper a \Leq \omegaoper f$.
 
 \item It is a property of finite monoids that, if a "J-class" contains an idempotent, all "L-classes" contain idempotents.  Hence $\Lclass(\omegaoper a)$ contains an idempotent and therefore there exists an $f$ such that $\omegaoper f$ is a idempotent. Now we show that $\omegaoper f = \omegaoper {(\omegaoper f)}$. Let $g = {\omegaoper f}$.  Since $g$ is an idempotent, from \Cref{thm:fundamental}.\Cref{itm:omegaL},  $g \Req \omegaoper g$ and $\omegaoper g \Leq \omegaoper f$.  Therefore $g \Heq \omegaoper g$ and from \Cref{thm:fundamental}.\Cref{itm:nogroup} we have that $g = \omegaoper g$.
 \end{enumerate}
 \end{proof}
 
There is a natural dual to the above lemma.

 \begin{lemma}
 \label{lem:omegaandstar}
 Let idempotent $e \Jeq \omegaoper e$ and $e \Jeq \omegastaroper e$. Then the following hold.
 \begin{enumerate}
 \item \label{itm:gie} For all idempotents $f \in \Jclass(e)$, $\omegastaroper f \productoper \omegaoper f$ is the unique element in $\Lclass(\omegaoper e) \cap \Rclass(\omegastaroper e)$.
  \item \label{itm:omegaomegastar} Let $e \Jeq (\omegaoper e \productoper \omegastaroper e)$. Then, $g = \omegastaroper e \productoper \omegaoper e$ is an idempotent where $g = \omegaoper g = \omegastaroper g$.
 \item \label{itm:scattered} Let $e = \omegaoper e = \omegastaroper e$. Then for all idempotents $f \in \Jclass(e)$, $f = \omegaoper f \productoper \omegastaroper f$ and $e = \omegastaroper f \productoper \omegaoper f$.
 \item \label{itm:shuffle-scat} Let $e = \shuffleoper S$ for some $S \subseteq \monoidset$. Then, for all idempotents $f$ in $\Jclass(e)$, we have $f = \omegaoper f \productoper \omegastaroper f$ and $e = \omegastaroper f \productoper \omegaoper f$.
 \end{enumerate}
 \end{lemma}
Let us consider an idempotent $e \Jeq \omegaoper e \Jeq \omegastaroper e$. We prove each of the Items in the Lemma one by one.
 \begin{proof}
 \begin{enumerate}
 \item Consider an idempotent $f \in \Jclass(e)$. Since $\{\omegaoper e, \omegastaroper e\} \subseteq \Jclass(e)$ from \Cref{thm:fundamental}.\Cref{itm:omegaL} $\{\omegaoper f, \omegastaroper f\} \subseteq \Jclass(e)$. We now show that $g = \omegastaroper f \productoper \omegaoper f$ is also in $\Jclass(e)$.
From Location lemma (see \Cref{lem:monoidprop}.\Cref{itm:location}) if $\Rclass(\omegaoper f) \cap \Lclass(\omegastaroper f)$ contains an idempotent, then $f \Jeq \omegastaroper f \omegaoper f$. Since $f \in \Rclass(\omegaoper f) \cap \Lclass(\omegastaroper f)$, we have that $g \in \Jclass(e)$. Since $g \Lleq \omegaoper f$ and $g \Jeq \omegaoper f$ (from \Cref{lem:monoidprop}.\Cref{itm:JlRisR}), $g \Leq \omegaoper f$. Similarly, $g \Req \omegastaroper f$. Therefore $g$ is in $\Lclass(\omegaoper f) \cap \Rclass(\omegastaroper f)$ = $\Lclass(\omegaoper e) \cap \Rclass(\omegastaroper e)$. The claim follows from the fact that $\Hclass(g) = \{g\}$ (see \Cref{thm:fundamental}.\Cref{itm:nogroup}).

  \item Let $h = \omegaoper e  \productoper \omegastaroper e$. Since $e \Jeq h$ it follows that $e \Leq h$ and $e \Req h$ and therefore $e \Heq h$. Since $\Jclass(e)$ has no non-trivial group (from Structure theorem). Hence $e = h = \omegaoper e  \productoper \omegastaroper e$.

  Now consider the element
 $g = \omegastaroper e \productoper \omegaoper e$. However  $(\omegastaroper e \productoper \omegaoper e)\productoper (\omegastaroper e \productoper \omegaoper e) = \omegastaroper e \productoper ( \omegaoper e \productoper \omegastaroper e) \productoper \omegaoper e = \omegastaroper e \productoper \omegaoper e$ showing $g$ is an idempotent.

 \item Let us assume that $e = \omegaoper e = \omegastaroper e$. From \Cref{lem:omegaprop}.\Cref{itm:omega}, $e$ is an idempotent. Consider an idempotent $f \in \Jclass(e)$. From \Cref{itm:gie}, $\omegastaroper f \productoper \omegaoper f$ and $\omegastaroper e \productoper \omegaoper e = e$ are in the same "H-class". However, due to \Cref{thm:fundamental}.\Cref{itm:nogroup} there is only one element in the "H-class" and therefore $e = \omegastaroper f \productoper \omegaoper f$. Since $e=e^2$, $\omegaoper f \productoper \omegastaroper f \Jeq f$ and therefore $f = \omegaoper f \productoper \omegastaroper f$.


 \item Since $e = \shuffleoper e$, we have $e = \omegaoper e = \omegastaroper e$. The claim follows from the above Item.
 \end{enumerate}
 \end{proof}
 \begin{example}\AP
 	The "o-algebra" $\monoidGap$ has three "J-classes", $\{\unit\}$, $\{\zero\}$, and $\{\cci, \oci, \coi, \ooi\}$.
	The eggbox representation of $\{\cci, \oci, \coi, \ooi\}$ is as follows:
	\begin{align*}
		\begin{tabular}{|p{1.5cm}|p{1.8cm}|}\hline
		$\cci$ & $\coi = \omegaoper \cci$
		\\\hline
		$\oci = \omegastaroper \cci$ & $\ooi = \omegastaroper \cci \omegaoper \cci$
		\\\hline
	 \end{tabular}
	\end{align*}
\end{example}

\subsection{"Profinite identities" characterize "varieties" of "o-monoids"}
In this section, we introduce notions of "o-monoids" that shall allow us to define the classes of "o-monoids" that we study.
We also introduce the sufficient material for showing that thse classes are "varieties".

We begin with introducing several forms of idempotents are important in understanding finite monoids. For "o-monoids" they can have several characteristics that we introduce now:
\begin{itemize}
\itemAP ""Gap insensitive idempotents"" (the set of which is $\intro*\giidemp$)
	are idempotents $e \in \monoid$ such that $\omegaoper e \productoper \omegastaroper e = e$
\itemAP ""Ordinal idempotents"" (the set of which is $\intro*\oidemp$) are idempotents $e$ such that $\omegaoper e = e$.
	The name comes from the fact that an idempotent is "ordinal@@idempotent" if and only if $\morphism(u)=e$ for all "words" $u\in \words{\{e\}}$ of non-empty ordinal (ie well-founded) domain.
\itemAP Symmetrically, ""ordinal* idempotents"" (the set of which is $\intro*\ostidemp$) are all idmpotents $e$
  such that $\omegastaroper e = e$.
\itemAP \intro{Scattered idempotents} (the set of which is $\intro*\scatidemp$) are all elements $e$ which are at the same time
	"ordinal@@idempotent" and "ordinal* idempotents". That is, $e= \omegaoper e = \omegastaroper e$.
	The name stems from the fact that an idempotent is "scattered@@idempotent" if and only if $\morphism(u)=e$ for all "words" $u\in \words{\{e\}}$ of non-empty "scattered" "domain".
\itemAP	""Shuffle idempotents"" (the set of which is $\intro*\shuffleidemp$) are all elements $e$ such that $\shuffleoper{e} = e$.
\itemAP	""Shuffle simple idempotents"" (the set of which is $\intro*\shsimp$) are all elements $e$ where for all $K\subseteq \monoidset$
 	such that $e\productoper a\productoper e=e$ for all $a\in K$, $\shuffleoper{(\{e\}\cup K)} = e$.
\end{itemize}

Note that by definition every \kl{shuffle simple idempotent} is a \kl{shuffle idempotent} and every \kl{scattered idempotent} is both an \kl{ordinal idempotent} and an \kl{ordinal* idempotent}.  
Furthermore all "scattered idempotents"~$e$ are such that $e=e\productoper e = \omegaoper e \productoper \omegastaroper e$, and hence are "gap insensitive".
We also note that since $(\{e\}^\shuffle)^\omega=(\{e\}^\shuffle)^{\omega*}=\{e\}^\shuffle$, every \kl{shuffle idempotent} is a \kl{scattered idempotent}. Since $e \productoper \omegaoper e = \omegaoper e$  for any \kl{scattered idempotent} $e$, $\omegaoper e = e \productoper \omegaoper e = \omegaoper e \productoper  \omegastaroper e \productoper \omegaoper e$ and therefore $\omegaoper e \productoper \omegastaroper e = e$ is an \kl{idempotent}. Similarly, for an \kl{ordinal idempotent} (resp. \kl{ordinal* idempotent}) $e$, $e = e \productoper \omegaoper e = e \productoper e$ is an idempotent.
The following remark is easy to see from the above discussion.
\begin{remark}\AP\label{remark:idempotents}
	$\shsimp \subseteq \shuffleidemp \subseteq \scatidemp = \oidemp \cap \ostidemp \subseteq \giidemp$
	and $\oidemp \cup \ostidemp \cup \giidemp \subseteq \idemp$.
\end{remark}

\begin{definition}\AP\label{definition:algebraic-properties}
	We consider the following properties:
	\begin{itemize}
	\itemAP $\intro*\aperiodic$ (""aperiodicity"") if for all $a\in \monoid$ there exists $n$ such that $a^n=a^{n+1}$,
	\itemAP $\intro*\eigi$ if all \kl{idempotents} are \kl{gap insensitive},
	\itemAP $\intro*\oigi$ if all \kl{ordinal idempotents} are \kl{gap insensitive},
	\itemAP $\intro*\ostigi$ if all "ordinal* idempotents" are \kl{gap insensitive},
	\itemAP $\intro*\scish$ if all \kl{scattered idempotents} are \kl{shuffle idempotent}, and
	\itemAP $\intro*\shiss$ if all \kl{shuffle idempotents} are "shuffle simple@@idempotent".
	\end{itemize}
\end{definition}

In order for languages recognized by a class of "o-monoids" to be of interest they need to be closed under boolean combinations. Complementation comes for free but for closure under union and intersection the class of "o-monoids" need to satisfy certain properties.

\begin{definition}\AP\label{def:variety}%
	A ""variety"" of "o-monoids"~$V$ is a class of "o-monoids" that is closed under sub-"o-monoids", quotients and finite Cartesian products.
\end{definition}
Our aim is to show that the classes $\aperiodic$, $\eigi$, $\oigi$, $\ostigi$, $\scish$, $\shiss$ are "varieties".
The standard way to define a class of \kl{monoids} is by using ``identities''. When a class of "o-monoids" (or some algebraic structure) satisfy an identity, then they form a variety.

\AP  An \intro{implicit operation}~$f$ of arity $k$ is a collection of maps $(f_\monoid)$
 indexed by "o-monoids" such that for all "o-monoid"~$\monoid$, $f_\monoid$ is a function $\monoid^k$ to $\monoid$ that satisfies the property that whenever there is a "o-monoid" morphism $\gamma$
 from $\monoid$ to $\monoidN$, and for all $x_1,\dots,x_k\in M$,
 \begin{align*}
f_{\monoidN}(\gamma(x_1),\dots,\gamma(x_k))=\gamma(f_{\monoid}(x_1),\dots,f_{\monoid}(x_k))\ .
 \end{align*}
When the "o-monoid" in which it is applied is clear from the context, we omit it in the notations. An example of an "implicit operation" is the idempotent power of an element.

Note that the \kl{derived operations} can naturally be seen as \kl{implicit operations}, and that operations
constructed from \kl{implicit operations} are also \kl{implicit operations}.

A ""profinite identity"" (or simply an ""identity"" from now) is an equality of the form $f(x_1,\dots,x_k)=g(x_1,\dots,x_k)$ where $f$ and $g$ are \kl{implicit operations} of arity~$k$. A "o-monoid"~$\monoid$ satisfies this \kl{identity}
if for all $a_1,\dots,a_k$ in $\monoid$, $f_\monoid(a_1,\dots,a_k)=g_\monoid(a_1,\dots,a_k)$.
\begin{lemma}\AP
	Let $S$ be the set of "o-monoids" satisfying a given set of identities. Then $S$ is a "variety" of "o-monoids".
\end{lemma}
\begin{proof}
Varieties are closed under intersection. Therefore it is sufficient to show the theorem when $S$ is a set of "o-monoids" satisfying just a single identity. Below we show that "o-monoids" that satisfy an identity is closed under sub-"o-monoids", quotients and finite Cartesian products.

The first state	ment is obvious.
For the second, assume that $\monoid$ satisfies $f(x_1,\dots,x_k)=g(x_1,\dots,x_k)$,
and that $\gamma$ is a \kl{morphism} from $\monoid$ onto $\monoidN$.
Let $b_1,\dots,b_k\in \monoidN$. Using the surjectivity assumption of $\gamma$,
there exist $a_1,\dots,a_k\in M$ such that $\gamma(a_1)=b_1$,\dots,$\gamma(a_k)=b_k$.
We now have:
\begin{align*}
	 f_\monoidN(b_1,\dots,b_k) & =f_\monoidN(\gamma(a_1),\dots,\gamma(a_k)) \\
	& = \gamma(f_\monoid(a_1,\dots,a_k)) = \gamma(g_\monoid(a_1,\dots,a_k)) \quad \text{(since $\monoid$ satisfies identity)} \\
	& =g_\monoidN(\gamma(a_1),\dots,\gamma(a_k)) =g_\monoidN(b_1,\dots,b_k)
\end{align*}
Thus $\monoidN$ satisfies $f(x_1,\dots,x_k)=g(x_1,\dots,x_k)$. For the last statement, we have:
\begin{align*}
	f_{(\monoid \times \monoidN)} \big((a_1,b_1),\dots,(a_k,b_k)\big)
					& =\big(f_\monoid(a_1,\dots,a_k), f_\monoidN(b_1,\dots,b_k)\big) \\
			& = \big(g_\monoid(a_1,\dots,a_k), g_\monoidN(b_1,\dots,b_k)\big) \quad \text{(since $\monoid$ ad $\monoidN$ satisfies identity)} \\
			& = g_{(\monoid \times \monoidN)} \big((a_1,b_1),\dots,(a_k,b_k)\big)
\end{align*}
Therefore $\monoid \times \monoidN$ satisfies the identity $f(x_1,\dots,x_k)=g(x_1,\dots,x_k)$.
\end{proof}
As a consequence of the above lemma it is sufficient to show that the classes \aperiodic, $\eigi$, $\oigi$, $\ostigi$, $\scish$, $\shiss$ satisfy a set of "identities". The following lemma shows that the notions used in our classification of "o-monoids" can be interpreted in terms of \kl{implicit operations}.
\begin{lemma}\AP\label{lemma:implicit operations}
\begin{itemize}
\itemAP there exists an \kl{implicit operation} ``$\intro*\impliciti$'' such that for all elements $a$ in a "o-monoid", $a^{\impliciti}$
 	is an \kl{idempotent} and if $e$ is an \kl{idempotent}, then $e^{i}=e$.%
		\footnote{This is the classical~$a^\omega$ operation in finite monoids, but this notation would be confusing in our case.}
\itemAP there exists an \kl{implicit operation} ``$\intro*\implicitoi$'' such that for all elements $a$ in a "o-monoid", $a^{\implicitoi}$
 	is an \kl{ordinal idempotent} and if $e$ is an \kl{ordinal idempotent}, then $e^{\implicitoi}=e$.
\itemAP there exists an \kl{implicit operation} ``$\intro*\implicitosi$'' such that for all elements $a$ in a "o-monoid", $a^{\implicitosi}$ is an \kl{ordinal* idempotent} and if $e$ is an \kl{ordinal* idempotent}, then $e^{\implicitosi}=e$.
\itemAP there exists an \kl{implicit operation} ``$\intro*\implicitsc$'' such that for all elements $a$ in a "o-monoid", $a^{\implicitsc}$
 	is a \kl{scattered idempotent} and if $e$ is a \kl{scattered idempotent}, then $e^{\implicitsc}=e$.
\itemAP there exists an \kl{implicit operation} ``$\intro*\implicitsh$'' of arity $k+1$ such that
    $\implicitsh(e,a_1,\dots,a_k)$ is a \kl{shuffle idempotent} $f$ such that $f\productoper a_n\productoper f=f$ for all~$n=1\dots k$. Furthermore, if $e$ is a \kl{shuffle idempotent} such that $e\productoper a_n\productoper e=e$ for all $n=1\dots k$,
    	then $\implicitsh(e,a_1,\dots,a_k)=e$.
\end{itemize}
\end{lemma}
\begin{proof}
\emph{Operation ``$\impliciti$'':}
Let $a$ be an element in the finite "o-monoid" $\monoid$. Consider the sequence $a^{1!}$, $a^{2!}$, \dots{}.
Since the number of elements are finite there is an $i_a$ be such that $a^{i_a!}$ is an idempotent. Among all the elements $a$ in $\monoid$, pick the $i_a!$ which is the maximum. That is, let $i = max \{i_a! \mid a \in \monoid\}$. Clearly $x^i$ is an idempotent for all $x$ in $\monoid$. It is also clear that if $e$ is an \kl{idempotent} $e^{\impliciti}=e$.


\emph{Operation ``$\implicitoi$'':}
Let $a$ be an element in a finite "o-monoid" $\monoid$. Consider the sequence defined by $a_0=a$, $a_1 = \omegaoper a_0$, $a_2 = \omegaoper a_1$, \dots{}, $a_{n+1} = \omegaoper a_n$, \dots{}. We argue that for any $a_n, a_{n+1}$ and $a_{n+2}$ one of the following holds: either $a_{n} \Jg a_{n+2}$ or $a_{n+2} = a_{n+1}$. Let us assume $a_n \Jeq a_{n+2}$. Since $a_n  \Jeq \omegaoper {a_n}  \Jeq \omegaoper {(\omegaoper a_n)}$ from \Cref{thm:fundamental}.{\Cref{itm:omegaL}} it follows that $a_{n+2} = a_{n+1}$.

Since the number of "J-classes" are finite, the sequence $\big(a_n)_{n \in \nats}$ is ultimately constant.
Let $a^{\implicitoi}$ be this limit value. Once more, as defined by an ultimately constant sequence, it is an \kl{implicit operation}.
Furthermore, if $e$ is an \kl{ordinal idempotent}, \kl{i.e.}, $\omegaoper e=e$ then clearly $e^{\implicitoi}=e$.

\emph{Operation ``$\implicitosi$'':} This is the dual of the above case.

\emph{Operation ``$\implicitsc$'':}
Let $a$ be an element in a finite \kl{o-monoid}~$\monoid$. Consider the sequence defined by~$a_0=a$
and $a_{n+1}=\omegastaroper{(a_n)}  \productoper \omegaoper{(a_n)}$. We argue that for three consecutive $a_n, a_{n+1}$ and $a_{n+2}$ either $a_{n} \Jg a_{n+2}$ or $a_{n+2} = a_{n+1}$. Let us assume $a_n \Jeq a_{n+2}$. Since $a_{n+1}$ is an idempotent (since $a_{n+1}  \Jeq \omegaoper a_{n+1}$), from \Cref{lem:omegaandstar}.\Cref{itm:gie} it follows that $a_{n+2} = a_{n+1}$.

This means that the sequence is ultimately constant.
Let $a^{\implicitsc}$ be this limit value. Once more, as defined by an ultimately constant sequence, it is an \kl{implicit operation}.
Furthermore, if $e$ is a \kl{scattered idempotent}, \kl{i.e.},  $\omegaoper e=e=\omegastaroper e$,
we clearly have $e^{\implicitsc}=e$.

\emph{Operation ``$\implicitsh$'':}
Finally, given $e,a_1,\dots,a_k$, consider the sequence defined by $e_0=e$,
and
\begin{align*}
f_n&=\omegastaroper e_n \cdot \omegaoper e_n \cdot a_1\cdot \omegastaroper e_n \cdot \omegaoper  e_n \cdot a_2 \cdots a_k\cdot \omegastaroper e_n \cdot \omegaoper e_n \ ,\\
e_{n+1}&=\shuffleoper{\{f_n\}}\ .
\end{align*}
From \Cref{thm:fundamental}.{\Cref{itm:shuffle}} either $e_{n+1} = e_n$ or $e_{n} \Jg e_{n+1}$. The sequence is therefore ultimately constant. Let $e^{\implicitsh}$ be this limit value. Furthermore, if $e$ is a \kl{shuffle idempotent} such that $e \productoper a_i \productoper e = e$ for all $i \leq n$, then $f_n = e \productoper a_1 \productoper e \productoper a_2 \productoper \dots \productoper a_k \productoper e = e$ and therefore $e^{\implicitsh} = e$.
\end{proof}

Using these "implicit operations", it is easy to recast the various properties we are interested in into "identities":
\begin{lemma}
\begin{itemize}
\item $\aperiodic$ is equivalent to $\idap$.
\item $\eigi$ is equivalent to $\ideigi$.
\item $\oigi$ is equivalent to $\idoigi$.
\item $\ostigi$ is equivalent to $\idostigi$.
\item $\scish$ is equivalent to $\idscish$.
\item $\shiss$ is equivalent to $\idshiss$.
\end{itemize}
\end{lemma}
\begin{corollary}\AP\label{lem:varities}
Let $S$ be a set of all "o-monoids" satisfying one or more of the following set of properties: $\aperiodic$,  $\eigi$ , $\oigi$, $\ostigi$, $\scish$and $\shiss$. Then $S$ is a variety of "o-monoids".
\end{corollary}

\subsection{Types of "J-classes" and their link with "identities"}
In the structural understanding of "o-monoids", it is convenient to use several properties of their "J-classes".
\AP We define a "J-class" to be: 
\begin{itemize}
\itemAP ""regular@@J-class"" if it contains an "idempotent",
\itemAP ""ordinal regular"" if it contains an "ordinal idempotent",
\itemAP ""ordinal* regular"" if it contains an "ordinal* idempotent",
\itemAP ""gap insensitive regular"" if it contains a "gap insensitive idempotent",
\itemAP ""scattered regular"" if it contains a "scattered idempotent",
\itemAP ""shuffle regular"" if it contains a "shuffle idempotent", and
\itemAP ""shuffle simple regular""  if it contains a "shuffle simple idempotent".
\end{itemize}

We give now some basic facts about these "J-classes".
\begin{lemma} \label{lem:Jclassprop}
The following properties hold for "J-classes".
	\begin{enumerate}
		\item Every "R-class" and "L-class" within a "regular J-class" contains an idempotent.
		\item Every "R-class" within an "ordinal regular" "J-class" contains an idempotent and its "omega power".
		\item Every "L-class" within an "ordinal* regular" "J-class" contains an idempotent and its "omega* power".
		\item \label{itm:gisc}A "J-class" is "gap insensitive regular" if and only if it is "scattered regular".
		\item All "idempotents" in a "scattered regular" "J-class" are "gap insensitive".
		\item "Scattered regular" "J-class" are both "ordinal@ordinal regular" and "ordinal* regular" 
		\item Each "J-class" contains at most one "scattered idempotent". 
		\item Every "shuffle regular" "J-class" is "scattered regular".
		\item Every "shuffle simple regular" "J-class" is "shuffle regular".
	\end{enumerate}
\end{lemma}
\begin{proof}
	Consider a "J-class" $J$. We prove the items in the order they are listed.
	\begin{enumerate}
		\item This is a standard property of finite monoids (see \cite{Pin86}).
		\item Let $J$ be an "ordinal regular" $\Jclass$-class. Then there exists an "idempotent" $e \in J$ such that $\omegaoper e \in J$. Let $R$ be any $\Req$-class in $J$. By item (1), $R$ contains an "idempotent" $f$. By the Structure Theorem (\Cref{thm:fundamental}.\Cref{itm:omegaL}), $R$ also contains $\omegaoper f$. Thus, every $\Req$-class in $J$ contains both an "idempotent" and its "omega power".
		\item This is the dual of item (2), applying the same reasoning to $\Leq$-classes and "omega$^*$ powers".
		\item (Forward direction) Suppose $J$ is "gap insensitive regular". Then it contains a "gap insensitive" "idempotent" $e = \omegaoper e \productoper \omegastaroper e$. By \Cref{lem:omegaandstar}.\Cref{itm:omegaomegastar}, the element $\omegastaroper e \productoper \omegaoper e$ is a "scattered idempotent", so $J$ is "scattered regular".

		(Backward direction) If $J$ is "scattered regular", then it contains a "scattered idempotent". By \Cref{remark:idempotents}, every "scattered idempotent" is also "gap insensitive", so $J$ is "gap insensitive regular".

		\item Let $J$ be a "scattered regular" $\Jclass$-class and let $e \in J$ be a "scattered idempotent". By \Cref{remark:idempotents}, "scattered idempotents" are "gap insensitive", and by \Cref{lem:omegaandstar}.\Cref{itm:omegaomegastar}, all "idempotents" in $J$ are "gap insensitive".
		\item This follows directly from \Cref{remark:idempotents}.
		\item Suppose $J$ contains two "scattered idempotents" $e$ and $f$, satisfying $e = \omegaoper e = \omegastaroper e$ and $f = \omegaoper f = \omegastaroper f$. Then by the Structure Theorem (\Cref{thm:fundamental}.\Cref{itm:omegaL}), we have $\omegaoper e \Leq \omegaoper f$ and $\omegastaroper e \Req \omegastaroper f$, so $e \Heq f$. By \Cref{thm:fundamental}.\Cref{itm:nogroup}, the "H-class" of $e$ has size 1, hence $e = f$.
		\item Let $e$ be a "shuffle idempotent" in $J$. Then by \Cref{remark:idempotents}, $e$ is "scattered", and thus $J$ is "scattered regular".
		\item Let $J$ be "shuffle simple regular". Then, again by \Cref{remark:idempotents}, it is "shuffle regular".
	\end{enumerate}
\end{proof}

We can now restate the equational properties of "o-monoids" in terms of structural properties of their "J-classes".

\begin{theorem}\AP\label{thm:jclass}
Let $\monoid$ be an "o-monoid". Then:
\begin{enumerate}
\item $\monoid$ satisfies $\foeqs$ if and only if all its "regular" "J-classes" are "shuffle simple regular".
\item $\monoid$ satisfies $\fofiniteeqs$ if and only if all its "ordinal regular" and "ordinal* regular" "J-classes" are "shuffle simple regular".
\item $\monoid$ satisfies $\focuteqs$ if and only if it is "aperiodic" and all its "scattered regular" "J-classes" are "shuffle simple regular".
\item $\monoid$ satisfies $\fofinitecuteqs$ if and only if all its "scattered regular" "J-classes" are "shuffle simple regular".
\item $\monoid$ satisfies $\foscatteredeqs$ if and only if all its "shuffle regular" "J-classes" are "shuffle simple regular".
\end{enumerate}
\end{theorem}

\begin{proof}[Proof Sketch]
We prove the theorem in reverse order, from Item (5) to Item (1).

For the forward direction of Item (5), suppose $\monoid$ satisfies $\foscatteredeqs$. Let $J$ be a "shuffle regular" "J-class", so it contains a "shuffle idempotent" $e \in \shuffleidemp$. Since $\foscatteredeqs$ holds, $e \in \shsimp$, and hence $J$ is "shuffle simple regular". Conversely, if every "shuffle regular" "J-class" is shuffle simple regular, then for any $e \in \shuffleidemp$, $\gJ(e)$ is "shuffle simple regular", so $e \in \shsimp$, and therefore $\foscatteredeqs$ holds.

The forward and backward directions for Items (4) and (3) follow similarly by using the characterizations in \Cref{lem:Jclassprop} and \Cref{remark:idempotents}, along with the appropriate implications between "shuffle simple@shuffle simple regular", "scattered", and "gap insensitive idempotents".

For Item (2), assume $\monoid$ satisfies $\fofiniteeqs$. Let $J$ be an "ordinal regular" "J-class", so it contains an "ordinal idempotent" $e \in \oidemp$. By the equations, $e \in \giidemp$. Then, by \Cref{lem:Jclassprop}.\Cref{itm:gisc}, $J$ is "scattered regular". From Item (4), it follows that $J$ is "shuffle simple regular". A symmetric argument applies for ordinal* regular "J-classes".

Conversely, suppose all "ordinal" and "ordinal* regular" "J-classes" are shuffle simple regular. Then any $e \in \oidemp$ lies in such a $J$ and hence $e \in \shsimp \subseteq \giidemp$, proving that $\oigi$ holds. A similar argument shows $\ostigi$.

For Item (1), suppose $\monoid$ satisfies $\foeqs$. Let $J$ be a "regular" "J-class". Then $J$ contains an "idempotent" $e$ such that $e \in \giidemp$, so by \Cref{lem:Jclassprop}, $J$ is "scattered regular". Then, by Item (4), $J$ is "shuffle simple regular".

Conversely, if every "regular" "J-class" is shuffle simple regular, then for any "idempotent" $e$, $\gJ(e)$ is "shuffle simple regular" and thus "gap insensitive regular". Hence $e \in \giidemp$, showing that $\eigi$ holds. The other "idempotent" equations follow similarly.

This completes the proof.
\end{proof}


\deprec{unlineWithValue}
\deprec{nulineWithValue}
\deprec{shuffleline}
\deprec{unline}
\deprec{sline}
\deprec{nuline}

\knowledgenewrobustcmd{\cuts}{\cmdkl{\mathcal{C}}}
\newrobustcmd{\cutformula}{\formula{nc}}
\newrobustcmd{\lcuts}{\prec}
\newrobustcmd{\rcuts}{\succ}

\newrobustcmd{\leftlangone}[1]{P_{#1}^1}
\newrobustcmd{\rightlangone}[1]{S_{#1}^1}
\newrobustcmd{\leftlangtwo}[1]{P_{#1}^2}
\newrobustcmd{\rightlangtwo}[1]{S_{#1}^2}

\knowledgenewrobustcmd{\formulaLeq}[1]{\cmdkl{\fformula{Leq}}_{\cmdkl{#1}}}
\knowledgenewrobustcmd{\formulaReq}[1]{\cmdkl{\fformula{Req}}_{\cmdkl{#1}}}
\knowledgenewrobustcmd{\formulaJ}[1]{\cmdkl{\fformula{Words}}_{\Jeq(#1)}}

\knowledgenewrobustcmd{\formulawords}[1]{\cmdkl{\fformula{Words}}_{\cmdkl{#1}}}

\knowledgenewrobustcmd{\formulalanguage}[1]{\cmdkl{\fformula{Language}}_{#1}}
\knowledgenewrobustcmd{\formulaomega}[1]{\cmdkl{\omega\text{-}\fformula{Words}}_{\cmdkl{#1}}}
\knowledgenewrobustcmd{\formulaomegastar}[1]{\cmdkl{\omega^*\text{-}\fformula{Words}}_{\cmdkl{#1}}}
\knowledgenewrobustcmd{\formulawordsbelow}{\cmdkl{\fformula{Words}}_{\nJg a}}
\knowledgenewrobustcmd{\formulaomegastarbelow}{\cmdkl{\omega^*\text{-}\fformula{Words}}_{\notin\mainZp}}
\knowledgenewrobustcmd{\formulaomegabelow}{\cmdkl{\omega\text{-}\fformula{Words}}_{\notin\mainZp}}
\knowledgenewrobustcmd{\formulafinwords}[1]{\cmdkl{\fformula{Fin}\text{-}\fformula{Words}}_{\cmdkl{#1}}}
\knowledgenewrobustcmd{\formulascatwords}[1]{\cmdkl{\fformula{Scat}\text{-}\fformula{Words}}_{\cmdkl{#1}}}

\knowledgenewrobustcmd{\fJfall}{\cmdkl{\fformula{Words}}_{\cmdkl{\overline Z}}}
\knowledgenewrobustcmd{\wordsJ}{\cmdkl{\fformula{Words}}_{\cmdkl{J}}}

\knowledgenewrobustcmd{\fJa}{\cmdkl{\fformula{Words}}_{\cmdkl{\Jeq a}}}
\knowledgenewrobustcmd{\fLa}{\cmdkl{\fformula{Words}}_{\cmdkl{\Leq a}}}
\knowledgenewrobustcmd{\fRa}{\cmdkl{\fformula{Words}}_{\cmdkl{\Req a}}}
\knowledgenewrobustcmd{\fHa}{\cmdkl{\fformula{Words}}_{\cmdkl{\Heq a}}}


\newrobustcmd{\foformula}[1]{\formula{#1}^{\formula{FO}}}

\knowledgenewrobustcmd{\fodefcuts}{\cmdkl{\fo\cuts}}
\knowledgenewrobustcmd{\existsfodefcut}{\cmdkl{\exists_{\fodefcuts}}}
\knowledgenewrobustcmd{\forallfodefcuts}{\cmdkl{\forall_{\fodefcuts}}}
\knowledgenewrobustcmd{\cutset}[1]{\cmdkl{\mathcal{#1}}}

\knowledge{\fo-definable cuts}[\fo-definable|\fo-definable gaps|\fo-definable cut]{notion}
\knowledge{\fo-definable cut formulas}[\fo-definable cut formula]{notion}
\knowledgenewrobustcmd{\hole}{\cmdkl{o}}

\newrobustcmd{\fWordsusv}{(L_1,L_2,L_3)\text{-}\formulawords{}}
\newrobustcmd{\subtract}[2]{#1 \backslash #2}
\newrobustcmd{\param}{p}
\newrobustcmd{\fgapright}[1]{\overrightarrow{\varphi_{#1}}}
\newrobustcmd{\fgapleft}[1]{\overleftarrow{\varphi_{#1}}}
\renewrobustcmd{\complement}[1]{\overline{#1}}

\newrobustcmd{\proofiff}{\Leftrightarrow}

\newrobustcmd{\msoiff}{\leftrightarrow}
\newrobustcmd{\msoimplies}{\rightarrow}

\newrobustcmd{\fw}[1] {\Words {#1}}
\newrobustcmd{\fwJg}[1] {\Words {\Jg #1}}
\newrobustcmd\fwJgeq[1]{\Words{{\Jgeq}#1}}

\newrobustcmd\fwJeq[1]{\Words{\Jclass #1}}

\newrobustcmd{\fwnJg}[1] {\Words {\nJg #1}}
\newrobustcmd{\fwnJgeq}[1] {\Words {\nJgeq #1}}
\newrobustcmd{\fwReq}[1] {\Words {\Rclass #1}}
\newrobustcmd{\fwLeq}[1] {\Words {\Lclass #1}}
\newrobustcmd{\fwHeq}[1] {\Words {\Hclass #1}}


\newrobustcmd\jr{\euright{\nmainZp}}
\newrobustcmd\jl {\euleft{\nmainZp}}

\section{From "o-monoids" to regular expressions}
\label{section:monoid to expression}

In this section, we establish the direction from algebra to expression in \Cref{theorem:main}.

\subsection{General structure of the translation from algebra to expressions}
\label{subsection:monoid to expression structure}

\AP The goal of this section is to establish the difficult implications of \Cref{theorem:main}, which we reformulate using the characterization via "J-classes":

\begin{lemma}[main lemma]\AP\label{lemma:monoid-to-expression-generic}\label{lem:main}\label{lemma:main}
	Let~$L$ be "recognized" by an "o-monoid" $\monoid$. Then the following holds:
	\begin{enumerate}
	\item \label{item:fo} If all "regular J-classes" in $\monoid$ are "shuffle simple@@regular", there exists a "marked star-free expression" for~$L$.
	\item \label{item:wmso} If all "ordinal@@regular" and "ordinal* regular" "J-classes" in $\monoid$ are "shuffle simple@@regular", there exists a "marked expression" for~$L$.
	\item \label{item:focut} If $\monoid$ is "aperiodic" and all its "scattered regular" "J-classes" are "shuffle simple@@regular", there exists a "power-free expression" for~$L$.
	\item \label{item:wmsocut} If all "scattered regular" "J-classes" in $\monoid$ are "shuffle simple@@regular", there exists a "scatter-free expression" for~$L$.
	\item \label{item:scat}\label{item:foscat} If all "shuffle regular" "J-classes" in $\monoid$ are "shuffle simple@@regular", there exists a "scatter expression" for~$L$.
	\end{enumerate}	
\end{lemma}

\AP \Cref{section:monoid to expression} is devoted to establishing \Cref{lemma:monoid-to-expression-generic}. We provide a unified proof that establishes all the cases in \Cref{lemma:monoid-to-expression-generic} simultaneously, depending on the properties of the "o-monoid"~$\monoid$ and the target expression class~$\expr$.

Assume we are given an "o-monoid"~$\monoid=(\monoidset,\productoper)$ that "recognizes" a "language"~$L$, and that $\monoid$ satisfies the hypothesis of one of the items in \Cref{lemma:monoid-to-expression-generic}. Our aim is to construct an "$\expr$-expression" for~$L$, where $\intro*\expr$ is the corresponding class of expressions. For example, if $\monoid$ satisfies the hypothesis of \Cref{item:fo}, we aim to construct a "marked star-free expression" recognizing~$L$.

\AP For all $S\subseteq M$, we write:
\[
	\intro*\Words S \defs \invmorphism(S) = \{u\in\words\alphabet \mid \morphism(u)\in S\}\ ,
\]
and for $a\in M$, we write $\Words a$ as shorthand for $\Words{\{a\}}$.
We also define:
\[
\Words{{\Jleq} a} := \Words{\{b\mid b\Jleq a\}}, \quad
\fwJeq a := \Words{\{b\mid b\Heq a\}}.
\]
Similarly, we define:
\[
\Words{{\nJleq} a}, \quad \Words{{\Jgeq} a}, \quad, \fwReq a, \quad \fwHeq a \dots
\]
as analogous abbreviations.

Note that:
\[
\Words S = \bigcup_{a\in S}\Words a,
\]
so it suffices to prove that $\Words a$ is "$\expr$-expressible" for each $a \in M$ in order to conclude that $L$ is "$\expr$-expressible".

\AP The core of the proof is an induction on the size of the $\gJ$-upward closure~$\intro*\mainZ\subseteq\monoidset$ of an element~$a\in M$, for which we establish the "induction hypothesis@IHJ":
\begin{quote}\AP
	""induction hypothesis@IHJ"": $\Words a$ is "$\expr$-expressible".
\end{quote}

\AP Once $a$ is fixed, let~$\intro*\mainJ$ be its "J-class", and $\intro*\mainZp:= \mainZ\setminus\mainJ$ the set of elements strictly $\gJ$-above~$a$ (see \Cref{fig:ih}).

\AP For each $b\in\mainZp$, the set $\Words b$ is "$\expr$-expressible" by the "induction hypothesis@IHJ", since the $\gJ$-upward closure of $b$ is strictly contained in~$\mainZ$.

Our goal is now, by combining these "$\expr$-expressions" for $\Words b$ (for~$b\in\mainZp$), to show that $\Words a$ is "$\expr$-expressible". This is carried out in several steps, which we now outline:

\begin{itemize}
\itemAP We first construct an "$\expr$-expression" for the complement set $\Words\nmainZ$. This is addressed in \Cref{subsection:nJgeq-expression}, and is in fact the most complex part of the proof, as it heavily depends on the nature of $\expr$. Once this is done, we directly obtain an expression for:
\[
\Words\mainJ = \complement{\Words\nmainZ} \setminus \Words\mainZp.
\]
\itemAP In the second step, we construct for all~$a\in\mainJ$ an "$\expr$-expression" for the sets $\fwReq a$ and $\fwLeq a$, which correspond to the words mapping to elements $\Req a$ and $\Leq a$, respectively. From these, we deduce an expression for $\fwHeq a$.
\itemAP In the final step, we provide an "$\expr$-expression" for $\Words a$ itself (see \Cref{subsection:equal-expression}).
\end{itemize}

Before beginning the proof, we develop the necessary tools required for constructing these expressions.

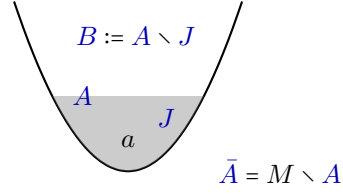
\begin{figure}[h]
\begin{center}
\begin{tikzpicture}
    \draw[thick, domain=-1.5:1.5, smooth, variable=\x] plot ({\x}, {\x*\x});
    \begin{scope}
        \clip [domain=-1.5:1.5]plot (\x, {\x*\x});
        \fill [gray, opacity=0.4] (-1,0) rectangle (1,1);
    \end{scope}
    \node at (0.1, 1.8) {$\mainZp := \mainZ \setminus \mainJ$};
    \node at (0, 0.4) {$a$};
    \node at (0.5, 0.7) {$\mainJ$};
    \node at (-0.6,1) {$\mainZ$};
    \node at (2,0) {$\nmainZ = M \setminus \mainZ$};
\end{tikzpicture}
\end{center}
\caption{\label{fig:ih}Subsets of $M$ involved in an inductive step.}
\end{figure}

\subsection{The Tools}
\AP
Our aim is to introduce the tools required for the rest of the proof.
For some set~$S\subseteq M$, we shall define
\[
\intro*\euright S \defs  \cofinal(\Words{S}),
\]
and $\intro*\euleft S$ for the symmetric version.
Note that if~$S$ is a subset of $\mainZp$ or equal to $\nmainZp$, then we already know that $\euright S$ and $\euleft S$ are ``$\expr$-expressible''.

\AP For all idempotents~$e$, we define the language $\intro*\eright e$ of words that ``look like'' an omega sequence of $e$'s:
\begin{align*}
 \reintro*\eright e \defs \cofinal(\fwJeq e)\ & \bigcap\ \fprefix(\fwJgeq e) \\
  & \bigcap\ \initial(\fwReq e)
\end{align*}
and $\intro*\eleft e$ for the symmetric version. Note also here that if $e\in\mainZp$, then we already know that $\eright e$ and $\eleft e$ are ``$\expr$-expressible''.
 
\AP
The following lemma (and its symmetric version) shows some properties of these languages.
\begin{lemma}\AP \label{lem:eomega}
	Consider an idempotent $e$.
	\begin{enumerate}
	  \item \label{item:euright} Let words $u_0, u_1, \dots$ be such that $\morphism(u_0) \Req e$, $\morphism(u_i) \Jeq e$, and $\morphism(u_0 \dots u_i) \Jeq e$ for all $i \in \nats$. Then $\prod_{i \in \nats} u_i \in \eright e$.
	  \item For all $w \in \eright e$, we have $\morphism(w) = \omegaoper e$.
	  \item $\eright e \subseteq \jr$ for all $e \notin \mainZp$.
	  \item For all $w \in \jr$, either $\morphism(w) \notin \mainZ$ or there exists $e \in \mainJ$ such that $\morphism(w) \Leq \omegaoper e$.
	\end{enumerate}
	\end{lemma}
	
	\begin{proof}
	\textbf{(1)} 
	Let $u_0, u_1, \dots$ be as in the statement, and set $w = \prod_{i \in \nats} u_i$.  
	Since $\morphism(u_i) \Jeq e$ and $\morphism(u_0 \dots u_i) \Jeq e$ for all $i$, there are infinitely many marked factors $u$ with $\morphism(u) \Jeq e$. Thus, $w$ satisfies $\cofinal(\fwJeq e)$.  
	
	Moreover, since $\morphism(u_0 \dots u_i) \Jeq e$ for each $i$, $w$ satisfies $\fprefix(\fwJgeq e)$.  
	Finally, because $\morphism(u_0) \Req e$, it follows by \Cref{lem:monoidprop}.\Cref{itm:JlRisR} that $\morphism(u_0 \dots u_i) \Req e$ for all $i$. Therefore, $w$ satisfies $\initial(\fwReq e)$, and hence $w \in \eright e$.
	
	\medskip
	\noindent \textbf{(2)}
	Let $w \in \eright e$. Since $w$ satisfies $\cofinal(\fwJeq e)$, it has no last letter.
	
	Thus, there exists an infinite set of positions $I \subseteq \dom w$ such that for every $y \in \dom w$, there is $x \in I$ with $x > y$.  
	Define a coloring $c(i,j) = \morphism(\subword{w}{[i,j)})$ for $i<j$ in $I$.  
	Since there are finitely many colors, by the infinite Ramsey theorem, there exists an infinite subset $I' \subseteq I$ and an idempotent $f$ such that $c(i,j) = f$ for all $i<j$ in $I'$.  
	(That $f$ is idempotent follows since for all $i<j<k$ in $I'$, we have $f = c(i,k) = c(i,j) \cdot c(j,k) = f \cdot f$.)
	
	We now show that $f \Jeq e$.  
	Consider any factor $u_l$ among the $u_i$'s. Since there exist $i<\dom {u_l}<j$ in $I$ with $\morphism(u_l) \Jeq e$, it follows that $e \Jgeq f$.  
	Conversely, since $w$ satisfies $\fprefix(\fwJgeq e)$, we have $\morphism(\subword{w}{(-\infty,j)}) \Jgeq e$ for all $j$, and therefore $f \Jgeq e$.  
	Hence, $f \Jeq e$.
	
	Thus, $\morphism(w) = x \productoper \omegaoper f$ for some $x$.  
	Since $w$ satisfies $\initial(\fwReq e)$, we have $x \productoper f \Rleq e$.  
	But since $x \productoper f \Jeq e$, by \Cref{lem:monoidprop}.\Cref{itm:JlRisR}, we conclude $x \productoper f \Req e$.  
	By Green's lemma (\Cref{lem:monoidprop}.\Cref{itm:green}), it follows that
	\[
	x \productoper \omegaoper f \in \Rclass(e) \cap \Lclass(\omegaoper f) = \Hclass(\omegaoper e),
	\]
	and since $\Hclass(\omegaoper e)$ is a singleton, we conclude that $\morphism(w) = \omegaoper e$.
	
	\medskip
	\noindent \textbf{(3) and (4)}
	The proofs of the third and fourth claims are similar to the arguments above and are omitted.
	\end{proof}
	
	The next lemma describes how, if a word has no maximal position, its right limit behaviour can be captured by a suffix belonging to some set $\eright e$ for an idempotent~$e$. The proof crucially uses Ramsey's theorem to find a regular pattern in the word's structure.
	\begin{lemma}\AP\label{lem:factor}
		Let $w$ be a non-empty word that does not have a last letter. Then $w$ can be factorized into $w'v$ where $w'$ has a last letter, $\morphism(w') \Lleq e$, and $v \in \eright e$ for some idempotent~$e$. Furthermore, if for all strict prefixes $u$ of $w$ we have $\morphism(u) \in\mainZ$ and $\morphism(w) \in\nmainZ$, then $e \in \mainZ$ and $\omegaoper e \Jl  e$.
	\end{lemma}
	
	\begin{proof}
	 Let $w$ be a non-empty word without a last letter. Then there exists an $\omega$-sequence of positions $I$ such that for any point $y \in \dom w$ there is a point $x \in I$ with $x > y$. For any $i<j$ in $I$, define the coloring $c(i,j)=\morphism(\subword{w}{[i,j)})$. Since the range of $c$ is finite, the infinite Ramsey theorem applies: there exists an infinite subset $I' \subseteq I$ and an element $e$ such that $c(i,j)=e$ for all $i<j$ in $I'$. 
	
	 Note that $e$ is idempotent: for all $i<j<k$ in $I'$, we have $e=c(i,k)=c(i,j) \cdot c(j,k) = e \cdot e$. 
	
	 Now, factor $w$ as $w'v$ where $w'$ ends at any position other than the first position of $I'$ and $v$ is the suffix beginning at that point. Applying \Cref{lem:eomega}.\Cref{item:euright} to $v$, we conclude that $v \in \eright e$. Also, $\morphism(w') \Lleq e$ by construction.
	
	 For the second part, assume that for all strict prefixes $u$ of $w$, $\morphism(u) \in \mainZ$, and $\morphism(w) \in \nmainZ$. Suppose for contradiction that $e \Jeq \omegaoper e$. Then, since $\morphism(w)=\morphism(w')\productoper\morphism(v)$ and $\morphism(v)=\omegaoper e$, there exists an element $x$ such that $\morphism(w)\productoper x \Jeq \morphism(w')$. But $\morphism(w')\in\mainZ$ by assumption, and $\mainZ$ is $\Jeq$-upward closed. Thus, $\morphism(w) \in\mainZ$, contradicting our hypothesis. Therefore, it must be that $e \in \mainZ$ and $\omegaoper e \Jl e$.
	\end{proof}
	
The following lemma provides a list of languages that can be expressed using “$\expr$-expressions,” by applying marked concatenations to simpler building blocks obtained through the "inductive hypothesis@IHJ".

\begin{lemma}\AP\label{lem:expr-expressible} The following languages are expressible using marked concatenations applied to $\expr$-expressions: 
	\begin{enumerate} 
		\item $\eright e$, $\words \alphabet \eright e$, $\words \alphabet \eright e \alphabet \words \alphabet$ for all idempotents $e \in \mainZp$. 
		\item $\eleft f$, $\eleft f \words \alphabet$, $\words \alphabet \alphabet \eleft f \words \alphabet$ for all idempotents $f \in \mainZp$. 
		\item $\jr$, $\jr \alphabet \words \alphabet$. 
		\item $\jl$, $\words \alphabet \alphabet \jl$. 
		\item $\eright e \eleft f$, $\words \alphabet \eright e \eleft f \words \alphabet$ for idempotents $e, f \in \mainZp$ satisfying $\omegaoper e \productoper \omegastaroper f \Jl e$. 
		\item $\eright e \eleft f$, $\words \alphabet \eright e \eleft f \words \alphabet$ for idempotents $e, f \in \mainZp$ satisfying $\omegaoper e \productoper \omegastaroper f \Jl f$. 
		\item $\eright e \jl$, $\words \alphabet \eright e \jl$ for idempotents $e \in \mainZp$. 
		\item $\jr \eleft e$, $\jr \eleft e \words \alphabet$ for idempotents $e \in \mainZp$. 
		\item $\words \alphabet \eleft f \words \alphabet$ for idempotents $f \in \mainZp$ with $\omegaoper f \Jl f$. 
		\item $\words \alphabet \eright e \words \alphabet$ for idempotents $e \in \mainZp$ with $\omegaoper e \Jl e$. 
	\end{enumerate} 
\end{lemma}

\begin{proof} The claim for Items 1–4 follows directly from the inductive hypothesis and the definitions of $\eright e$, $\eleft f$, $\jr$, and $\jl$.

We now show Item 5.
Consider $\eright e \eleft f$ where $e,f \in\mainZp$ and $\omegaoper e \productoper \omegastaroper f \nJgeq e$.

We first construct an $\expr$-expression $E$ that captures words with a prefix from $\eright e$.
Define, for $b, c, d \in \monoid$ and $\sigma \in \alphabet$, the expression
\[
E_{(d,b,\sigma,c)} \defs \negation{(\Words b \sigma)\ \negation{\Words c \alphabet \words \alphabet}}
\]
where $b \Req d \Req e$, $c \Jeq e$, and $d = b \productoper \morphism(\sigma) \productoper c$. The expression $E$ is the conjunction of all such $E_{(d,b,\sigma,c)}$ and it ensures that whenever a prefix $u\sigma$ satisfies $\morphism(u) \Req e$, it can be extended to a prefix $v = u\sigma u'$ with $\morphism(u') \Jeq e$ and $\morphism(v) \Req e$.
Thus, $E$ captures the existence of an infinite sequence of words that map to $e$. In particular, this implies the words have a prefix in $\eright e$.

Next, we construct an expression $F$ that recognizes a suffix from $\eleft f$, following the prefix from $\eright e$. Define, for suitable $b, c, d$ and $\sigma$, the expression
\[
F_{(d,b,\sigma,c)} \defs \negation{\big(\Words{\nJgeq e} \cap \negation{\words \alphabet \alphabet \Words c}\big) \sigma \Words b}
\]
where $b \Leq d \Leq f$, $c \Jeq f$, and $d = c \productoper \morphism(\sigma) \productoper b$. Taking the conjunction over all such $F_{(d,b,\sigma,c)}$ (this gives the expression $F$) ensures that beyond a point where $\morphism(u) \nJgeq e$, there is a factor to the left that satisfies $\eleft f$.

Thus, $E \cap F$ defines $\eright e \eleft f$.

To define $\words \alphabet \eright e \eleft f$, observe that by \Cref{lem:factor}, any word in this language can be decomposed as $u'\sigma u''$ where $u'' \in \eright e \eleft f$. Hence, $\words \alphabet \alphabet \eright e \eleft f$ describes it.

The proof for Item 6 (with roles of $e$ and $f$ swapped) follows symmetrically.

Items 7 and 8 follow similarly by adapting the argument to $\jl$ and $\jr$.

For Item 9, consider $\eleft f$ with $f \nJgeq \omegaoper f$. From the inductive hypothesis, $\eleft f$ is expressible. Then, $\words \alphabet \eleft f$ is the union of languages of the form $\words \alphabet \alphabet \eleft f$, $\words \alphabet \alphabet \jr \eleft f$, and $\words \alphabet \alphabet \eright e \eleft f$ for some $e \in \mainZp$.
Since $f \nJgeq \omegaoper f$, $\omegaoper e \productoper \omegastaroper f \nJgeq f$, and thus $\words \alphabet \alphabet \eright e \eleft f$ is expressible.

The argument for Item 10 is analogous.
\end{proof}

\subsection{The expression for $\Words \nmainZ$}
\label{subsection:nJgeq-expression}

\subsubsection{The basic witnesses}

Consider a word $w$ such that $\morphism(w) \notin \mainZ$. Our goal is to identify a factor of $w$ that causes it to fall below $\mainJ$. We call such a factor a \emph{$\mainJ$-witness} (or simply a \emph{witness}), formally defined as follows.

\begin{definition}[$\mainJ$-witness]
A \reintro{witness} is a word $w$ of one of the following types:
\begin{enumerate}
  \itemAP \intro{Letter witness}: $w$ is a single letter and $\morphism(w) \in \nmainZ$.
  \itemAP \intro{Concatenation witness}: $w = uv$ where $u, v \neq \epsilon$, $\morphism(u), \morphism(v) \in \mainZ$, but $\morphism(uv) \in \nmainZ$.
  \itemAP \intro{Omega witness}: $w = u_1 u_2 u_3 \dots$ with $\morphism(u_i) = e$ for all $i$, where $e \in \mainZ$ is an idempotent and $\omegaoper e \in \nmainZ$.\footnote{Note $e \neq 1$ and hence $u_i$s are necessarily non-empty.}
  \itemAP \intro{Omega$^*$ witness}: $w = \dots u_3 u_2 u_1$ with $\morphism(u_i) = e$ for all $i$, where $e \in \mainZ$ is an idempotent and $\omegastaroper e \in \nmainZ$.
  \itemAP \intro{Shuffle witness}: $w = \prod_{i \in \rationals} u_i$, where each $u_i$ is a word, and $\prod_{i \in \rationals} \morphism(u_i)$ is a perfect shuffle of some $K \subseteq \mainZ$ such that $\shuffleoper K \in \nmainZ$.
\end{enumerate}
\end{definition}

\begin{lemma}\AP\label{lem:witness}
If $\morphism(w) \in \nmainZ$, then $w$ contains a witness as a factor.
\end{lemma}

\begin{proof}
Assume $w$ does not contain any witness. We will show that $\morphism(w) \in \mainZ$.

Let $\alpha$ be the domain of $w$. Since $w$ has no letter witness, for all $i \in \alpha$, we have $\subword{w}{[i]} \in \mainZ$. Define:
\[
\mathcal{I} \defs \left\{ I \subseteq \alpha \mid I \text{ is an interval and } \morphism(\subword{w}{I}) \in \mainZ \right\}.
\]
That is $\mathcal{I}$ is the set of all intervals $I$ of $\alpha$ such that $\morphism(\subword{w}{I}) \in \mainZ$. We now identify closure properties of $\mathcal{I}$.

- \textbf{Closure under finite union of consecutive intervals:}
If $I_1, I_2, \dots, I_k \in \mathcal{I}$ are consecutive intervals, then $I_1 \cup I_2 \cup \dots \cup I_k \in \mathcal{I}$.

We prove by induction on $k$. The base case $k=1$ is trivial. Now we show the inductive step. Since $\morphism(\subword{w}{I_1}), \morphism(\subword{w}{I_2}) \in \mainZ$ and there is no concatenation witness, we have $\morphism(\subword{w}{I_1} \subword{w}{I_2}) \in \mainZ$. This reduces the number of intervals by one and hence inductive hypothesis applies.

- \textbf{Closure under omega unions:}
Let $(J_k)_{k \in \omega}$ be a sequence of consecutive intervals in $\mathcal{I}$. Then $\bigcup_{k \in \omega} J_k \in \mathcal{I}$.

To show this, we use \Cref{lem:factor} and factor $\subword{w}{J}$ into $u_0 u_1 \dots$ with $\morphism(u_0) = b$ and $\morphism(u_i) = e$ for $i \geq 1$ and some idempotent $e$. Each $u_i$ lies in a finite union of intervals in $\mathcal{I}$, hence $\morphism(u_i) \in \mainZ$. Absence of omega witnesses implies $\morphism(u_1u_2\dots) \in \mainZ$.

- \textbf{Closure under omega$^*$ unions:} Similar to the previous claim, but for a sequence $(J_k)_{k \in \omega^*}$ of intervals. Again, $\bigcup_{k \in \omega^*} J_k \in \mathcal{I}$.

- \textbf{Closure under rational shuffle unions:} Let $(J_k)_{k \in \rationals}$ be a collection of intervals in $\mathcal{I}$ such that $I = \bigcup_{k \in \rationals} J_k$ is an interval and $\morphism(\subword{w}{I})$ is a perfect shuffle word. Then $I \in \mathcal{I}$.

The above closure follows from the assumption that no shuffle witness occurs in $w$.

We now invoke Zorn’s Lemma.

\begin{claim}[Zorn's Lemma]
The poset $(\mathcal{I}, \subseteq)$ has a maximal element.
\end{claim}
\begin{proof}[Proof sketch]
We show that every chain in $\mathcal{I}$ has an upper bound. 
\begin{itemize}
\item Finite chains: Since the $\mathcal{I}$ is closed under finite union, the union of the finite chain is an upper bound.
\item Chains with common left endpoint: Let the chain be $I_1, I_2, \dots$. Consider $J_k = I_k \setminus I_{k-1}$ for all $k > 1$ and $J_1 = I_1$. Each of these $J_k$ is in $\mathcal{I}$ since $\morphism(\subword{w}{J_k}) \Jgeq \morphism(\subword{w}{I_k})$. Since $\mathcal{I}$ is closed under omega unions, we have $\bigcup_{k \in \omega} J_k \in \mathcal{I}$. Hence union of $I_k$ is an upper bound.
\item Chains with common right endpoint: This is symmetric to the case above. 
\item General chains can be split into left- and right-coinciding subchains and handled via finite union closure.
\end{itemize}
Since every chain in $\mathcal{I}$ has an upper bound, Zorn's Lemma applies and we conclude that $\mathcal{I}$ has a maximal element.
\end{proof}

Let $C \subseteq \mathcal{I}$ be the set of maximal intervals obtained via Zorn's Lemma. We show that $C = \{\alpha\}$:
\begin{itemize}
\item If $|C| = 1$, then $\morphism(w) \in \mainZ$.
\item If some $C_1, C_2 \in C$ intersect, their intersection contradicts maximality.
\item If $C_1, C_2$ are consecutive, then $C_1 \cup C_2 \in \mathcal{I}$, again contradicting maximality.
\item If $C$ forms a dense linear order, then there is a domain $C' \subseteq C$ such that
 $\prod_{I \in C'} \subword{w}{I}$ forms a shuffle word (see \cite{Shelah75}), contradicting the absence of a shuffle witness.
\end{itemize}
Hence, only the first case is possible, and we conclude $\morphism(w) \in \mainZ$.
\end{proof}

\subsubsection{Detecting witnesses using $\expr$-expressions}\label{sec:detect}

We now prove the inductive step of the completeness direction of \Cref{theorem:main}. Thanks to \Cref{lem:witness}, the presence of one of the five witnesses as factor of a word $w$ can be used for guaranteeing $\morphism(w) \not\in\mainZ$. 

We begin with detecting words that ``falls in~$\nmainZ$ at a gap''.
\begin{lemma}[\reintro{gap fall}]\AP\label{lem:gapfall}\label{lemma:gap-fall-expression}\phantomintro{gap fall}
	There exists an "$\expr$-expressible" language~$G\subseteq \Words\nmainZ$ such that
	for all non-empty "words"~$u,v$ satisfying 
	\begin{itemize}
	\item $u$ has no last letter,
	\item $v$ has no first letter, 
	\item $\morphism(u') \in \mainZ$ and $\morphism(v') \in \mainZ$ for all non-empty prefixes $u'$ of $u$ and all non-empty suffixes $v'$ of $v$, and
	\item $\morphism(pq)\in\nmainZ$ for all non-empty suffixes~$p$ of $u$, and all non-empty prefixes $q$ of~$v$.
	\end{itemize}
	then $uv\in G$.
\end{lemma}
\begin{proof} 
    Let $u$ and $v$ be non-empty words such that $u$ has no last letter and $v$ has no first letter.
    Applying \Cref{lem:factor} separately to $u$ and $v$, we obtain factorizations $u = u' p$ and $v = q v'$, where: \begin{itemize} \item $u'$ has a last letter, $p \in \eright e$ for some idempotent $e$, \item $v'$ has a first letter, $q \in \eleft f$ for some idempotent $f$, \item and $\omegaoper e \productoper \omegastaroper f \in \nmainZ$. \end{itemize} Since $u' \in \words \alphabet \alphabet$ and $v' \in \alphabet \words \alphabet$, it suffices to show that $\eright e \eleft f$ is "$\expr$-expressible".
    
    We proceed by a case analysis: 
    \begin{itemize} 
        \item \emph{Case 1:} $e, f \in \mainZp$. Then $\eright e \eleft f$ is "$\expr$-expressible" by \Cref{lem:expr-expressible}. \item \emph{Case 2:} $e \in \mainJ$ and $f \in \mainZp$. Then $\eright e \eleft f$ is contained in $\jr \eleft f$, which is "$\expr$-expressible" by \Cref{lem:expr-expressible}. 
        \item \emph{Case 3:} $e \in \mainZp$ and $f \in \mainJ$. Then $\eright e \eleft f$ is contained in $\eright e \jl$, which is "$\expr$-expressible" by \Cref{lem:expr-expressible}. 
        \item \emph{Case 4:} $e,f \in \mainJ$ and both $\omegaoper e \in \nmainZ$ and $\omegastaroper f \in \nmainZ$. Then $\mainJ$ is both "ordinal@@regular" and "ordinal* regular", meaning we are not in \Cref{item:fo} of \Cref{lemma:monoid-to-expression-generic}. Therefore, "$\expr$-expressions" can either use "marked Kleene star" or "unrestricted concatenation". In this case, $\eright e \eleft f$ is contained in $\jr \jl$, which is definable using "$\expr$-expressions" by \Cref{item:definfinity} of \Cref{lem:operations}. 
        \item \emph{Case 5:} $e, f \in \mainJ$ and either $\omegaoper e \in \mainJ$ or $\omegastaroper f \in \mainJ$. Then $\mainJ$ is either "ordinal@@regular" or "ordinal* regular", but not "gap insensitive". In particular they are not "shuffle simple@@idempotent". Hence we are not in \Cref{item:fo} or \Cref{item:wmso} of \Cref{lemma:monoid-to-expression-generic}. Hence, "$\expr$-expressions" can use "unrestricted concatenation". Again, $\eright e \eleft f$ is contained in $\jr \jl$, which is definable using inductively constructed "$\expr$-expressions". 
    \end{itemize} 
        Thus, in all cases, $\eright e \eleft f$ (and therefore the desired language $G$) is "$\expr$-expressible". 
\end{proof}

We now show how to detect "omega witnesses".
\begin{lemma}\AP\label{lem:omega}\label{lemma:omega-witness-expressino}
    There exists an "$\expr$-expressible" language~$O \subseteq \Words{\nmainZ}$ such that for all words~$w$ containing an "omega witness", we have~$w \in O$.
\end{lemma}
\begin{proof}
    Let~$w$ be a word that has an "omega witness". This means there is an $\omega$-sequence of words~$u_1, u_2, \dots$ such that $v = \prod_{i \in \omega} u_i$ is a factor of~$w$, $\morphism(v) \in \nmainZ$, and for all~$i \in \omega$, $\morphism(u_i) = e$ for some idempotent~$e \in \mainZ$ with $\omegaoper e \in \nmainZ$. From \Cref{lem:eomega}.\Cref{item:euright}, we have~$v \in \words \alphabet \eright e$.
    
    We proceed by a case analysis, each time providing an "$\expr$-expressible" language~$K$ such that~$w \in K \subseteq \Words\nmainZ$:
    \begin{itemize}
        \item If $e \in \mainZp$: 
        We set~$K \defs \words \alphabet \eright e \words \alphabet$. 
        This language is "$\expr$-expressible" by \Cref{lem:expr-expressible}, since~$e \in \mainZp$ and~$\omegaoper e \nJgeq e$. Moreover, $w \in K \subseteq \Words\nmainZ$.
    
        \item If $e \in \mainJ$ and~$v$ is a suffix of~$w$: 
        We set~$K \defs \jr$. 
        This language is "$\expr$-expressible" by \Cref{lem:expr-expressible}, and $w \in K \subseteq \Words\nmainZ$ by \Cref{lem:eomega}.
    
        \item If $e \in \mainJ$ and there is a letter to the right of~$v$: 
        We set~$K \defs \jr \alphabet \words \alphabet$. 
        This language is "$\expr$-expressible" by \Cref{lem:expr-expressible}, and $w \in K \subseteq \Words\nmainZ$ by \Cref{lem:eomega}.
    
        \item If $e \in \mainJ$ and there is no letter to the right of~$v$: 
        We set~$K \defs \words \alphabet \eright e$. 
        By \Cref{lem:expr-expressible} and \Cref{lem:gapfall}, $w$ satisfies the hypotheses and hence $w \in K \subseteq \Words\nmainZ$.
    \end{itemize}
    
    The language~$O$ is then simply the union of the finitely many languages~$K$ described above.
\end{proof}

By symmetry we also get.
\begin{lemma} \label{lem:omegastar}
There exists a "$\expr$-expressible" language $O* \subseteq \Words{\nmainZ}$ such that for all words $w$ containing an "omega witness", we have $w \in O*$.
\end{lemma}
\begin{proof}
The proof follows symmetrically by considering $\eleft e$ and the corresponding variants of \Cref{lem:eomega} and \Cref{lem:expr-expressible}.
\end{proof}

We now move toward detecting "concatenation witnesses". This requires a bit more work.

This lemma establishes a partial "$\expr$-expressibility" result needed for detecting concatenation witnesses. 

\begin{lemma}\AP\label{lemma:concatenation-witness-to-expression}
	For all \( r \in \monoid \), there exists an "$\expr$-expressible" language \( R_r \) such that:
	\begin{itemize}
		\item ($r$-compatibility) for all \( w \in R_r \), we have \( \morphism(w) \Rleq r \), and
		\item ($R_r$ witness) for every word \( w \) that has a "concatenation witness", one of the following holds:
		\begin{itemize}
			\item \( w \) contains a "gap fall" (see \Cref{lemma:gap-fall-expression}),
			\item \( w \) contains an "omega witness" (see \Cref{lemma:omega-witness-expressino}), or
			\item the "concatenation witness" of $w$ can be factorized as \( x\sigma y \) for some "words" \( x, y \) and letter \( \sigma \), such that \( \morphism(x) \in \mainZ \), \( y \in R_r \) for some \( r \in M \), and \( \morphism(x\sigma) \productoper r \in \nmainZ \).
		\end{itemize}
	\end{itemize}
\end{lemma}

\begin{proof}
\textbf{(A) Construction of \( R_r \).}
For each \( r \in M \), define \( R_r \) as the union of the following languages:
\begin{itemize}
	\item \( \words\alphabet \), if \( r = \unit \),
	\item \( \eleft e\words\alphabet \), if \( r = \omegastaroper e \) for some idempotent \( e \in \mainZp \),
	\item \( \jl \), if \( r = \omegastaroper e \) for some idempotent \( e \in \mainJ \),
	\item \( \eright e\jl \), if \( r = \omegaoper e \productoper \omegastaroper f \) for idempotents \( e \in \mainZp \), \( f \in \mainJ \),
	\item \( \eright e\words\alphabet \), if \( r = \omegaoper e \) for an idempotent \( e \in \mainZp \) with \( \omegaoper e \Jl e \),
	\item \( \eright e\eleft f\words\alphabet \), if \( r = \omegaoper e \productoper \omegastaroper f \) for idempotents \( e, f \in \mainZp \), under additional conditions: either \( e \) and \( f \) are not \( \Jeq \)-equivalent, or \( \omegaoper e \Jl e \) or \( \omegastaroper f \Jl f \).
\end{itemize}

\textbf{(B) Correctness of $r$-compatibility:}  
Let \( w \in R_r \). Then \( w \) belongs to one of the six languages given above.
\begin{itemize}
	\item If \( w \) is from a language not involving \( \jl \), then \( \morphism(w) \Rleq r \) by construction.
	\item If \( w \in \jl \), correctness follows from \Cref{lem:eomega}.
\end{itemize}

\textbf{(C) Correctness of $R_r$ witness:}  
Suppose \( w \) has a "concatenation witness". Therefore there exists two consecutive factors $w_1w_2$ of $w$ such that $\morphism(w_1), \morphism(w_2) \in \mainZ$ but $\morphism(w) \in \nmainZ$.

Let \( u \) be the longest prefix of \( w_1w_2 \) such that all strict prefixes \( u' \) of \( u \) satisfy \( \morphism(u') \in \mainZ \). Let \( v \) be the rest of \( w_1w_2 \). That is, $uv = w_1w_2$.

Note that $u$ is not empty since $w_1$ is a prefix of $u$. Furthermore, $\morphism(v) \in \mainZ$ since $v$ is a suffix of $w_2$.

We proceed by a case distinction:
\begin{itemize}
	\item If $u$ has a last letter and $\morphism(u)\in\nmainZ$. We decompose $u$ as $x\sigma$ for some letter~$\sigma$.
		We have $\morphism(x)\in\mainZ$ and hence $uv=x\sigma y$ with $y:=v\in\words\alphabet$. Since $R_{\unit} = \words \alphabet$, we have $y \in R_{\unit}$, and $\morphism(x\sigma)\productoper\unit\in\nmainZ$, as expected.
	\item If $u$ has a last letter and $\morphism(u)\in\mainZ$. We decompose~$u$ as $x\sigma$ as above. We proceed by case distinction.
	\begin{itemize}
		\item $v$ is empty: This is not possible, since $\morphism(uv)\in\nmainZ$ and $\morphism(u) \in \mainZ$.
		\item $v$ has a first letter $\sigma'$: Then $u\sigma\sigma'$ would be such that all its strict prefixes~$p$ satisfy $\morphism(p)\in\mainZ$. This is a contradiction to the maximality assumption in the definition of~$u$.
		\item Otherwise: We have, $v$ is non-empty and has no first letter. Applying \Cref{lem:factor} to~$v$, we get $v\in\eleft e\words\alphabet$ for some idempotent~$e\in\mainZ$. Again, by maximality in the definition of~$u$, we know that $\morphism(u\sigma p)\in\nmainZ$ for all non-empty prefixes~$p$ of~$v$, which implies $\morphism(u)\productoper \omegastaroper e\in\nmainZ$. Now, if $e\in\mainJ$, then we have $uv=x\sigma y$ with $y:=v\in\jl$, and $\jl \subseteq R_{\omegastaroper e}$.
		Otherwise, if $e\in\mainZp$, we have $uv = x\sigma y$ with $y:=v\in\eleft e\words\alphabet$, where $\eleft e \words \alphabet \subseteq R_{\omegastaroper e}$.
	\end{itemize}
	\item If~$u$ has no last letter and $\morphism(u)\in\nmainZ$. By \Cref{lem:factor}, there is a factorization of $u$ as $x\sigma p$, with $x, p$ "words" and $\sigma$ letter such that $p\in \eright e$ for some idempotent~$e\in\mainZ$ with $\omegaoper e \Jl e$. In particular, we have $\morphism(x\sigma)\productoper \omegaoper e\in\nmainZ$.
	
	Now, if~$e\in\mainJ$, then $p$ is an "omega witness", and the case is proved.
	Otherwise, $e\in\mainZp$, and $uv=x\sigma y$ with $y:=pv\in\eright e\words\alphabet$ and $\eright e \words \alphabet \subseteq R_{\omegaoper e}$.
	\item If~$u$ has no last letter, $\morphism(u)\in\mainZ$ and $v$ has a first letter. We decompose $v$ as $\sigma y$. Hence, all strict prefixes $u'$ of $u\sigma$ satisfy $\morphism(u') \in \mainZ$ contradicting the maximality of $u$. Therefore, this case is not possible
\item If~$u$ has no last letter, $\morphism(u)\in\mainZ$ and $v$ has no first letter. Using twice \Cref{lem:factor},
	we can write $u$ as $x\sigma p$ with $p\in \eright e$ for some idempotent $e\in\mainZ$, and $v$ as $qz$ with $q\in\eleft f$ for some idempotent $f\in\mainZ$.
	
	If~$e\in\mainJ$, then $\morphism(p)\Leq \omegaoper e$, and hence $\morphism(pq)\Jeq  \omegaoper e\productoper \omegastaroper f$. If $\omegaoper e \productoper \omegastaroper f \in \nmainZ$, then we are in the case of a "gap fall" witnessed by $pq$. Otherwise, $\omegaoper e \productoper \omegastaroper f \in \mainZ$ and hence $\morphism(pq) \in \mainZ$. This contradicts the maximality condition of $u$, since all prefixes $u'$ of $uq$ satisfies the condition $\morphism(u') \in \mainZ$. 
	
	Otherwise, $e\in\mainZp$. If~$f\in\mainZp$, then we have $w = x\sigma y$ with $y:=pqz\in \eright e\eleft f\words\alphabet$. Furthermore, $\eright e\eleft f\words\alphabet\subseteq R_{\omegaoper e \productoper \omegastaroper f}$ and ~$\morphism(x\sigma)\productoper \omegaoper e\productoper \omegastaroper f = \morphism(x\sigma p q)\in \nmainZ$. 
	The last case is when $f\in\mainJ$, for which we have $w = x\sigma y$ with $y:=pqz\in \eright e\jl$. Moreover, $\eright e \jl \subseteq R_{\omegaoper e \productoper \omegastaroper f}$ and ~$\morphism(x\sigma)\productoper \omegaoper e\productoper \omegastaroper f = \morphism(x\sigma pq)\in \nmainZ$.
\end{itemize}

\textbf{(D) "$\expr$-expressibility" of \( R_r \):}  
All cases are "$\expr$-expressible" by \Cref{lem:expr-expressible}, except for the last item. In this case \( R_r = \eright e\eleft f\words\alphabet \) and \( e \Jeq f \). Since $\omegaoper e\Jl e$ or $\omegastaroper f\Jl f$, we know that $\Jclass(e)$ is "regular@@J-class" but not "scattered regular". 	Let us now show the "$\expr$-expressibility" of~$\eright e\eleft f\words\alphabet$ following the case distinction made in \Cref{lemma:monoid-to-expression-generic}:
\begin{enumerate}
\item (\Cref{item:fo}) This is not possible since $\Jclass(e)$ is "regular@@J-class", and hence "shuffle regular", which implies "scattered regular". A contradiction.
\item (\Cref{item:wmso}) We first remark that $\Jclass(e)$ is neither "ordinal*@@regular" and nor "ordinal regular". Indeed, if it would be "ordinal regular", then it would be "shuffle regular", and hence "scattered regular", which is a contradiction. The same goes if $\Jclass(e)$ is "ordinal* regular". Hence, we have $\omegaoper e\Jl e$ and $\omegastaroper f\Jl f$.
	This means by \Cref{lem:expr-expressible} that $\eright e\eleft f\words\alphabet$ is "$\expr$-expressible".
\item (\Cref{item:focut,item:wmsocut,item:scat}) In this case, "unrestricted concatenation" is allowed, and the "$\expr$-expressibility" of $\eright e\eleft f\words\alphabet$ follows.
\end{enumerate}

This completes the proof.
\end{proof}

Note that there is also a symmetric version of this lemma, involving languages~$L_\ell$.

\begin{lemma}\AP\label{lemma:Ll}
	For all \( \ell \in \monoid \), there exists an "$\expr$-expressible" language \( L_\ell \) such that:
	\begin{itemize}
		\item for all \( w \in L_\ell \), we have \( \morphism(w) \Lleq \ell \), and
		\item for every word \( w \) that has a "concatenation witness", one of the following holds:
		\begin{itemize}
			\item \( w \) contains a "gap fall" (see \Cref{lemma:gap-fall-expression}),
			\item \( w \) contains an "omega$^*$ witness" (see \Cref{lem:omegastar}), or
			\item \( w \) can be factorized as \( x \sigma y \) for some "words" \( x, y \) and letter \( \sigma \), such that \( \morphism(y) \in \mainZ \), \( x \in L_\ell \) for some \( \ell \in M \), and \( \ell \productoper \morphism(\sigma y) \in \nmainZ \).
		\end{itemize}
	\end{itemize}
\end{lemma}

Our goal is to show that the concatenation witness of a word $w$ can be expressed as a union of "$\expr$-expressible" languages. We can combine the previous two Lemmas.

\begin{lemma} \label{lem:concat}
	For every word \( w \) that has a "concatenation witness", one of the following holds:
	\begin{itemize}
		\item \( w \) contains a "gap fall",
		\item \( w \) contains an "omega witness", or
		\item \( w \) contains an "omega$^*$ witness", or
		\item \( w \) can be factorized as \(x \sigma y \) for some "words" \( x, y \) and letter $\sigma$, such that \( x \in L_\ell \), \(y \in R_r \) for some \(\ell, r \in M\) where \(\ell \productoper \morphism(\sigma) \productoper r \in \nmainZ \), or 
		\item \( w \) can be factorized as \(x \sigma' z \sigma y \) for some "words" \( x, z, y \) and letters \(\sigma', \sigma \), such that \(\morphism(z) \in \mainZp\), \(x \in L_\ell \), and \(y \in R_r \) for some \( \ell, r \in M \) where \(\ell \productoper \morphism(\sigma' z \sigma) \productoper r \in \nmainZ \) 
	\end{itemize}
	Furthermore, there exists a "$\expr$-expressible" language $C \subseteq \Words{\nmainZ}$ such that for all non-empty words $w$ containing a "concatenation witness", we have $w \in C$.
\end{lemma}
\begin{proof}
	Let $w$ have a "concatenation witness" $uv$ where $\morphism(u) \in \mainZ$, $\morphism(v) \in \mainZ$ and $\morphism(uv) \in \nmainZ$. 

	If $u \in \eright e$ and $v \in \eleft f$ for some idempotents $e, f \in \mainZ$, then $uv$ is a "gap fall" and hence $uv \in G$ by \Cref{lem:gapfall} where $G$ is recognized by an "$\expr$-expression".

	Let us also assume that $w$ do not contain an "omega witness" or an "omega$^*$ witness".

	Hence, we can apply \Cref{lemma:concatenation-witness-to-expression} on $w$ and obtain a factorization $x' \sigma y$ of $uv$ where $x', y$ are words and $\sigma$ is a letter such that $\morphism(x') \in \mainZ$, $y \in R_r$ for some $r$ such that $\morphism(x'\sigma) \productoper r \in \nmainZ$.

	By applying a further factorization using \Cref{lemma:Ll}, we get that $uv$ can be factorized as one of the following:
	\begin{itemize}
	 \item $x \sigma' z \sigma y$ for $\sigma'$ and $\sigma$ letters such that $x \in L_\ell$, $\morphism(z) \in \mainZ$ and $y \in R_r$ for some $\ell, r$ such that $\ell \productoper \morphism(\sigma' z \sigma) \productoper r \in \nmainZ$,
	 \item or $x \sigma y$ for $\sigma$ letter such that $x \in L_\ell$ and $y \in R_r$ for some $\ell, r$ such that $\ell \productoper \morphism(\sigma) \productoper r \in \nmainZ$.
	\end{itemize}
   
	We argue now that in the former case, $\morphism(u) \in \mainZp$. This follows from \Cref{lem:monoidprop}.\Cref{itm:Jfall}. 
	Hence, "concatenation witnesses" belong to the union of "$\expr$-expressible" languages
	\begin{align*}
		& L_\ell \sigma' \Words{b} \sigma R_r, \quad & \text{ where } b \in \mainZp \text{ and } \ell \productoper \morphism(\sigma') \productoper b \productoper \morphism(\sigma) \productoper r \in \nmainZ, \text{ and} \\
		& L_\ell \sigma R_r, \quad & \text{ where } \ell \productoper \morphism(\sigma) \productoper r \in \nmainZ
	\end{align*}
	Thus, we can conclude that there is an "$\expr$-expressible" language $C \subseteq \Words{\nmainZ}$ as desired.
\end{proof}
   
\begin{lemma} \label{lem:shuffle}
	There exists an "$\expr$-expressible" language $S \subseteq \Words{\nmainZ}$ such that for every non-empty word $w$ with a shuffle witness, we have $w \in S$.
\end{lemma}	
\begin{proof}
	Suppose $w$ admits a shuffle witness $x$. Then there exists a factor $\prod_{i \in \rationals} u_i$, where each $u_i$ is a word and the image $\prod_{i \in \rationals} \morphism(u_i)$ is a shuffle of some set $K \subseteq \mainZ$, such that $\shuffleoper K \in \nmainZ$.
	
	We construct an "$\expr$-expression" for $w$ via a case analysis based on the properties of the "J-class" $\mainJ$:
	\begin{itemize}
		\item \textbf{$\mainJ$ is not regular:}  
		Then no idempotent $e \in \mainJ$ exists. The shuffle witness $x$ can be factorized as $uyv$ where $\morphism(y) \in K$ and $u \in \eright e$ for some idempotent $e \in \nmainZ$. 

		If $y$ has a first letter $\sigma$, we have that $w \in \jr \sigma \words \alphabet$.

		If $y$ does not have a first letter, $y$ can be factorized by \Cref{lem:factor} as $q v$ where $q \in \eleft f$ for some idempotent $f \in \mainZ$. However, since $\mainJ$ is not regular, $f \in \mainZp$. Thus, $w \in \jr \eleft f \words \alphabet$.

		Each of the above languages are "$\expr$-expressible" by \Cref{lem:expr-expressible}.
		
		\item \textbf{$\mainJ$ is regular but neither ordinal regular nor ordinal* regular:}  Hence, we are not in \Cref{item:fo} of \Cref{lem:main} and hence we can use either "marked concatenation" or "unrestricted concatenation" in our "$\expr$-expressions".

		From our assumption, we have that $\mainJ$ do not contain any omega power or omega* power of an idempotent in $J$. Therefore, any word $u$ that contains infinitely many factors from $\nmainZp$ is such that $\morphism(u) \in \nmainZ$. This language, $\fInfinity(\nmainZp)$, is "$\expr$-expressible" as shown in \Cref{lem:operations}.\Cref{item:definfinity}.
		
		Since $w$ contains a shuffle witness, it is contained in the language $\fInfinity(\nmainZp)$.

		\item \textbf{$\mainJ$ is ordinal and ordinal* regular but not gap-insensitive:}  
		Hence, we are not in \Cref{item:fo} or \Cref{item:wmso} of \Cref{lem:main}. So unrestricted concatenation is permitted in our "$\expr$-expressions". The shuffle witness $x$ can be factorized as $pq$ where $p \in \jr$ and $q \in \jl$. The expression $\jr \jl$ is "$\expr$-expressible" using unrestricted concatenation, as shown in \Cref{lem:expr-expressible}.
		
		\item \textbf{$\mainJ$ is scattered regular but not shuffle regular:}  
		We are in the setting of \Cref{item:foscat} of \Cref{lem:main} and can use the scatter operation. Since $w$ has a shuffle witness, it contains densely many non-overlapping factors whose morphism images lie outside $\mainZp$. This language is "$\expr$-expressible" using the scatter operation, as given in \Cref{lem:operations}.\Cref{item:defdense}.
		
		\item \textbf{$\mainJ$ is shuffle regular:}  
		This is the only remaining case. Since $\mainJ$ is shuffle simple regular, only marked concatenation is allowed. Let $e \in \mainJ$ be a shuffle simple idempotent such that $\shuffleoper e \in \mainJ$. We use the following claim:
		
		\begin{claim} \label{lem:shufflesimple}
			There exists $d \in K$ such that $e \cdot d \cdot e \in \nmainZ$.  
			If additionally $d \Jeq e$, then either $e \cdot d \in \nmainZ$ or $d \cdot e \in \nmainZ$.
		\end{claim}
		
		\begin{proof}
			Assume for contradiction that $e \cdot d \cdot e \Jeq e$ for all $d \in K$. By "shuffle simple regular" assumption $\shuffleoper{\{\{e\} \cup K\}} = e$. From the equations of "o-monoids", $e = \shuffleoper{\{\shuffleoper K \cup \{e\}\}}$ and hence $\shuffleoper K \Jgeq e$ which is a contradiction. Furthermore, by \Cref{lem:monoidprop}.\Cref{itm:Jfall}, if $d \Jeq e$ and $e \cdot d \cdot e \nJgeq e$, then either $e \cdot d \nJgeq e$ or $d \cdot e \nJgeq e$.
		\end{proof}

		Since $x$ is a shuffle witness, we can find a factorization of $x = u s v$ where $u \in \jr$, $v \in \jl$, and $\morphism(s) = d$ for a $d \in K$ given by the above claim. 
		
		We distinguish two subcases:
		\begin{itemize}
			\item \textbf{$d \in \mainJ$:}  
			Suppose $d \cdot e \in \nmainZ$. Then $d \cdot \omegastaroper e \in \nmainZ$. Consider the word $z = s \omegastarword e$. Since $d, e \in \mainZ$, we have that $z$ does not contain a omega or omega* witness. 

			If $z$ has a gap fall of the form $\eright f \jl$ for some idempotent $f \in \mainZp$, then the word $s v$ is also in the "$\expr$-expressible" gap fall language $G$ (since $v \in \jl$) as shown in \Cref{lem:gapfall}. 

			Note that $z$ cannot have a gap fall of the form $\jr \jl$ since $\mainJ$ is shuffle regular and $\omegaoper f \omegastaroper e \in \mainJ$ for idempotent $f \in \mainJ$.

			Otherwise, $z$ has a concatenation witness, given by \Cref{lem:concat}, of the form either $p \sigma' w' \sigma q$ or $p \sigma q$ where $\sigma', \sigma$ are letters, and $p \in R_\ell$, $q \in R_r$ for some $\ell, r \in M$ where $R_r = \jl$ or $R_r = \eright f \jl$ for some idempotent $f \in \mainZp$. Furthermore, there is an "$\expr$-expressible" language $C \subseteq \Words{\nmainZ}$ such that $R_\ell \sigma' w' \sigma R_r \subseteq C$ and $R_\ell \sigma R_r \subseteq C$. Therefore, the word $s u$ is also in $C$.

			\item \textbf{$d \in \mainZp$:}  
			Here $e \cdot d \cdot e \in \nmainZ$. Consider the factorization $usv$. The word $s$ either has a first/last letter, or has a prefix/suffix in $\eleft f$ or $\eright e$ for idempotents $e, f \in \mainZp$. In all cases, this pattern is "$\expr$-expressible" by \Cref{lem:expr-expressible}.
		\end{itemize}
	\end{itemize}
	
	All possible cases are covered, which concludes the proof.
\end{proof}

We can now show that the languages $\Words{\nmainZ}$ and $\Words{\mainJ}$ are "$\expr$-expressible" by recognizing all the witnesses given in \Cref{lem:witness}.
	
\begin{lemma}\label{lem:nJgeqa}
	There are "$\expr$-expressions" for~$\Words{\nmainZ}$ and $\Words{\mainJ}$.
\end{lemma}
\begin{proof}
	We first construct the expression for $\Words{\nmainZ}$. Consider a word $w$ where $\morphism(w) \in \nmainZ$. Thanks to \Cref{lem:witness} it has one of the five witnesses - "letter witness", "concatenation witness", "omega witness", "omega$^*$ witness", "shuffle witness". Hence, $w$ belongs to the language $L$ defined as the union of:
	\begin{itemize}
		\item languages~$\words\alphabet \sigma\words\alphabet$ for a letter~$\sigma$ such that $\morphism(\sigma)\in\nmainZ$,
		\item the language~$O$ from \Cref{lem:omega},
		\item the language~$O*$ from \Cref{lem:omegastar},
		\item the language~$C$ from \Cref{lem:concat},
		\item the language~$S$ from \Cref{lem:shuffle}.
	\end{itemize}
	Clearly~$L$ is "$\expr$-expressible" and it defines the language $\Words{\nmainZ}$. The language $\Words{\mainJ} = \Words{\nmainZp} \backslash \Words{\nmainZ}$ is also therefore "$\expr$-expressible".
\end{proof}
	
\subsection{The expression for $\morphism(w) \Heq a$}
\label{subsection:Heq-expression}
The goal of this section is to show that the language $\fwHeq a$ is $\expr$-expressible. Our first step is to characterize the words $w$ such that $\morphism(w) \Req a$. The corresponding decomposition (and its symmetric variant for $\Leq$) forms the basis for constructing $\expr$-expressions for $\fwReq a$ and $\fwHeq a$.

\begin{lemma}\AP\label{lem:Rclass}
Let $w \in \nonemptywords \monoidset$. Then $\morphism(w) \Req a$ if and only if $\morphism(w) \in \mainJ$ and $w$ can be factorized as $x\sigma v$, where $\sigma$ is a letter, and there exist an element $r$ and a language $R_r$ (as defined in \Cref{lemma:concatenation-witness-to-expression}) such that:
\begin{enumerate}
  \item $\morphism(x\sigma) \productoper r \Rleq a$, \label{item:Rclass-sv}
  \item $x = \epsilon$ or $\morphism(x) \in \mainZp$, \label{item:Rclass-s}
  \item $v \in R_r$. \label{item:Rclass-v}
\end{enumerate}
\end{lemma}

\begin{proof}
\textbf{(Right to Left)} From item~\ref{item:Rclass-v}, we have $\morphism(v) \Rleq r$. Thus,
\[
\morphism(w) = \morphism(x\sigma v) \Rleq \morphism(x \sigma) \productoper r \Rleq a.
\]
Since $\morphism(w) \in \mainJ$, we conclude $\morphism(w) \Req a$ using \Cref{lem:monoidprop}.\ref{itm:JlRisR}.

\textbf{(Left to Right)} Assume $\morphism(w) \Req a$. We aim to find a decomposition $w = x\sigma y$ satisfying the conditions of the lemma.

Let $u$ be the longest prefix of $w$ such that every strict prefix $p$ of $u$ satisfies $\morphism(p) \in \mainZp$. Let $v$ be the suffix such that $w = uv$. We distinguish several cases:

\begin{itemize}
\item \emph{$u$ has a last letter and $\morphism(u) \in \mainJ$.}  
Then $u = x\sigma$ for some $x$, $\sigma$ and $\morphism(x\sigma) \Req a$ (since $\morphism(w) \Req a$). Let $y = v$ and choose $R_r = \words \alphabet$ with $r = 1$. All conditions are satisfied.

\item \emph{$u$ has a last letter and $\morphism(u) \in \mainZp$.}  
Then $v$ must be nonempty but have no first letter (else $u$ can be extended). So $v$ is in $\eleft e \words \alphabet$ for some idempotent $e \in \mainZ$ with $\morphism(v) \Rleq \omegastaroper e$.
Set $u = x\sigma$.

If $e \in \mainZp$, define $R_r = \eleft e \words \alphabet$ with $r = \omegastaroper e$.

Otherwise $e \in \mainJ$, and we take $R_r = \euleft \nmainZp$ with $r = \omegastaroper e$.

In both cases, the conditions of the lemma are satisfied.

\item \emph{$u$ has no last letter and $\morphism(u) \in \mainJ$.}  
By \Cref{lem:factor}, $u = x\sigma u'$ where $u' \in \eright e$ for some idempotent $e \in \mainZ$. Since $\morphism(p) \in \mainZp$ for all prefixes $p$ of $u$, we conclude $e \in \mainZp$. Moreover, $\morphism(u) \in \mainJ$ implies $\omegaoper e \Jl  e$. Take $R_r = \eright e \words \alphabet$ with $r = \omegaoper e$, and let $y = u'v$. Again, the decomposition satisfies the lemma.

\item \emph{$u$ has no last letter and $\morphism(u) \in \mainZp$.}  
Then $v$ is nonempty and has no first letter. So $v$ begins with a word in $\eleft f \words \alphabet$ for some idempotent $f \in \mainZ$ such that $\morphism(v) \Rleq \omegastaroper f$. By \Cref{lem:factor}, $u = x\sigma u'$ with $u' \in \eright e$ for some $e \in \mainZp$.

If $f \in \mainZp$, let $R_r = \eright e \eleft f \words \alphabet$ with $r = \omegaoper e \productoper \omegastaroper f$.

If $f \in \mainJ$, let $R_r = \eright e \euleft \nmainZp$ with $r = \omegaoper e \productoper \omegastaroper f$.

Set $y = u'v$. The lemma's conditions are fulfilled.
\end{itemize}
\end{proof}

\begin{lemma}\AP\label{lem:fHa}
The languages $\fwReq a$, $\fwLeq a$, and $\fwHeq a$ are $\expr$-expressible.
\end{lemma}

\begin{proof}
	From \Cref{lem:Rclass}, we have that $\fwReq a$ is the intersection of the language $\Words{\mainJ}$ with the union of languages of the form
	\[
		\fw b\ \sigma\ R_r \words \alphabet,
	\]
where $b \in \mainZp$, $\sigma$ is a letter, $R_r$ is $\expr$-expressible, and $b \productoper \morphism(\sigma) \productoper r \Rleq a$. We show this can be defined by an "$\expr$-expression".

- From \Cref{lem:nJgeqa}, we have $\expr$-expressions for $\Words \mainJ$.

- By inductive hypothesis, $\fw b$ is $\expr$-expressible for all $b \in \mainZp$.

- From \Cref{lemma:concatenation-witness-to-expression}, $R_r \words \alphabet$ is $\expr$-expressible.

Therefore, $\fwReq a$ is $\expr$-expressible.

By symmetry, the same reasoning shows that $\fwLeq a$ is $\expr$-expressible. Since $\fwHeq a = \fwLeq a \cap \fwReq a$, we conclude that $\fwHeq a$ is also $\expr$-expressible.
\end{proof}

\subsection{The expression for $\Words a$}
\label{subsection:equal-expression}
In the previous subsection, we showed that $\fwHeq a$ is $\expr$-expressible. When $\mainJ$ is aperiodic, $\Hclass(a)$ is a singleton set with only the element $a$, and hence $\Words a = \fwHeq a$. In this case, the construction of an $\expr$-expression for $\Words a$ follows directly.

Our interest now lies in the case where $\mainJ$ contains a non-trivial group. Therefore, we are not in \Cref{item:fo} or \Cref{item:focut} of \Cref{lem:main}, and we are allowed to use "marked Kleene star" and "marked concatenation" in our expressions.

From \Cref{thm:fundamental}, it follows that if $\mainJ$ contains a non-trivial group, then there is no $e \in \mainJ$ such that $\omegaoper e \in \mainJ$ or $\omegastaroper e \in \mainJ$. Thus, $\mainJ$ behaves like a "J-class" of a finite monoid, and the proof strategy for constructing an expression parallels that of the finite case. For completeness, we now state the relevant lemma and  the proof.

\begin{lemma}
	\label{lem:groupwitness}
	Assume $\mainJ$ contains a non-trivial group. Let $w$ be a word such that $\morphism(w) \in \mainJ$. Then, there exists an integer $k$ such that $w$ can be factored as
	\[
	w = x_0 \sigma_1 x_1 \sigma_2 \dots \sigma_k x_k
	\]
	where $\sigma_i \in \alphabet$, and each factor $x_i$ satisfies one of the following:
	\begin{itemize}
		\item $\morphism(x_i) \in \mainZp$, or
		\item $x_i \in \eright e$ or $x_i \in \eleft e$ for some idempotent $e \in \mainZp$, or
		\item $x_i \in \eright e \eleft f$ for some idempotents $e, f \in \mainZp$ such that $\omegaoper e \productoper \omegastaroper f \Jl e$ or $\omegaoper e \productoper \omegastaroper f \Jl f$.
	\end{itemize}
\end{lemma}
\begin{proof}
	Assume $\mainJ$ contains a non-trivial group. Let $w$ be a word such that $\morphism(w) \in \mainJ$. We construct the factorization inductively. Suppose we already have a factorization of the form $(x_0, \sigma_1, \dots, x_{j-1}, \sigma_j)$ for some prefix of $w$. Let $w'$ denote the remaining suffix, so that $w = x_0 \sigma_1 \dots x_{j-1} \sigma_j w'$.

Let $u$ be the longest prefix of $w'$ such that every strict prefix $p$ of $u$ satisfies $\morphism(p) \in \mainZp$. Let $v$ be the corresponding suffix such that $w' = uv$. We consider the following cases:
\begin{itemize}
    \item \textbf{$u$ has a last letter:} Let $u = z \sigma$. Then $\morphism(z) \in \mainZp$ and we can extend the factorization to $(x_0, \sigma_1, \dots, x_{j-1}, \sigma_j, z, \sigma)$.

    \item \textbf{$v$ has a first letter $\sigma$ and $\morphism(u) \in \mainZp$:} In this case, every prefix of $u\sigma$ belongs to $\mainZp$, contradicting the maximality of $u$. Hence, this case is not possible.

    \item \textbf{$v$ has a first letter $\sigma$, $u$ does not have a last letter, and $\morphism(v) \in \mainJ$:} By \Cref{lem:factor}, we can factor $u = z \sigma' p$ where $\morphism(z) \in \mainZp$, $\sigma'$ is a letter, and $p \in \eright e$ for some $e \in \mainZp$. We extend the factorization to $(x_0, \sigma_1,\dots,x_{j-1},\sigma_{j}, z, \sigma', p, \sigma)$.

    \item \textbf{$v$ does not have a first letter and $u$ does not have a last letter:} Apply \Cref{lem:factor} to factor $u = z \sigma p$ with $\morphism(z) \in \mainZp$ and $p \in \eright e$ for $e \in \mainZp$. Similarly, factor $v = q \sigma' y$ where $\sigma'$ is a letter and $q \in \eleft f$ for an $f \in \mainZ$. Since \Cref{thm:fundamental}.\Cref{itm:nogroup} rules out idempotents $f \in \mainJ$ with $\omegaoper f \in \mainJ$, we must have $f \in \mainZp$. We then extend the factorization with $(z, \sigma, pq, \sigma')$.
\end{itemize}
In each case, we successfully extend the factorization. It remains to show that this process terminates in finitely many steps. Observe that for all $j$, we have $\morphism(u_j \sigma_{j+1} u_{j+1}) \in \mainJ$.

Suppose the process continues for $\omega$ steps. Then there exists an infinite sequence of such factors. By \Cref{lem:factor}, this would imply that $w$ has a factor in $\eright e$ for some $e \in \mainJ$. However, by \Cref{thm:fundamental}.\ref{itm:nogroup}, we have $\omegaoper e \Jl e$, contradicting the assumption that $\morphism(w) \in \mainJ$. Hence, the factorization must terminate in finite number of steps.
\end{proof}

\begin{lemma}
\label{lem:fa}
There exists an $\expr$-expression defining the language $\Words a$.
\end{lemma}
\begin{proof}
	If $\monoid$ is aperiodic, then it contains no non-trivial group, and we have $\Words a = \fwHeq a$. From \Cref{lem:fHa}, it follows that $\Words a$ is $\expr$-expressible.

Assume now that $\monoid$ is not aperiodic. Then, \Cref{item:fo} and \Cref{item:focut} of \Cref{lemma:monoid-to-expression-generic} are not applicable, and we may use "marked Kleene star" in $\expr$-expressions.

Now consider \Cref{lem:groupwitness}. It implies that any word $w \in \Words a$ can be partitioned into finitely many segments $(u_0, \sigma1u_1, \sigma_2u_2, \dots, \sigma_ku_k)$, each of which begins (except for the first partition) with a letter followed by a word from one of finitely many $\expr$-expressible languages. This allows us to construct a deterministic finite automaton (DFA) where:
\begin{itemize}
    \item States correspond to elements of $\Rclass(a)$.
    \item Transitions are labeled by expressions of the form $\sigma H$, where $H$ is recognizes $\fw b$ for some $b \in \mainZp$, or a language of the form $\eright e$, $\eleft f$, or $\eright e \eleft f$ for some idempotents $e, f \in \mainZp$ satisfying the conditions of \Cref{lem:groupwitness}.
\end{itemize}

\Cref{lem:groupwitness} ensures that $w$ is composed of such segments. We designate the state corresponding to $a$ as final. For any initial state, an expression using concatenation, "Kleene star", and union can be constructed to describe the accepted language. Since these operations preserve "left marked expressions" and the base languages are already left-marked, the resulting expression uses only "marked operations".

The initial state of the DFA depends on an initial prefix $u_0$ of $w$, which can itself be detected using inductively defined expressions.

This concludes the proof.
\end{proof}

Since this completes the proof of one "inductive step@IHJ", the induction hypothesis is now established. Therefore, the language $L$ is "$\expr$-expressible", completing the proof of \Cref{lemma:main}.

\section{Conclusion}
In this paper, we have characterized algebraically and effectively five natural families of recognizable languages of "countable words" that can be defined by "regular expressions". These expressions involve boolean operations, "marked@marked concatenation" and "unrestricted@unrestricted concatenation" concatenations, "marked@marked Kleene star" and "unrestricted Kleene stars@unrestricted Kleene star", and "scatter operations". These results complete our results in \cite{colcombeticalp15} that characterize the same families of languages using logical formalisms. 
We also provide an analysis of the structure of "J-classes".

\bibliographystyle{plain}
\bibliography{linearOrderings}

\end{document}